\documentclass[twoside,11pt]{article} \usepackage{jair, theapa,
rawfonts}

\usepackage{amsfonts} \usepackage{amsmath} \usepackage{amssymb}
\usepackage{mathptmx} \usepackage{nicefrac} \usepackage{pifont}
\usepackage{graphicx} \usepackage{algorithm} \usepackage{algorithmic}
\usepackage{url} \usepackage{pdflscape} \usepackage{color}
\usepackage{multirow} \usepackage{booktabs}

\usepackage{tikz}
\usetikzlibrary{arrows} \usetikzlibrary{decorations.pathreplacing}
\usepackage{pgflibraryshapes} \usetikzlibrary{plotmarks}

\newcommand{\vertices}{\ensuremath{V}}

\newcommand{\p}{\mbox{\rm P}}

\newcommand{\np}{\mbox{\rm NP}}

\newtheorem{theorem}{Theorem}

\newtheorem{proposition}[theorem]{Proposition}

\newcommand{\OMIT}[1]{} 
\newcommand{\Omit}[1]{} 

\newenvironment{proof}{\noindent{\bf
Proof.}\hspace*{1em}}{\literalqed\bigskip} \def\literalqed{{\
\nolinebreak\hfill\mbox{\qedblob\quad}}}

\newcommand\qedblob{\mbox{\ding{113}}}

\newcommand{\shapley}{SV} \newcommand{\short}{\textsc{Short}}
\newcommand{\reroute}{\textsc{Reroute}}

\newcommand{\fracshapley}{\phi^{\textsc{SV}}}
\newcommand{\fracdist}{\phi^{\textsc{Depot}}}
\newcommand{\fracreroute}{\phi^{\textsc{Reroute}}}
\newcommand{\fracshort}{\phi^{\textsc{Short}}}
\newcommand{\fracmoat}{\phi^{\textsc{Moat}}}
\newcommand{\fracchris}{\phi^{\textsc{Chris}}}
\newcommand{\fracblend}{\phi^{\textsc{Blend}}}
\newcommand{\fracproxy}{\phi^{\textsc{Proxy}}}

\ShortHeadings{A Study of
Proxies for Shapley Allocations of Transport Costs} {Aziz, Cahan,
Gretton, Kilby, Mattei, \& Walsh} \firstpageno{1}

\begin{document}

\author{\name Haris Aziz \email haris.aziz@nicta.com.au \\ \addr NICTA
and UNSW,\\ Sydney, Australia \AND \name Casey Cahan \email 
ccah002@aucklanduni.ac.nz \\ \addr University of Auckland,\\ Auckland,
New Zealand \AND \name Charles Gretton \email
charles.gretton@nicta.com.au \\ \addr NICTA and ANU, \\ Canberra,
Australia\\ Griffith University, \\ Gold Coast, Australia \AND \name
Phillip Kilby \email phillip.kilby@nicta.com.au \\ \addr NICTA and ANU,
\\ Canberra, Australia \AND \name Nicholas Mattei \email
nicholas.mattei@nicta.com.au \\ \addr NICTA and UNSW,\\ Sydney,
Australia \AND \name Toby Walsh \email toby.walsh@nicta.com.au \\ \addr
NICTA and UNSW,\\ Sydney, Australia}

\title{A Study of Proxies for Shapley Allocations of Transport Costs}

\maketitle

\begin{abstract}

We propose and evaluate a number of solutions to the problem of
calculating the cost to serve each location in a single-vehicle
transport setting. 
Such cost to serve analysis has application both strategically and
operationally in transportation.
The problem is formally given by the \emph{traveling salesperson game}
(TSG), a cooperative total utility game in  which agents correspond to
locations in a {\rm travelling salesperson problem} (TSP). 
The cost to serve a location is an allocated portion of the cost of an
optimal tour. 
The {\em Shapley value} is one of the most important normative division
schemes in cooperative games, giving a principled and fair allocation
both for the TSG and more generally.
We consider a number of direct and sampling-based procedures for
calculating the Shapley value, and  present the first proof that
approximating the Shapley value of the TSG within a constant factor is
NP-hard.
Treating the Shapley value as an ideal baseline allocation, we then
develop six proxies for that value which are relatively easy to compute.
We perform an experimental evaluation using Synthetic Euclidean games as
well as games derived from real-world tours calculated for fast-moving
consumer goods scenarios.
Our experiments show that several computationally tractable allocation
techniques correspond to good proxies for the Shapley value.

\end{abstract}

\section{Introduction}

We study transport scenarios where deliveries of consumer goods are made
from a depot to locations on a road network. At each location there is a
customer, e.g.~a vending machine or shop, that has requested some goods,
e.g.~milk, bread, or soda. The vendor who plans and implements
deliveries is faced with two vexing problems. First, the familiar
combinatorial problem of routing and scheduling vehicles to deliver
goods cost effectively.
Many varieties of this first problem exist~\cite{golden2008vehicle}, and
for our proposes we shall refer to it as the {\em vehicle routing
problem} (VRP). We begin our investigation supposing that VRP has been
solved heuristically, and therefore after the assignment of locations to
routes has been made.

The second vexing problem is determining how to evaluate the \emph{cost
to serve} each location. Specifically, the vendor must decide how to
apportion the costs of transportation to each location in an equitable
and economically efficient manner. The results of cost to serve analysis
have a variety of important applications. Using the allocation directly
the vendor can of course charge locations their allocated portion of the
transportation costs. More realistically, vendors use the cost
allocations when (re-)negotiating contracts with customers. Supply chain
managers may also reference a cost allocation when deciding whether or
not to continue trade with a particular location. Finally, provided
market conditions are favourable, sales managers can be instructed to
acquire new customers in territories where existing cost allocations are
relatively high in order to share the cost of delivery among more
locations.

Addressing the second vexing problem, this paper stems from our work
with a fast-moving consumer goods company that operates nationally both
in Australia and New Zealand. The company serves nearly 20,000 locations
weekly using a fleet of 600 vehicles. Our industry partner is under
increasing economic pressure to realise productivity improvements
through optimisation of their logistical operations. 
A key aspect of that endeavour is to  understand the contribution of
each location to the overall cost of distribution. 
In this study we focus at the individual route level for a single truck,
where we apportion the costs of the deliveries on that route to the
constituent locations. We formalise this setting as a \emph{traveling
salesperson game} (TSG)~\cite{potters1992},  where the cost to serve all
locations is given by the solution to an underlying \emph{traveling
salesperson problem} (TSP). Formalised as a game, we can use principled
solution concepts from cooperative game theory, notably the Shapley
value~\cite{shapley1953}, in order to allocate costs to locations in a
fair and economically efficient manner.

Calculating the Shapley value of a game is a notoriously hard
problem~\cite{CEW11a}. A direct calculation for a TSG requires the
optimal solution to exponentially many distinct instances of the TSP.
Sampling procedures can  be used for approximating the value, however
these too do not offer a practical solution for larger games. Moreover,
we prove  that there is no polynomial-time $\alpha$-approximation of the
Shapley value for any constant $\alpha \geq 1$ unless $\p =\np$.
To circumscribe these computationally difficulties, this work explores
six proxies\footnote{We use the word \emph{proxy} instead of
\emph{approximation} to ease discussion and, technically, many of these
measures are stand-ins for the Shapley value, not approximations of it.}
for the Shapley value. Our proxies offer tractable alternatives to the
Shapley value, and in some cases appeal to  other allocation concepts
from cooperative  game theory~\cite{PeSu07a,curiel2008}. Two of our
proxies appeal to the well-known \emph{Held-Karp} and
\emph{Christofides} TSP heuristics, respectively.

We report a detailed experimental comparison of proxies using a large
corpus of Synthetic  Euclidean games, and problems derived from
real-world tours calculated for fast-moving  consumer goods businesses
in the cities of Auckland (New Zealand), Canberra, and Sydney
(Australia). We highlight three computationally tractable proxies that
give good approximations of the Shapley value in practice. Our
evaluation also considers the ranking of locations---least to most
costly---induced by the Shapley and  proxy values. 
Ranking is relevant when, for example, we are just interested 
in identifying the most costly locations to serve. We again find that
three of our proxies provide good ranking accuracy taking the rank
induced by the Shapley value as the target.

 \section{Preliminaries}\label{sec:preliminaries}

We use the framework of cooperative game theory to gain a deeper
understanding of our delivery and cost allocation
problems~\cite{PeSu07a,CEW11a}. In cooperative game theory, a game is a
pair $(N,c)$. $N$ is the set of agents and the second term
\mbox{$c:2^N\rightarrow\mathbb{R}$}  is the \emph{characteristic
function}. Taking $S\subseteq N$, $c(S)$ is the cost of subset $S$. A
\emph{cost allocation} is a vector $x=(x_0,\ldots, x_n)$ denoting that
cost $x_i$ is allocated to agent $i\in N$. We restrict our attention to
economically \emph{efficient} cost allocations, which are allocations
satisfying $\sum_{i\in N}x_i=c(N)$.

For any cooperative game $(N,c)$, a \emph{solution concept} $\phi$
assigns to each agent $i\in N$ the cost $\phi_i(N,c)$. There may be more
than one allocation satisfying the properties of a particular solution
concept, thus $\phi$  is not necessarily single-valued, and might give a
set of cost allocations \cite{PeSu07a}. A minimal requirement of a
solution concept is \emph{anonymity}, meaning that the cost allocation
must not depend on the identities of locations. Prominent solution
concepts include the \emph{core}, \emph{least core}, and the
\emph{Shapley value}. For $\epsilon\geq 0$, we say that cost allocation
$\phi$ is in the (multiplicative) \emph{$\epsilon$-core} if $\sum_{i\in
S}\phi_i\leq (1 + \epsilon)c(S)$ for all $S\subseteq N$
\cite{faigle1993some}. The $0$-core is referred to simply as the
\emph{core}. Both the core and $\epsilon$-core can be empty. The
$\epsilon$-core which is non-empty for the smallest possible $\epsilon$
is called the least core. This particular $\epsilon$ is referred to as
the \emph{least core value}.\footnote{The $0$-core of the transport game
we focus on in this work can be empty.  However, if the game is convex,
the Shapley value lies in the core \cite{Tamir198931}.}

Our work focuses on the single-valued solution concept called the
\emph{Shapley value}~\cite{shapley1953}. Writing $\shapley_i(N,c)$ for
the Shapley value of agent $i$, formally we have:

\begin{equation} \label{eq:shapleyexactmargins} \shapley_i(N,c)  = 
\sum_{S \subset N \setminus \{i\}} \frac{|S|!(|N|-|S|-1)!}{|N|!} (c(S
\cup \{i\})-c(S)). \end{equation}

\noindent In other words, the Shapley value divides costs based on the
marginal cost contributions of agents

In the \emph{traveling salesperson problem (TSP)} a salesperson must
visit a set of locations $N=\{1,\ldots, n\}\cup\{0\}$ starting and
ending at a special \emph{depot} location $0$. For $i,j\in N\cup\{0\}$
$i\neq j$,  $d_{ij}$ is the strictly positive distance traversed when
traveling from location $i$ to $j$. Here, $d_{ij}=\infty$ if traveling
directly from $i$ to $j$ is impossible. Taking distinct $i,j,k\in
N\cup\{0\}$, the problem is \emph{symmetric} if and only if $d_{ij} =
d_{ji}$ for all $i,j\in N\cup\{0\}$. It satisfies the \emph{triangle
inequality} if and only if $d_{ij} +d_{jk} \geq d_{ik}$~\cite{GaJo79a}.

A TSP is \emph{Euclidean} when each location is given by coordinates in
a (two dimensional) Euclidean space; therefore $d_{ij}$ is the Euclidean
distance between $i$ and $j$. A Euclidean TSP is both symmetric and
satisfies the triangle inequality.

A \emph{tour} is given by a finite sequence of locations that starts and
ends at the depot $0$. The \emph{length} of a tour is the sum of
distances between consecutive locations. For example, the length of
$[0,1,2,0]$ is $d_{01}+d_{12}+d_{20}$. An  optimal solution to a TSP is
a minimum length tour that visits every location. It is \np-hard to 
find an optimal tour, and generally there is no $\alpha$-approximation
for any $\alpha$ unless $\p=\np$.  An $\alpha$-approximation for a given optimisation problem 
is an algorithm that runs on an instance $x$ and returns
a feasible solution $F(x)$ which has cost $c(F(x))$ related
to the optimal solution $OPT(x)$ by the following relation \cite{papado:b:compcomplexity}:
$$\frac{|c(F(x)) - c(OPT(x))|}{\max\{c(OPT(x)), c(F(x))\}} \leq \alpha.$$
Informally, $\alpha$ is a bound on the relative error of an approximation
function. When $\forall i,j$ $d_{ij}$ are finite, the triangle
inequality  and symmetry hold, then polynomial-time approximations
exist~\cite{HeKa62a,Chri1976}.

Given a TSP, the corresponding \emph{traveling salesperson game} (TSG)
is a pair $(N,c)$. $N$ is the set of agents which corresponds to the set
of locations.\footnote{From here on we focus on a restriction of general
games to delivery games (TSGs) and therefore we use {\em location}
instead of {\em agent} for ease of exposition.} The second term
\mbox{$c:2^N\rightarrow\mathbb{R}$}  is the characteristic function.
Taking  $S\subseteq N$, $c(S)$ is the length of the shortest tour of all
the  locations in $S$. A \emph{cost allocation} is a vector
$x=(x_1,\ldots, x_n)$ denoting that cost $x_i$ is allocated to location
$i\in N$. For the special depot location, we shall always take $x_0 =0$
\cite{potters1992}

\section{Some Properties of the Shapley Value}

The Shapley value has many attractive properties when used as a cost
allocation scheme by a vendor. For example, whereas the $0$-core can be
empty, and therefore not yield any allocation at all~\cite{Tamir198931},
the Shapley value always exists in the TSG setting. The Shapley value is
also, for general games, the unique assignment of costs  that satisfies
three important properties: (1) \emph{anonymity}, the cost allocated to
a particular location is dependent only on the impact it has to the
total cost; (2) \emph{efficiency}, the entire cost of serving all $N$
locations is allocated; and (3) \emph{strong monotonicity}. 
The latter states that if the total cost of a coalition is reduced, 
then the allocation to all locations participating in that coalition 
is either reduced or not increased~\cite{young1985monotonic}. 
Formally, the marginal contribution
from player $i$ to the total cost of coalition $S$ is: \[ c^i(S) =
\begin{cases} c(S) - c(S \setminus \{i\}) & \text{if } i \in S \\ c(S
\cup \{i\}) - c(S) & \text{if } i \notin S. \end{cases} \]

Strong monotonicity can be stated as: $\forall S: c^i(S) \geq {c'}^i(S)
\Longrightarrow \phi_i(N,c) \geq \phi_i(N,c')$. Due to these and other
derivative axiomatic properties, the Shapley value has been termed ``the
most important normative payoff division scheme'' in cooperative game
theory~\cite{Wint02a}.

Another important property of the Shapley value is that it would
allocate any fixed costs incurred when serving a location to that
location alone. If we treat a variant of the TSG where some locations
have an associated fixed cost in addition to their transportation
costs--- e.g.~parking and loading fees ---then the Shapley value will
allocate those fixed costs to the associated locations. Formally, given
a fixed cost $f(i)$ of serving location $i$, $f(i)$ does not need to be
removed before computing the Shapley value, as follows. Suppose $c$ is
the characteristic function of the TSG defined above, and $c'$ satisfies
the identity $c'(S)=c(S) + \sum_{i\in S}f(i)$.

\begin{proposition} $\shapley_i(N,c')=\shapley_i(N,c)+f(i).$
\end{proposition}

\begin{proof} \begin{align*} \shapley_i(N,c)&=\sum_{S\subseteq
N\setminus \{i\}}(|S|!)(|N|-|S|-1)!(c(S\cup \{i\})-c(S))/{|N|!}\\
&=\sum_{S\subseteq N\setminus \{i\}}(|S|!)(|N|-|S|-1)!((c'(S\cup
\{i\})+f(i))-c(S))/{|N|!}\\ &=\sum_{S\subseteq N\setminus
\{i\}}(|S|!)(|N|-|S|-1)!(c'(S\cup \{i\})-c'(S))/{|N|!} +
\sum_{S\subseteq N\setminus \{i\}}(|S|!)(|N|-|S|-1)!(f(i))/{|N|!}\\
&=\shapley_i(N,c')+ (\sum_{S\subseteq N\setminus
\{i\}}(|S|!)(|N|-|S|-1)!/{|N|!})(f(i))\\ &=\shapley_i(N,c')+
(|N|!/|N|!)(f(i))\\ &=\shapley_i(N,c')+ (f(i)) \end{align*}

\end{proof}

We also have that by charging locations according to the Shapley value,
we can expect to incentivize them to \emph{recruit} new customers in
their vicinity. Locations recruiting for a vendor can reasonably expect
to lower the transportation costs they are allocated. In detail,
consider a vendor trading with locations $N=\{1..|N|\}$. From the
vendors perspective, adding a new location, $|N|+1$, to an existing
delivery route is clearly a good idea if the revenue generated by
delivering to that location is greater than the marginal cost $c(N\cup
\{|N|+1\})-c(N)$ of the new delivery. Because existing locations in the
vicinity of $|N|+1$ are already paying for deliveries, charging at the
threshold $c(N\cup \{|N|+1\})-c(N)$ however will typically be unfair. In
that case existing customers would likely be subsidizing new customers,
and therefore disincentivize  to find new business for the vendor. The
Shapley value mitigates this, and can be expected to provide recruitment
incentives. 
Making this discussion more concrete, suppose the game is a Euclidean
scenario with $N=\{x\}$ a single agent at distance $100$ from the depot
and the new agent $y$ is at distance $5$ from $x$. The transportation
cost of serving $\{x,y\}$ can be as high as $210$. Clearly, charging the
new agent at most $c(\{x,y\})-c(\{x\}) = 10$ while $x$ continues to pay
around $200$ is unfair. On the other hand, if the vendor allocates costs
according to the Shapley value, the existing customer's costs
\emph{decrease} when the new agent joins.

Related to the above discussion, if the characteristic function is
concave then  the Shapley value lies in the non-empty $0$-core.
Formally, concavity is satisfied if for all $S\subseteq N\setminus
\{i\}: c(S\cup\{i\}\cup \{|N|+1\})-c(S\cup \{|N|+1\}) <
c(S\cup\{i\})-c(S)$. Charging customers according to core values
actually guarantees that they are incentivized to recruit.
Specifically, for all $i \in N: \shapley_i(N\cup \{|N|+1\},c)<
\shapley_i(N,c)$. In other words, the Shapley allocation of costs to
existing locations decreases when a new customer $|N|+1$ is added.
Unfortunately general TSGs do not necessarily have  concave
characteristic functions. However, concavity in expectation is all that
is required for existing locations to realise savings. In practice there
are synergies, and incentives for further recruitment on routes where we
charge according to the Shapley value. In our empirical data, even when
the game is not concave we frequently observe such incentives given a
Shapley allocation. And compared to charging customers according to
their marginal contribution to costs, we do not explicitly
disincentivize recruitment. Summarizing, if an agent knows that all
locations are charged according to the Shapley value, they can typically
expect incentives to recruit new locations in their vicinity.

\section{Computing the Shapley Value}

Our focus now shifts to calculation of the Shapley value. 
Considering games in general, it should be noted that a direct
evaluation of Equation~\ref{eq:shapleyexactmargins} requires we sum over
exponentially many quantities.
Such a direct approach to the calculation of the Shapley value is
therefore not practical for any game of a reasonable size.
Indeed, starting from the earliest literature~\cite{MaSh62a}, authors
motivate auxiliary restrictions and constraints, for example on the size
and importance of coalitions,  in order to describe games where the
Shapley value can be calculated.
More recent literature proposes a variety of approaches to directly
calculate the Shapley value for certain games~\cite{CoSa06,IeSh05},
however efficient calculation of the value for TSGs has remained elusive.
We require an accurate baseline in order to experimentally evaluate the 
proxies we later develop for the Shapley value of the TSG. To that
purpose we investigate exact and general sampling-based approximations
of the Shapley value. We treat our transport setting specifically,
describing a novel procedure for an exact evaluation of the Shapley
value of a TSG by following Bellman's dynamic programming solution to
the underlying TSP. 
We also discuss how in general the Shapley value can be evaluated
approximated using a sampling procedure. We pursue that sampling
approach in TSGs, considering two distinct characterisations of the
Shapley value which are amenable to sampling-based evaluation. We
performed a detailed empirical study of sampling-based evaluations using
Synthetic TSGs instances where the underlying TSP model is Euclidean. 
In closing we give a hardness proof relating to the computation of the
Shapley value of TSGs, showing that  approximation of the Shapley value
in that game is intractable.

\subsection{Dynamic Programming}

We found that the steps performed by a {\em dynamic programming} (DP)
solution to the underlying TSP expose the margins---i.e. terms of the
form $c(S \cup \{i\})-c(S)$---that are summed over in a direct
evaluation of Equation~\ref{eq:shapleyexactmargins}. The Shapley value
of a TSG can therefore be computed more-or-less as a side effect while a
DP procedure computes the optimal solution to the underlying TSP.

These ideas can be made concrete by following the  procedure outlined
by~\citeauthor{bellman1962dynamic}~\citeyear{bellman1962dynamic}.
The equations at the heart of that TSP solution procedure  recursively
define a cost function, $c(S,j)$, which is the shortest path through all
locations in $S$ starting at the depot $0$ and ending at
$j$.\footnote{Our notations depart slightly from Bellman's seminal work.
Whereas we take $c(S,j)$ to be the cost of each optimal tour-prefix path
(i.e. starting at the depot $0$ and ending at $j$), Bellman originally
took $c(S,j)$ to be the cost of optimal tour-suffix paths starting from
$j$, traversing the locations in $S$ and ending at the depot $0$.}
\begin{eqnarray*} c(\{j\},j) & = & d_{0j} \\ c(S,j) & = & \min_{k \in S,
k \neq j} (c(S \setminus \{j\},k) + d_{kj}) \end{eqnarray*}

\noindent Following the above recursive definition, a DP process
iteratively tabulates $c(S,j)$ for successively larger coalitions $S$.
At iteration $n$ that procedure shall tabulate all quantities $c(S,j)$
taking $|S|=n$. By computing the values $c(S, 0)$ for $|S|<|N|$, we have
access to the characteristic function evaluation $c(S)$ of subtours of
locations in $S$, as follows:

\begin{eqnarray*} c(S)  = c(S, 0) =  \min_{j \in S} (c(S,j)+d_{j0}).
\end{eqnarray*}

\noindent Therefore, one can incrementally evaluate the sum in
Equation~\ref{eq:shapleyexactmargins} for a TSG, while calculating
optimal subtours for progressively larger coalitions withing a classical
DP procedure. 
Intuitively, as we compute a tour using Bellman's algorithm, by
additionally evaluating $c(S, 0)$ for each encountered subset $S$ we
obtain all quantities required to calculate the marginal costs of
locations. 
It is worth noting that the dynamic programming approach does not
address the exponential number of subsets we need to sum over in the
evaluation of Equation~\ref{eq:shapleyexactmargins}.
We have therefore highlighted a concrete relationship between a
classical procedure for the TSP and the Shapley value of the
corresponding TSG. However, this observation does not yield a practical
algorithm for games with many more than a dozen locations.

\subsection{Sampling-Based Evaluation}\label{sec:sample}

Using either the DP solution, or indeed the state-of-the-art TSP solver
Concorde~\cite{Concorde} in a direct calculation of the Shapley value,
we find it impractical to compute the  exact Shapley value for instances
of the TSG  larger than about 15 locations.
A direct method requires an exponential number of characteristic
function computations, each requiring we solved an NP-hard problem.
To obtain an accurate baseline for reasonably sized games our
investigation now turns to sampling procedures.
Indeed, because the Shapley value is a population average it is
reasonable to estimate the value using a sampling procedure.

The first use of sampling to approximate the Shapley value of games was
proposed and studied by Mann and Shapley \cite{MaSh60a}. Perhaps the
most elegant and general  method proposed by Mann and Shapley is called
{\em Type-0} sampling. This method repeatedly draws uniformly at random
a permutation of the agents. The marginal cost of each agent $i$ is then
calculated, by taking the difference in the cost of serving agents up to
and including $i$ in the permutation, and the cost of serving the agents
proceeding $i$. By repeatedly sampling permutations and the marginal
costs of including each agent $i$ in this way, overtime we arrive at an
unbiased estimate of the Shapley value. 
Further elaboration of this procedure for the TSG is given below. Type-0
sampling has appeared over the years in various guises, and is reported
under a variety  of different names in the literature on approximating
{\em power indices}---of which the Shapley value is but one---in
coalitional games.
%
A recent variant of Type-0 sampling appears as the  \emph{ApproShapley}
algorithm by Castro et al.~ in a paper which proves asymptotic bounds on
the sampling  error of that method~\cite{Castro:2009}.
\emph{ApproShapley} shall be the focus of our sampling work, however
prior to giving its details, it is worth briefly reviewing other classes
of game where sampling-based evaluations have been explored. 
Bachrach et al.~have previously examined Type-0 sampling in {\em simple
games}---i.e. cost of a coalition is either $0$ or $1$---deriving bounds
that are {\em probably approximately correct}. In other words, the
actual Shapley value lies within a given error range with high
probability~\cite{BMR+10}.
Continuing in this line of work, Maleki et al.~show that if the range or
 variance of the marginal contribution of the players is known ahead of
time, then more focused (termed \emph{stratified}) sampling techniques
may be able to decrease the number of samples required to achieve a
given error bound \cite{MTH+13}. Other methods of approximating the
Shapley value, specifically for weighted voting games, have appeared in
the literature including those based on multi-linear extensions
\cite{Leech03,Owen72} and focused random sampling \cite{FWJ08a,FWJ07a}

To calculate the Shapley value of a TSG via sampling we employ the
Type-0 method suggested by Mann and Shapley \cite{MaSh60a}, called
\emph{ApproShapley} by~\citeauthor{Castro:2009}. 
The pseudocode is given in Algorithm~\ref{algo:Appro}.
Writing $\pi(N)$ for the  set of $|N|!$ permutation orders of locations
$N$, taking $\Pi\in\pi(N)$ we write $\Pi_i$ for the subset of $N$ which
precede location $i$ in $\Pi$. An alternative formulation of the Shapley
value can be characterised in terms of $\pi(N)$, by noting that value
equates with marginal cost of each location when we construct coalitions
in all possible ways, as follows.

\begin{equation} \shapley_i(N,c)  = \frac{1}{|N|!} \sum_{\Pi \in \pi(N)}
(c(\Pi_i \cup \{i\})-c(\Pi_i)) \end{equation}

For each sampled permutation, \emph{ApproShapley} evaluates the
characteristic  function for each $i \leq |N|$ computing the length of
an optimal  tour for the set of locations in the $i$-sized prefix.  By
construction, the cost allocation produced by  \emph{ApproShapley} is
economically efficient. As a small but important optimization, in our
work we cache the result of each evaluation of the characteristic
function to avoid solving the same TSP twice.

\begin{algorithm} \caption{ApproShapley} \label{algo:Appro} \small
\renewcommand{\algorithmicrequire}{\textbf{Input}:}
\renewcommand{\algorithmicensure}{\textbf{Output}:}

\begin{algorithmic} \REQUIRE $N = \{1, \ldots, n\}$ locations with cost
$c(S)$ to serve a subset $S \subseteq N$ and $m$ number of iterations.
\ENSURE $\shapley_i$ for all $i \in |N|$ \end{algorithmic}

\algsetup{linenodelimiter=\,} \begin{algorithmic}[1] \small
\STATE{$\shapley \leftarrow []$} \FOR{$i \leftarrow 1$ \TO $|N|$}
\STATE{$\shapley_i \leftarrow 0$} \ENDFOR

\STATE{$SampleNumber \leftarrow 1$}

\FOR{$SampleNumber \leftarrow 1$ \TO $m$} \STATE{Randomly select a
permutation of the locations $Perm$ from $\pi(N)$} \STATE{$S \leftarrow
\emptyset$} \FOR{$i \leftarrow 1$ \TO $|N|$} \STATE{$S \leftarrow S \cup
\{Perm_i\}$} \STATE{$\shapley_{Perm_i} \leftarrow \shapley_{Perm_i} +
(c(S) - c(S \setminus \{Perm_i\} ) )$} \ENDFOR \ENDFOR

\STATE{$TotalValue \leftarrow \sum_{i \in N} \shapley_i$} \FOR{$i
\leftarrow 1$ \TO $|N|$} \STATE{$\shapley_i \leftarrow \shapley_i *
(\nicefrac{c(N)}{TotalValue}) $} \ENDFOR \RETURN $\shapley$
\end{algorithmic} \end{algorithm}

In our work, we also considered an alternative sampling method, which
samples not over permutations, but rather over subsets of locations as
implied by the formulation in Equation~\ref{eq:shapleyexactmargins} of
Section~\ref{sec:preliminaries}. There are fewer subsets than there are
permutations, a fact which we supposed could be an advantage in a
sampling-based evaluation of the Shapley value.
We name this method \emph{SubsetShapley}, which by construction also
produces an economically efficient allocation.
Later we empirically find the \emph{SubsetShapley} performs worse than
\emph{ApproShapley}, however because this approach does not yet appear
in the literature we believe it worthy of discussion. 
\emph{SubsetShapley} follows Algorithm~\ref{algo:Appro} except for Lines
7--10. In this case at every iteration of the loop at Line 6 we draw a
set $S_i\subseteq N\setminus \{i\}$ uniformly at random for each
location $i$. For each $i$, the update to $\shapley_i$ is then the
weighted marginal contribution, formally $\shapley_{i} \leftarrow
\shapley_{i} + |S|!(n-|S|-1)!(c(S\cup i) - c(S))$. The coefficient
$|S|!(n-|S|-1)!$ ensures that for each subset $S_i$ of locations
sampled, we account for the number of  permutations where locations
$S_i$  are ordered before location $i$.

In order to test which sampling method performs best, we ran convergence
tests on 50 random  instances for up to  5000 iterations. The instances
were Euclidean TSGs on a 1,000x1,000 dimensional square with 10
locations, each at coordinates given by a pair of 32-bit floating point
numbers. For each instance we calculated the exact Shapley value of
every location, so that we could compare the sampled allocations with
their exact counterparts. Figure~\ref{fig:sample-compare} graphically
summarises the results from this experimentation.

We find the \emph{ApproShapley} method of sampling over permutations
provides a faster convergence. After as few as $100$ iterations
\emph{ApproShapley} achieves an average error of $\approx 10\%$ per
location with a maximum error of $\approx 20\%$. Additionally, the
stability of the updates for \emph{ApproShapley}, as measured by the
percentage of the allocation that is re-assigned per iteration, is
already quite good  after 40 iterations. \emph{ApproShapley} quickly
converges to a correct and stable answer which it continues to refine as
more samples are taken. In practice, \emph{ApproShapley} achieves a
lower error, earlier, and continues to converge on an error of 0.0
faster than \emph{SubsetShapley}.

\begin{landscape} \begin{figure}[h!]\label{fig:sample-compare}
\centering 
\includegraphics[scale=0.60]{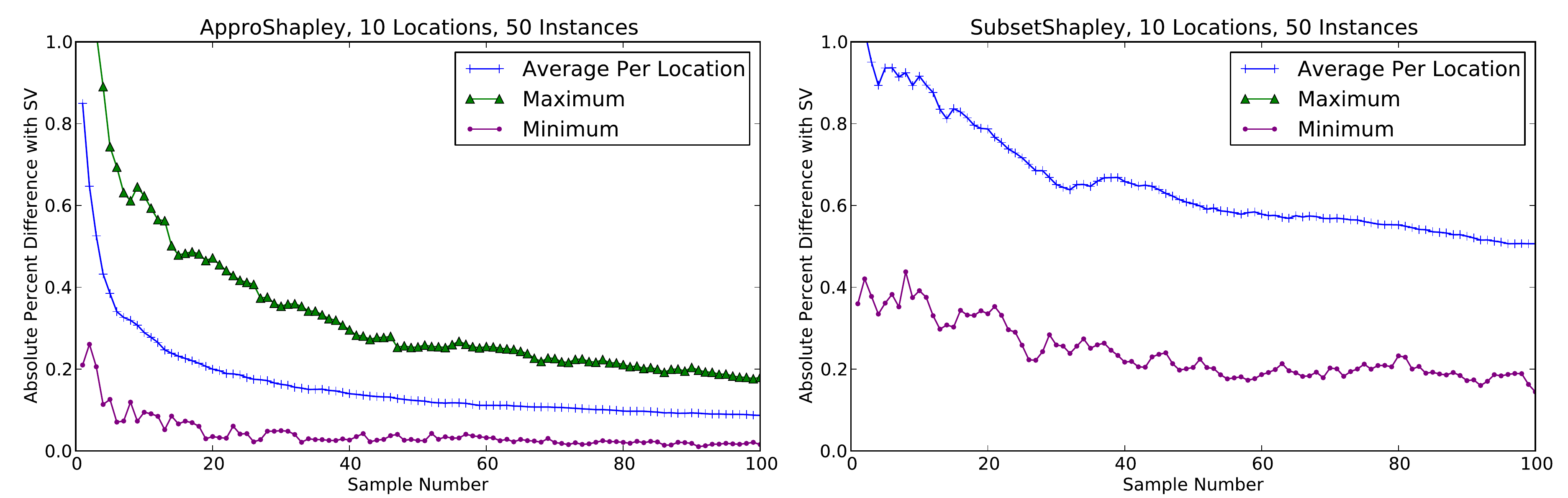}
\includegraphics[scale=0.60]{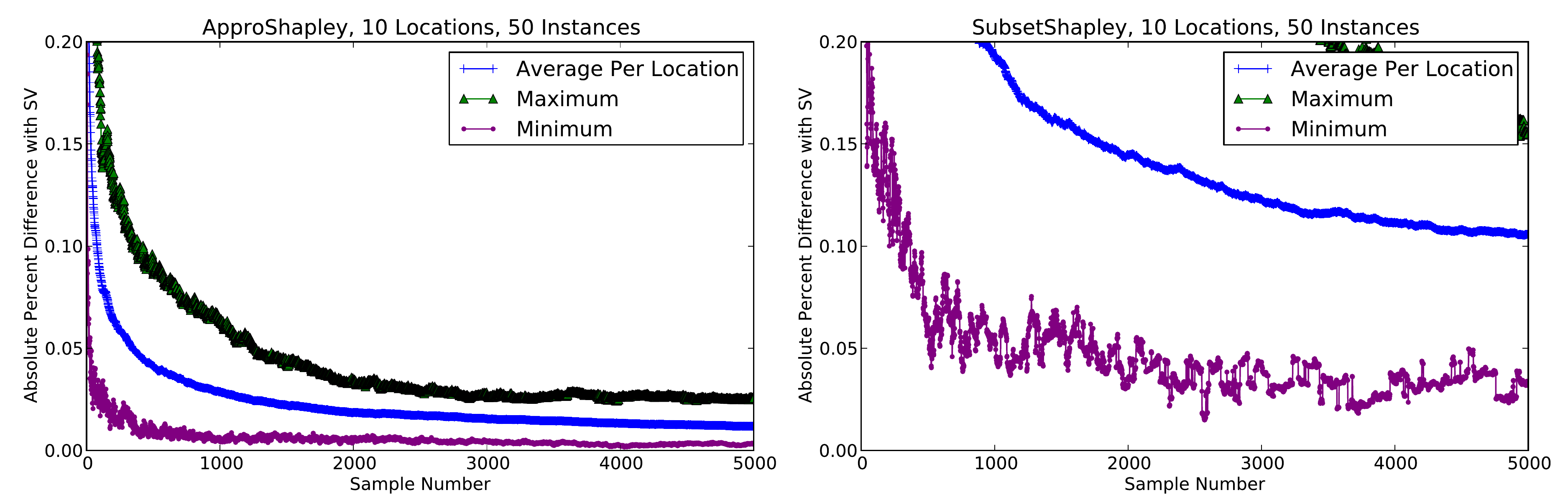}
\caption{Comparison of the performance of ApproShapley (left) and
SubsetShapley(right) for 100 iterations (top) and 5000 iterations
(bottom) for TSGs with 10 locations. The graphs show minimum and maximum (outside
the graph range for SubsetShapley)
error in the Shapley value for a single location averaged over 50
instances.  The error is computed as a percentage difference between the
actual Shapley value and the one computed by sampling.  Additionally,
the average percent error for all locations per iteration is shown.}
\end{figure} \end{landscape}

\subsection{Theoretical Hardness}

We now consider, for the most general setting of the TSG, the difficulty
of calculating the Shapley value. Below we prove that the Shapley value
of a location in the TSG cannot be approximated  within a constant
factor in polynomial-time unless $\p =\np$.

\begin{theorem} There is no polynomial-time $\alpha$-approximation of
the Shapley value of the location in a TSG for constant $\alpha \geq 1$
unless $\p =\np$. \end{theorem} \begin{proof} Let $G(N,E)$ be a graph
with nodes $N$ and edges $E$. If an $\alpha$-approximation exists we can
use it to solve the \np-complete Hamiltonian cycle problem on G. First,
from $G$ construct a complete weighted and undirected graph $G'(N,E')$,
where $(i,j)$ has weight $1$ if $(i,j)$ is in the transitive closure of
$E$, and otherwise has weight $n!\alpha$. If there is a Hamiltonian
cycle in $G$ then the Shapley value of any $i\in N$ in the TSG posed by
$G'$ is at most $1$. Suppose there is no Hamiltonian cycle in $G$.  We
show there exists a permutation $\pi$ of $N$   that induces a large
Shapley value for any node $j$ as follows: repeatedly add a node from
$N\backslash j$ to $\pi$ so that there remains a Hamiltonian cycle
amongst elements in $\pi$; when there is no such node then add $j$. The
marginal cost of adding $j$ to $\pi$ is at least $n!\alpha$. The Shapley
value of $j$ is the average cost of adding it to a coalition $S\subseteq
N \setminus j$, therefore its Shapley value is at least $\alpha$. Even
though edge weights in $G'$ are large, 
we can represent $G'$ compactly in $O(\log(n)+n^2 \log(\alpha))$ space.
An $\alpha$-approximation on $G'$ for $j$ therefore decides the
existence of the Hamiltonian cycle in $G$. \end{proof}

%
%
%
\section{Proxies for the Shapley Value}\label{sec:proxies}

The use of \emph{ApproShapley} requires that we solve an \np-hard
problem each time we evaluate the characteristic function. This is
feasible for small TSG instances with less than a dozen locations,
however it does create an unacceptable computational burden in larger,
realistically sized games. We now describe a variety of proxies for the
Shapley value that require much less computation in practice.

For the purposes of the discussion below we assume that an optimal tour
for the underlying TSP is given. Not all our proxies yield economically
efficient allocations of the cost of the optimal tour. For that reason,
we define proxies in terms of the induced {\em fractional} allocation of
the cost of the optimal tour. Later, we shall compare these fractional
allocations to that induced by computing the fractional Shapley value,
formally  $\fracshapley_i = \nicefrac{\shapley_i}{\sum_{j \in n}
\shapley_j}$.  This formulation based on fractional allocations allows
us to compare the cost allocations from all the proxies on equal
footing, in a way that would be used in operational contexts such as
transport settings. This formulation also enables us to
efficiently---i.e. in the game theoretic sense---allocate the cost of
the optimal route only having to solve the \np-hard TSP once.

\subsection{Depot Distance ($\fracdist$)}

The distance from the depot --- i.e. $d_{i0}$ for location $i$ --- is
our most straightforward proxy. We allocate cost to location $i$
proportional to $d_{i0}$. The fraction allocation to location $i$ is
$$\fracdist_i = \frac{d_{i0}}{\sum_{i=1}^n d_{i0}}.$$ For this proxy, a
location that is twice as distant from the depot as another has to pay
twice the cost.

\subsection{Shortcut Distance ($\fracshort$)}

Another proxy that is straightforward to calculate and which has been
used in commercial routing software is  the \emph{shortcut distance}.
This is the marginal cost savings of skipping a location when traversing
a given optimal tour. With no loss of generality, suppose the optimal
tour visits the locations according to the sequence $[0,1,2,\dots]$.
Formally, $\short_i = d_{i-1,i}+ d_{i,i+1} - d_{i-1,i+1}$, where
locations $0$ and $n+1$ are the depot, and $d_{ij}$ is the cost of
travel from location $i$ to $j$. The fractional allocation given by the
shortcut distance is then $$\fracshort_i = \frac{\short_i}{\sum_{j\in
N~} \short_j.}$$

\subsection{Re-routed Margin ($\fracreroute$)}

For a location $i\in N$, $\reroute_i$ is defined as $c(N) -
c(N\backslash i))$. The allocation to a player can be computed with at
most two calls to an optimal TSP solver. The fractional allocation  is
$$\fracreroute_i = \frac{(c(N) - c(N\backslash i))}{\sum_{j=N} (c(N) -
c(N\backslash j))}.$$

\subsection{Christofides Approximation ($\fracchris$)}

A more sophisticated proxy is obtained if we use a heuristic when
performing characteristic function evaluations in \emph{ApproShapley},
rather than solving the individual induced TSPs optimally. For this
proxy we use sampling to estimate the Shapley value and we use an
approximation algorithm to estimate the underlying TSP cost. To
approximate the underlying TSP characteristic function, the Christofides
heuristic~\cite{Chri1976}, an $O(N^3)$ time procedure is used. To obtain
a fractional quantity $\fracchris_i$, we divide the allocation to
location $i$ by the sum total of allocated costs. Assuming a symmetric
distance matrix satisfying the triangle inequality, the Christofides
heuristic is guaranteed to yield a tour that is within $3/2$ the length
of the optimal tour.

We briefly describe how the heuristic operates. The TSP instance is
represented as complete undirected graph $G=(\vertices, E)$, with one
vertex in $\vertices$ for each location, and an edge $E$ between every 
distinct pair of vertices. For $i,j\in\vertices$ the edge $(i,j)\in E$
has weight  $d_{ij}$. A tour is then obtained as follows: (1) compute
the minimum spanning tree (MST) for $G$,  (2) find the minimum weight
perfect matching for the complete graph over vertices with odd degree in
that MST  (typically performed using the {\em Hungarian} algorithm), 
(3) calculate an Eulerian tour for the Euler multigraph obtained by
adding edges from Step (2)  to the MST from Step (1), and  (4) obtain a
final tour for the TSP by removing duplicate locations from the Eulerian
tour.

\subsection{Nested Moat-Packing ($\fracmoat$)}

\newcommand{\moat}{\ensuremath{P}}

\begin{figure} \centering
\includegraphics{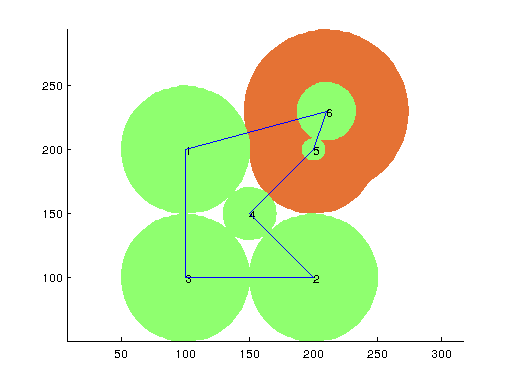}

\caption{Depicts an optimal nested moat-packing, and the optimal tour
(blue-line) for a TSP scenario with 6 locations. The locations are
indicated by their digit labels (e.g. ``1'', ``2'', $\ldots$), and occur
at the center of the green moats, which appear as disks. Each green disk
depicts a distinct moat associated with one location. The orange region
is the moat associated with the set of locations $\{5,6\}$. Following
the nesting-scheme, the orange moat surrounds an internal region
comprising the green moats around locations $5$ and $6$, respectively.
There are $7$ moats in total, and the optimal tour in this case
traverses the width of each moat exactly twice.}
\label{fig:examplesofmoats} \end{figure}

A cost allocation method based on a nested moat-packing was first
introduced by \citeauthor{faigle1998}~\citeyear{faigle1998}.  This
allocation is obtained by apportioning a grand-coalition cost equal to
the  value of the Held-Karp~\cite{HeKa62a} relaxation of the underlying
TSP, multiplied by a constant factor. It is worth briefly considering
some details of the background of this approach, and the geometric
intuitions.

The value of the Held-Karp relaxation of a TSP instance corresponds to a
fairly tight lower bound on the length of an optimal tour.  That value
is a lower bound for the TSP in the usual sense -- i.e. it is less than
or equal to the length of an optimal tour. By multiplying this value by
a small factor, specifically $1.5$, one can obtain an upper bound. The
approach discussed here allocates costs to locations so that the sum of
allocated costs equates with that upper bound.
The allocation gives $\frac{1}{2}$-core values provided the distance
matrix is symmetric and satisfies the triangle inequality. 
Formally, where $\moat_i$ is the cost of location $i$, the moat-packing
solution satisfies $\sum_{i\in N} \moat_i \geq c(N)$ and $\forall S
\subseteq N: \sum_{i\in S} \moat_i \leq (1 + \epsilon) c(S) $. It is is
known that $\epsilon \leq \frac{1}{2}$ and conjectured that $\epsilon
\leq \frac{1}{3}$. In this work we induce a fractional allocation,
written $\fracmoat_i$, by normalizing as we have done for other proxies. 
The solution to the Held-Karp relaxation, and therefore the moat-packing
allocation, has an interesting geometric interpretation which we briefly
discuss (see
\citeauthor{combinatorialoptimisation}~\citeyear{combinatorialoptimisation} 
for a longer exposition). A graphic providing
concrete examples of the required concepts is in
Figure~\ref{fig:examplesofmoats}. Our discussion distinguishes the
concept of a {\em point}, a geometric point given by its coordinates,
and a {\em location}, which is a point that corresponds to a customer in
the underlying TSP. The proposed allocation is calculated by surrounding
locations using a set of geometrically nested reagions called {\em
moats}. For example, in Figure~\ref{fig:examplesofmoats} we have 6
locations, each of which has its own green moat. In our graphic the
locations $5$ and $6$ have their own green moats that describe an
interior region which is then surrounded by an outer orange moat. A moat
is defined by an {\em interior region}, containing the set of locations
we are surrounding with the moat, and a {\em surrounding contour}. The
interior region occurs in the space encapsulated by the moat. That moat
is  the region between the boundary of the interior region and the
surrounding contour. The smallest distance between a point in the
interior region and one on the surrounding contour is greater than or
equal to zero. The minimum such distance gives the width of the moat.
Finally, there can be no locations in a moat. As is usual in our
setting, we need only consider moats comprising the set of  points whose
minimal straight-line distance to a point in the interior region is less
than or equal to the moat width. 
For example, taking the interior region for a single location to consist
only of its single point, the moat is the region between that point and
a circle contour of constant radius. The radius of that circle is the
width of the moat. Concretely, the green disks in
Figure~\ref{fig:examplesofmoats} depict circular moats around individual
locations.
To obtain a cost allocation, moats are arranged so that for a vehicle to
visit the set of locations in the underlying TSP, that vehicle must
traverse the width of each moat at least twice. Choosing moats in order
to to maximise the sum of their widths, the distance traversing all
chosen moats twice corresponds to the value of the Held-Karp lower
bound. One  obtains an $\epsilon$-core value by allocating each moat
width twice to locations outside the moat, and then scaling those
allocations, here by the constant factor $1.5$, to ensure the sum of
allocated costs exceeds the length of an optimal tour.

A compilation of the above ideas is expressed mathematically below in
the constraints and optimisation criterion in
Equation~\ref{eq:moatpacking}. Formally, the moat width,  $w_S$, for a
set of locations  $S\subseteq N$ is calculated by solving the LP in
Equation~\ref{eq:moatpacking}. Below, taking the TSP as given by a
weighted fully connected graph, we use the notation $\delta(S)$ for the
set of edges joining locations in $S$ to locations in $N\setminus S$.

\begin{equation} \label{eq:moatpacking} \begin{array}{c} \max \big( 2
\sum_{S\subseteq N \;\; S \not\equiv \emptyset} w_S\big) \\ s.t.\\
\begin{array}{ll} w_S \geq 0 & \forall S \subseteq N \;\; S \not\equiv
\emptyset\\ \sum_{ij\in\delta(S)} w_S \leq d_{ij} & \forall i,j \in N
\end{array} \end{array} \end{equation}

\noindent The dual of this LP corresponds to the well-known Held-Karp
relaxation of the TSP, which can be solved in polynomial-time.

Once a small set of non-zero $w_S$ terms are computed as per
Equation~\ref{eq:moatpacking}, a \emph{nested} packing is obtained by
following the post-processing procedure described
by~\citeauthor{Ozener:2013}~\citeyear{Ozener:2013}. A packing is
\emph{nested} if and only if $\forall S',S''$ s.t. $w_{S'} > 0$ and
$w_{S''} > 0$, if $S' \cap S'' \not\equiv \emptyset$ then either $S'
\subseteq S''$ or $S'' \subseteq S'$. For any optimal solution to
Equation~\ref{eq:moatpacking} there is a corresponding nested packing
with the same objective value~\cite{Cornuejols:1985:TSP}. The nested
constraint is required and, intuitively, it prevents overcharging a
subset of locations that coalesce in a moat -- i.e. prevents the
allocation from violating the universally quantified constraint in the
definition of the core.
For the nesting critera to be violated there must be three distinct
non-empty sets of locations $S$, $S'$ and $S''$, so that $w_{S \cup S'}
> 0$ and $w_{S' \cup S''}>0$. Post-processing iteratively identifies and
eliminates such cases. Identification is straightforward. For each
elimination we take the assignment $\tau\leftarrow \min\{w_{S \cup S'},
w_{S' \cup S''}\}$, and make the following assignment updates to the
moat widths: $w_S\leftarrow w_S + \tau$, $w_{S''}\leftarrow w_{S''} +
\tau$, $w_{S \cup S'}\leftarrow w_{S \cup S'} - \tau$, and $w_{S' \cup
S''}\leftarrow w_{S' \cup S''} - \tau$.
This iterative procedure terminates yielding a nested packing, however
the algorithm can take exponential time in the worst case. That being
said, in all our experiments we found that nesting takes only a fraction
of a second. Finally, an $\epsilon$-core allocation is obtained where,
for each $S\subseteq N$ we distribute the cost $3\times w_S$ arbitrarily
to the locations in the set $(N\backslash 0)\backslash S$ -- we
distribute the term evenly to all nodes outside that moat for $S$,
excluding the depot node $0$.

\subsection{Hybrid Proxy}

Early on in our experimentation, we made an important observation that
lead us to develop a sixth ``blended'' proxy, $\fracblend$. This proxy
is a linear combination of $\fracmoat$ and $\fracdist$. We
experimentally identify a $\lambda \in [0,1]$ for which $\lambda \times
\fracmoat + (1-\lambda) \times \fracdist$ provides an improved proxy for
$\fracshapley$ compared to either component proxies in isolation.

Our observation is that the $\fracmoat$ does not properly distribute the
depot allocation of moat widths to other locations. In order to stay
within the $\nicefrac{1}{2}$-core allocation, that width is distributed
in equal parts to all locations. Blending the $\fracmoat$ with
$\fracdist$ mitigates this problem, and as we observe, increases proxy
accuracy relative to $\fracshapley$. The value of the improvement seems
to decrease gradually as the size of games increases.
Figure~\ref{fig:hyb} plots the benefit of blending proxies at different
values of $\lambda$ in our corpus of Synthetic games and the in a corpus
of Real-World transport scenarios. A detailed description of the
Real-World scenarios is given later in Section~\ref{sec:realworld}.
Experimentally we found $\lambda = 0.6$ to be most effective in
Synthetic games.  A clear signal for the optimal value of $\lambda$ 
in Real-World games is not obvious, however there is clear support in our 
data for blending the moat
and depot distances proxies.  The graphs in Figure~\ref{fig:hyb} show the 
average and worst case error in cost allocation to a particular
location.  We also measured the root mean squared error (RMSE)
over all locations.  The RMSE did not provide a clear signal to support 
a particular blending parameter, though it did remain
clear that blending performed better than either proxy in isolation
for both bot Synthetic and Real-World games.

%
\begin{figure}[h] \centering \begin{minipage}[b]{0.49\linewidth}
\centering \includegraphics[width=\columnwidth,page=1]{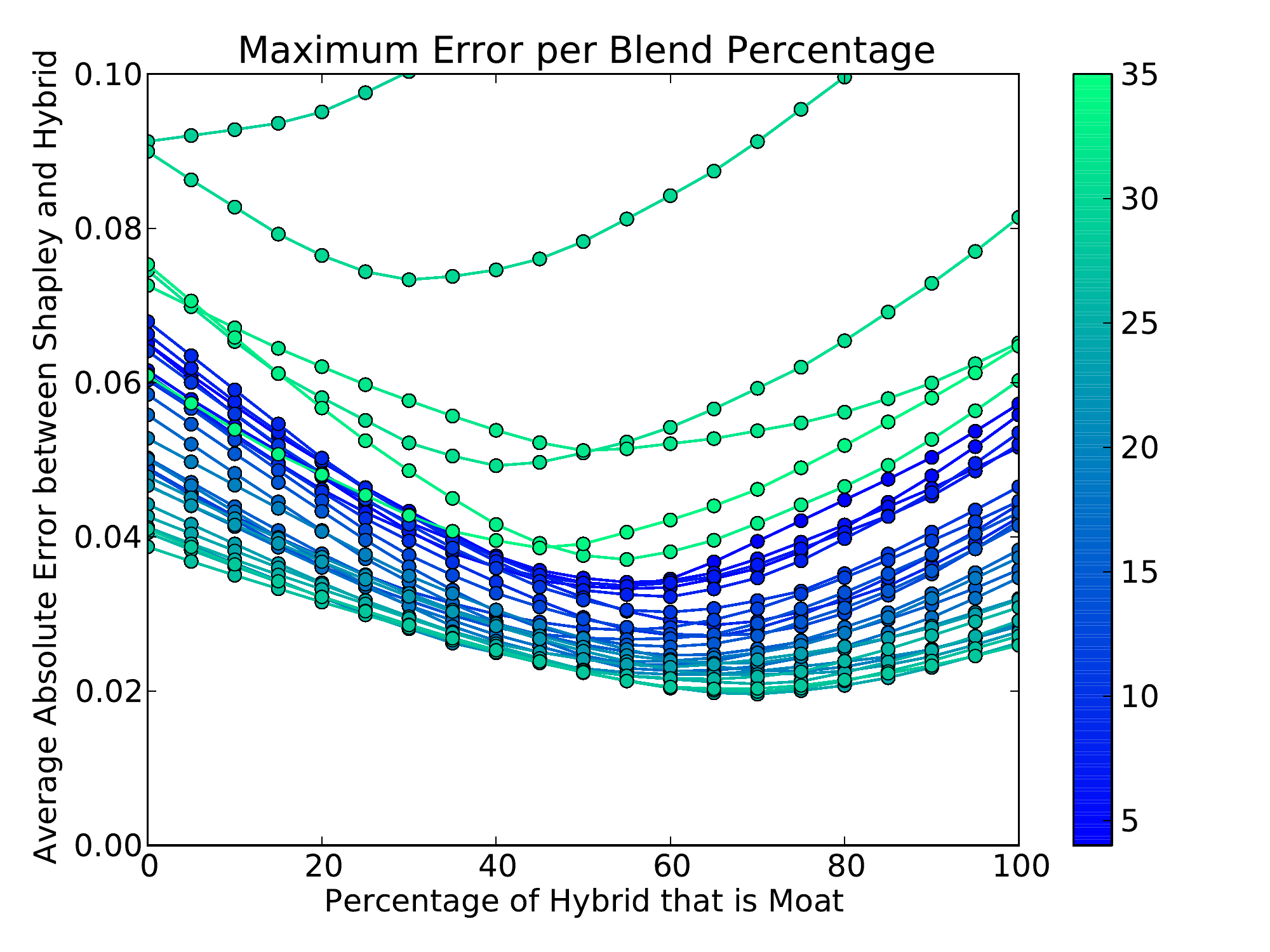} \end{minipage}
\begin{minipage}[b]{0.49\linewidth} \centering
\includegraphics[width=\columnwidth, page=3]{./Blend_Synth}
\end{minipage} \begin{minipage}[b]{0.49\linewidth} \centering
\includegraphics[width=\columnwidth, page=1]{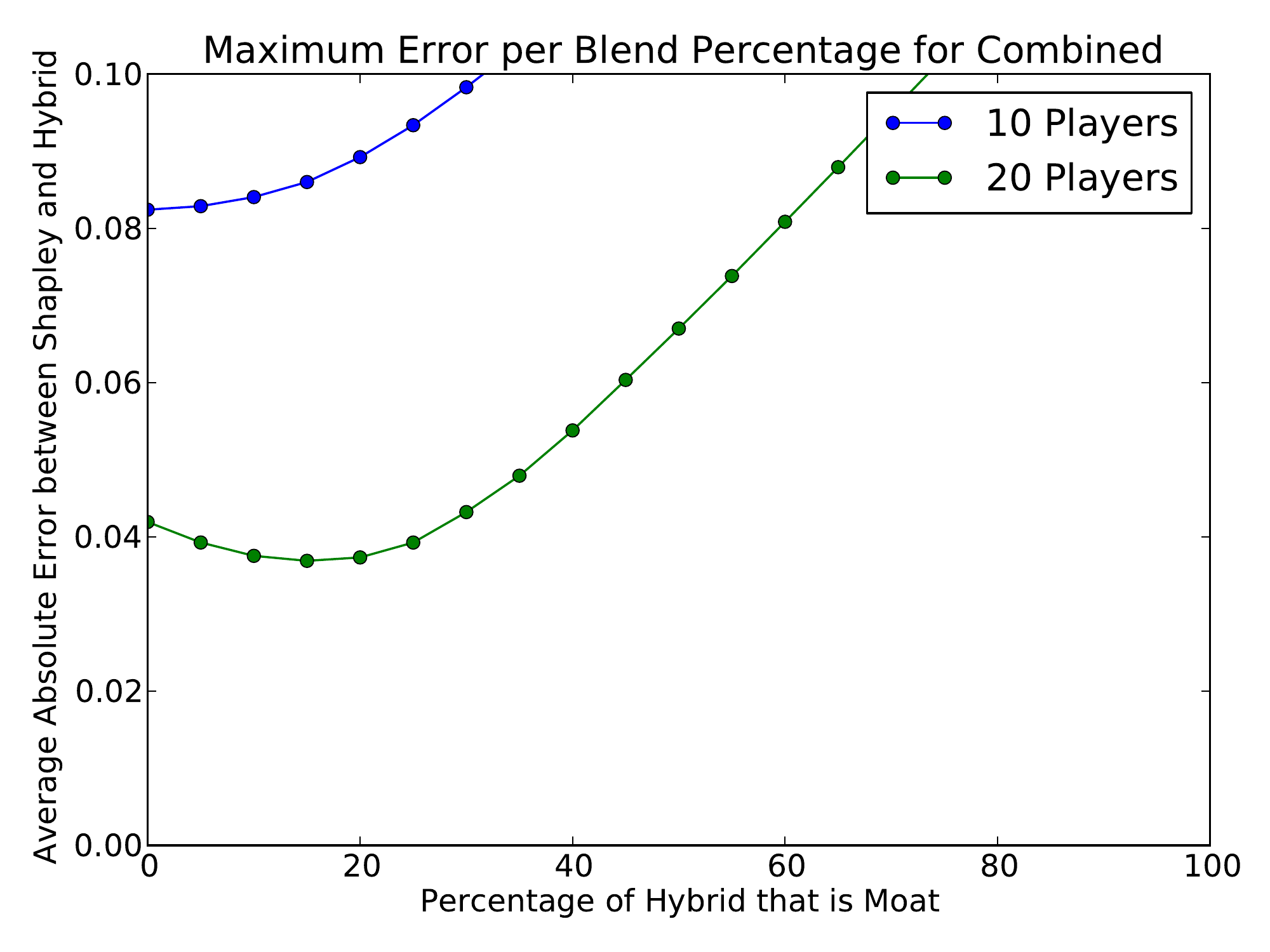}
\end{minipage} \begin{minipage}[b]{0.49\linewidth} \centering
\includegraphics[width=\columnwidth, page=3]{./Emp_Combined}
\end{minipage}

\caption{Effect of the blending parameter $\lambda$ on the error of
Shapley allocation prediction for all of the Synthetic datasets (top)
and all of the Real-World scenarios (bottom). The left-hand graph shows
the average worst case error measured at any single location, while the
right-hand graph shows the average error over all locations.}
\label{fig:hyb} 
\end{figure}

\clearpage \section{Analysis of Na\"ive Proxies}

We refer to the three proxies $\fracdist$, $\fracshort$ and
$\fracreroute$, as being \emph{na\"ive}. Contrastingly, we call
$\fracchris$, $\fracmoat$ and $\fracblend$ the  \emph{sophisticated}
proxies. The formulation of the na\"ive proxies $\fracdist$ and
$\fracshort$ make them amenable to direct analysis of their worst case
performance. We consider settings where the na\"ive proxies $\fracdist$
and $\fracshort$ can perform quite badly.

In order to illustrate this, consider a TSG where the depot is at one
corner of a square of dimension $a$ with one location at each of the
other 3 corners. Locations nearest the depot are indexed $1$ and $3$,
and the third location indexed $2$. \begin{center} \scalebox{1}{
\begin{tikzpicture}[auto]

\tikzstyle{ann} = [draw=none,fill=none,right] \tikzstyle{small}=[draw,
circle, fill=white!100,minimum size=5pt, inner sep=2pt]
\tikzstyle{depot}=[draw, square, fill=white!100] \tikzstyle{lab}=[]
\tikzstyle{collection}=[draw,circle,minimum size=11pt, inner sep=2pt]
\tikzstyle{dots}=[]

   \node[mark size=4pt,fill=white!100] (d) at (-1,-1)
   {\pgfuseplotmark{square*}}; \node[ann] at (-2.3,-1) {Depot};

    \node[small] (n1) at (-1,1) {}; \node[ann] at (-2.9,1) {Location 1};
    \node[small] (n3) at (1,-1) {}; \node[ann] at (1.1,-1) {Location 3};
    \node[small] (n2) at (1,1){}; \node[ann] at (1.1,1) {Location 2};   
   %
   \node[ann] () at (-.2,1.2) {$a$}; \node[ann] () at (-.2,-1.2) {$a$};
   \node[ann] () at (-1.4,.1) {$a$}; \node[ann] () at (1.05,0.1) {$a$};

%
%
%
%
        \draw[-,thick] (d) to (n1); \draw[-,thick] (n1) to (n2);
        \draw[-,thick] (n2) to (n3); \draw[-,thick] (n3) to (d);

    \end{tikzpicture}
    }
    \end{center}

\noindent Our na\"ive proxies yield the following allocations: $$
\begin{array}{c|ccc} i & \fracshapley & \fracdist & \fracshort \\ \hline
1,3 &  0.299a & 0.293a & 0.333a \\ 2   &  0.402a & 0.415a & 0.333a \\
\end{array} $$ \noindent Observe $\fracdist$ performs well in this case
(maximum of $\approx 11\%$ error) while $\fracshort$ does not (minimum
of $\approx 16\%$ error).

We now identify some pathological cases on which the $\fracshort$ and
$\fracdist$ proxies perform poorly. Our first result demonstrates that
$\fracdist$ and $\fracshort$ may under-estimate the true Shapley value
badly.

\begin{theorem} There exists an $n$ location TSP problem on which, for
some location $i$, the ratio $\nicefrac{\fracdist_i}{\fracshapley_i}$
goes to 0 as $n$ goes to $\infty$.  For the same problem the ratio
$\nicefrac{\fracshort_i}{\fracshapley_i}$ goes to $0$ as $n$ goes to
$\infty$ for $\Theta(n)$ of the locations. \end{theorem}

\begin{proof} Suppose the first $n-1$ locations are at distance $a$ from
the depot, whilst the $n$th location is located at a distance $a$ in the
opposite direction from the depot.

\begin{center} \scalebox{1}{ \begin{tikzpicture}[auto]

\tikzstyle{ann} = [draw=none,fill=none,right] \tikzstyle{large}=[draw,
circle, fill=white!100,minimum size=10pt, inner sep=5pt]
\tikzstyle{small}=[draw, circle, fill=white!100,minimum size=5pt, inner
sep=2pt] \tikzstyle{depot}=[draw, square, fill=white!100]
size=5pt, inner sep=1pt \tikzstyle{lab}=[]
\tikzstyle{collection}=[draw,circle,minimum size=11pt, inner sep=2pt]
\tikzstyle{dots}=[]

    \node[large] (many) at (-4,0) {}; \node[ann] () at (-6,-0.5)
    {Locations $1, \ldots, n-1$};

    \node[ann] () at (-2,0.2) {$a$};

        \node[small] (alone) at (4,0) {}; \node[ann] () at (3,-0.5)
        {Location $n$};

        \node[ann] () at (2,0.2) {$a$};

        \node[mark size=4pt,fill=white!100] at (0,0)
        {\pgfuseplotmark{square*}}; \node[ann] () at (-.50,-0.5)
        {Depot};

        \draw[-,thick] (many) to (alone);

    \end{tikzpicture}
    }
    \end{center}
%
Note that the normalization constant for $\fracshapley$, $\sum_{j \in n}
\shapley_j = 4a$. Now $\fracshapley_n = \nicefrac{2a}{4a} =
\nicefrac{1}{2}$ since the cost of adding the $n$th location to any
coalition is $2a$.  Leaving, for $i < n$, $$\fracshapley_i =
\frac{\nicefrac{2a}{(n-1)}}{4a} = \frac{1}{2(n-1)}.$$ On the other hand,
the normalization constant for $\fracdist$, $\sum_{i=1}^n d_{i0} = na$
since all locations are equidistant from the depot. Giving, for all $i
\leq n$, $\fracdist_i = \frac{1}{n}.$

Thus for $i < n$, $$\frac{\fracdist_i}{\fracshapley_i} =
\frac{\nicefrac{1}{n}}{\nicefrac{1}{2(n-1)}} = \frac{2n-1}{n}$$ which
goes to $2$ as $n \rightarrow \infty$. While
$$\frac{\fracdist_n}{\fracshapley_n} =
\frac{\nicefrac{1}{n}}{\nicefrac{1}{2}} = \frac{1}{2n}$$ which goes to
$0$ as $n \rightarrow \infty$.

Note that the shortcut proxy, $\fracshort$ performs poorly on this
example. For $i<n$, $\fracshort_i=0$ since all the locations are
co-located, leaving $\fracshort_n = 1$.  For $i<n$ we have
$\fracshapley_i=\nicefrac{1}{2}(n-1).$ Thus, for $i<n$,
$$\frac{\fracshort_i}{\fracshapley_i} = \frac{0}{\nicefrac{1}{2(n-1)}} =
0$$ and $$\frac{\fracshort_n}{\fracshapley_n} =
\frac{1}{\nicefrac{1}{2}} = 2$$ \end{proof}

Our second result demonstrates that $\fracdist$ can also over-estimate
the true Shapley value badly.

\begin{theorem} There exists an $n$ location TSG where the ratio
$\nicefrac{\fracshapley_i}{\fracdist_i}$ goes to 0 as $n$ goes to
$\infty$ for $\Theta(n)$ of the locations. \end{theorem}

\begin{proof} Suppose the first $n-1$ locations are at distance $a$ from
the depot, whilst the $n$th location is located at a distance $(n+1)a$
from the depot in the opposite direction.

 \begin{center} \scalebox{1}{ \begin{tikzpicture}[auto]

\tikzstyle{ann} = [draw=none,fill=none,right] \tikzstyle{large}=[draw,
circle, fill=white!100,minimum size=10pt, inner sep=5pt]
\tikzstyle{small}=[draw, circle, fill=white!100,minimum size=5pt, inner
sep=2pt] \tikzstyle{depot}=[draw, square, fill=white!100]
size=5pt, inner sep=1pt \tikzstyle{lab}=[]
\tikzstyle{collection}=[draw,circle,minimum size=11pt, inner sep=2pt]
\tikzstyle{dots}=[]

    \node[large] (many) at (-3,0) {}; \node[ann] () at (-5,-0.5)
    {Locations $1, \ldots, n-1$};

    \node[ann] () at (-2,0.2) {$a$};

        \node[small] (alone) at (7,0) {}; \node[ann] () at (6,-0.5)
        {Location $n$};

        \node[ann] () at (3.5,0.2) {$a(n+1)$};

        \node[mark size=4pt,fill=white!100] at (0,0)
        {\pgfuseplotmark{square*}}; \node[ann] () at (-0.5,-0.5)
        {Depot};

        \draw[-,thick] (many) to (alone);

    \end{tikzpicture}
    }
    \end{center}
%
%
Note that the normalization constant for $\fracshapley$, $\sum_{j \in n}
\shapley_j = 2a+2a(n+1) = 2a(n+2)$. The Shapley value $\shapley_i$ for
any $i<n$ is $\frac{2a}{n-1}$, thus $$\fracshapley_i =
\frac{\nicefrac{2a}{n-1}}{2a(n+2)} = \frac{1}{(n-1)(n+2)}.$$ While the
fractional Shapley allocation for location $n$ is $$\fracshapley_n =
\frac{2a(n+1)}{2a(n+2)} = \frac{1}{2}.$$

The normalization constant for $\fracdist$ is $\sum_{i=1}^n d_{i0} =
a(n-1) + a(n+1) = 2an$. For location $n$ the assignment from the
distance based proxy is $$\fracdist_n = \frac{a(n+1)}{2an} =
\frac{n+1}{2n}.$$ For $i < n$, $$\fracdist_i = \frac{a}{2an} =
\frac{1}{2n}.$$

Thus, for location $n$ we have $$\frac{\fracshapley_n}{\fracdist_n} =
\frac{\nicefrac{1}{2}}{\nicefrac{n+1}{2n}} = \frac{2n}{2n+1}$$ which
goes to $1$ as $n$ goes to $\infty$.

For $i < n$ we have $$\frac{\fracshapley_i}{\fracdist_i} =
\frac{\nicefrac{1}{(n-1)(n+2)}}{\frac{1}{2n}} = \frac{2n}{(n-1)(n+2)}$$
which goes to $0$ as $n$ goes to $\infty$.

For the $\fracshort$ we again have $i < n$, $\fracshort_i = 0$ leaving
$\fracshort_n = 1$. Thus, $\nicefrac{\fracshapley_n}{\fracshort_n} =
\nicefrac{1}{2}$ while for $i < n$,
$\nicefrac{\fracshapley_i}{\fracshort_i}$ is undefined. \end{proof}

Our third result demonstrates that $\fracshort$ may under-estimate the
Shapley value badly even on very simple examples which may be embedded
in larger problems.

\begin{theorem} There exists a $2$ location TSG instance for which
$\nicefrac{\fracshort}{\fracshapley} = 0$ for one of the two locations.
\end{theorem}

\begin{proof} Suppose the first location is located a distance $a$ from
the depot with the second location located a distance of $a$ farther
down the road.

 \begin{center} \scalebox{1}{ \begin{tikzpicture}[auto]

\tikzstyle{ann} = [draw=none,fill=none,right] \tikzstyle{large}=[draw,
circle, fill=white!100,minimum size=10pt, inner sep=5pt]
\tikzstyle{small}=[draw, circle, fill=white!100,minimum size=5pt, inner
sep=2pt] \tikzstyle{depot}=[draw, square, fill=white!100]
size=5pt, inner sep=1pt \tikzstyle{lab}=[]
\tikzstyle{collection}=[draw,circle,minimum size=11pt, inner sep=2pt]
\tikzstyle{dots}=[]

    \node[small] (first) at (0,0) {}; \node[ann] () at (-.8,-0.5)
    {Location 1};

    \node[ann] () at (2,0.3) {$a$};

    \node[small] (second) at (4,0) {}; \node[ann] () at (3.2,-0.5)
    {Location 2};

      \node[ann] () at (-2,0.3) {$a$};

        \node[mark size=4pt,fill=white!100] (d) at (-4,0)
        {\pgfuseplotmark{square*}}; \node[ann] () at (-4.5,-0.5)
        {Depot};

        \draw[-,thick] (d) to (first); \draw[-,thick] (first) to
        (second);

    \end{tikzpicture}
    }
    \end{center}

For the first location we have $\fracshort_1 = 0$, as removing it has no
effect on the distance we must travel to the second location.  This
leaves $\fracshort_2 = 1$. The Shapley value for the first location is
$$\shapley = \frac{2a}{2} + \frac{0}{2} = a.$$ Which gives $\fracshapley
= \nicefrac{a}{4}$ and thus $$\frac{\fracshort}{\fracshapley} =
\frac{0}{\nicefrac{a}{4}} = 0.$$

\end{proof}

Our fourth and final result demonstrates that $\fracshort$ may
over-estimate the Shapley value badly.

\begin{theorem} There exists a four location TSG for which
$\nicefrac{\fracshapley}{\fracshort} = 0$ for two of the four cities.
\end{theorem}

\begin{proof} Consider a four location TSG where locations $1$ and $2$
are $\epsilon$ from each other and the depot while cities $3$ and $4$
are at a distance $a$ from the depot and $\epsilon$ from each other.

 \begin{center} \scalebox{1}{ \begin{tikzpicture}[auto]

\tikzstyle{ann} = [draw=none,fill=none,right] \tikzstyle{large}=[draw,
circle, fill=white!100,minimum size=10pt, inner sep=5pt]
\tikzstyle{small}=[draw, circle, fill=white!100,minimum size=5pt, inner
sep=2pt] \tikzstyle{depot}=[draw, square, fill=white!100]
size=5pt, inner sep=1pt \tikzstyle{lab}=[]
\tikzstyle{collection}=[draw,circle,minimum size=11pt, inner sep=2pt]
\tikzstyle{dots}=[]

    \node[small] (first) at (-3,-.5) {}; \node[ann] () at (-3.2,1)
    {Location 1};

    \node[small] (second) at (-3,.5) {}; \node[ann] () at (-3.2,-1)
    {Location 2};

    \node[small] (fourth) at (5,.5) {}; \node[ann] () at (5.3,0.5)
    {Location 4};

    \node[small] (third) at (5,-.5) {}; \node[ann] () at (5.3,-0.5)
    {Location 3};

    \node[ann] () at (5.05,0) {$\epsilon$}; \node[ann] () at (-2.95,0)
    {$\epsilon$}; \node[ann] () at (-3.8,0.5) {$\epsilon$}; \node[ann]
    () at (-3.8,-0.5) {$\epsilon$};

    \node[ann] () at (1,0.8) {$ka$}; \node[ann] () at (1,-0.8) {$ka$};

     \node[mark size=4pt,fill=white!100] (d) at (-4,0)
     {\pgfuseplotmark{square*}}; \node[ann] () at (-5.5,0) {Depot};

     \draw[-,thick] (d) to (first); \draw[-,thick] (first) to (second);
     \draw[-,thick] (d) to (first); \draw[-,thick] (d) to (second);
     \draw[-,thick] (first) to (third); \draw[-,thick] (second) to
     (fourth); \draw[-,thick] (third) to (fourth);

    \end{tikzpicture}
    }
    \end{center}

We note that here $\epsilon << ka$, as such we will hide $\epsilon$
terms in $O(\epsilon)$.  The marginal cost saved by skipping any
location is $\epsilon$, this means that all locations have the same
allocation according to $\fracshort$, namely for all $i \in \{1, \dots,
4\}$, $\fracshort_i = \nicefrac{1}{4}$.

Note that the normalization constant for $\fracshapley$, $\sum_{j \in n}
\shapley_j = 2ka + O(\epsilon)$. To compute the Shapley values for
locations $1$ and $2$ we observe that, in any given permutation, each
location adds a multiple of $\epsilon$, thus by symmetry, for $i \in
\{3,4\}$, $$\fracshapley_i = \frac{O(\epsilon)}{2ka + O(\epsilon)}$$ To
compute the Shapley value for locations $3$ and $4$ we observe that, no
matter where in the permutation they appear, the first contributes $2ka$
while the other contributes only $\epsilon$.  Consequently, by symmetry,
for locations $i \in \{3,4\}$, $$\fracshapley_i =
\frac{\frac{2ka+O(\epsilon)}{2}}{2ka+O(\epsilon)} = \frac{1}{2}.$$

Thus, locations $i \in \{1,2\}$, we have
$$\frac{\fracshapley}{\fracshort} =
\frac{\frac{O(\epsilon)}{2ka+O(\epsilon)}}{\nicefrac{1}{4}} =
\frac{4O(\epsilon)}{2ka+O(\epsilon)}.$$ The term goes to $0$ as $k$ goes
to $\infty$.

\end{proof}

\section{Empirical Study}

We implemented each of the six proxies discussed, along with a version
of  \emph{ApproShapley} that uses \emph{Concorde}~\cite{Concorde} to
evaluate the characteristic function of the TSG. The {\em Concorde}
program is used to find optimal solutions to  TSPs. Rather than
calculating $\fracshapley$ by direct enumeration  as a baseline to
compare proxies, we estimate that value using \emph{ApproShapley} with
\emph{Concorde}. For the size of games we have considered, we find
that  $4000$ iterations of \emph{ApproShapley} to be sufficient to
obtain accurate  baseline values.

We experimented using a corpus of games comprised of two sets of TSGs.
The first set of games are Synthetic. For each $i\in[4,\ldots,35]$, we
generate $20$ instances of the Euclidean TSG with $i$ locations
occurring uniformly at random in a square  of dimension $1,000$. The
horizontal and vertical coordinates of the locations are represented
using  32-bit floating point numbers. Those Euclidean games are
available online at
\url{http://users.cecs.anu.edu.au/~charlesg/tsg_euclidean_games.tar.gz}.

 \label{sec:realworld} The second set of games is taken from large
 Real-World VRPs in the cities of Auckland, New Zealand; Canberra,
 Australia; and Sydney, Australia. Heuristic solutions to those VRPs are
 calculated using the~{\em Indigo} solver~\cite{kilby2011}. That is a
 flexible heuristic which implements an Adaptive Large Neighbourhood
 Search, the basic structure of which is described in detail by Ropke
 and Pisinger in~\cite{Ropk06Adaptive}.
 \footnote{{\em Indigo} is a strong vehicle routing solution platform,
recently computing $5$ new best solutions for $1,000$ customer problems from the 
VRPTW benchmark library. The solutions computed using Indigo were certified by Dr.~Geir Hasle,
Chief Research Scientist at SINTEF and maintainer of the 
VRPTW benchmark library, as the best currently known on September 24th of 2013. 
\url{http://www.sintef.no/Projectweb/TOP/VRPTW/Homberger-benchmark/1000-customers}.}
To give an indication of the scale and difficulty of these VRPs, the
Auckland model comprises $1,166$ locations to be served using a fleet of
at most $25$ vehicles over a $7$ day period. In the heuristic solution
we collect tours of length $10$ and $20$ to created TSGs for testing.
Because Real-World distance matrices are asymmetric, in all cases
asymmetry is negligible, we induce symmetric  problems by resolving for
the greater of $d_{ij}$ and~$d_{ji}$ -- i.e. setting $d_{ij} = d_{ji} =
\max\{d_{ij}, d_{ji}\}$. It total we obtain 69 Real-World games of size
10 and 44 games of size 20.~\footnote{Due to commercial agreements with
our industrial partners we cannot release these Real-World games.}

All experiments reported here were performed on a computer with an Intel
i7-2720QM CPU running at 2.20GHz, with 8GB of RAM, and running the {\em
Ubuntu 12.04.3 LTS} operating system. For Synthetic problems with $35$
locations, 4000 iterations of \emph{ApproShapley} with exact TSP
evaluations using \emph{Concorde}~\cite{Concorde} takes $545$ seconds.
Computing $\fracchris$, which replaces the exact TSP computation with an
evaluation of the Christofides heuristic, results in a reduction to
$11.39$ seconds in total. Computing $\fracmoat$ takes under 1 second.
All the na\"ive proxies, namely $\fracdist$, $\fracshort$, and
$\fracreroute$, take fractions of a second to compute.

Our experimental analysis assumes the reader is familiar with a number
of statistical measures which we summaries in
Appendix~\ref{appendix:stats}. To evaluate how well  proxies perform in
approximating $\fracshapley$ we measure the \emph{point-wise
root-mean-squared error} (RMSE) in each game. We also use Kendall's
$\tau$ \shortcite{kendall1938new} (written KT) to compare the ranking---i.e. least
expensive to most expensive---of locations induced by the  Shapley
allocation and our proxies. The value $\tau$ measures the amount of
disagreement between two rankings. It is customary to report $\tau$ as a
normalized value (correlation coefficient) between 1 and -1, where 
$\tau = 1$ means that two lists are perfectly correlated (equal) and
$\tau = -1$ means that two lists are perfectly  anti-correlated (they
are equal if one list is reversed). Our analysis makes use of the
significance, or $p$-value of a computed $\tau$. The $p$-value is
computed using a two-tailed $t$-test where the null hypothesis is that
there is no correlation between orderings ($\tau = 0$). Taking our
significance threshold to be the customary $0.05$, we can reject the
null hypothesis when $p \leq 0.05$. When $p \geq 0.05$ we fail to reject
the null hypothesis, a $p$-value $ \leq 0.05$ is a statistically
significant result. This means it is unlikely that two random,
uncorrelated lists would show such a high degree of correlation.

\subsection{Synthetic Data}

Figure \ref{fig:snakes} shows the average root mean squared error and
average KT distance for each proxy from $\fracshapley$ for all game
sizes of the Synthetic data. A complete set of tables and results from
the Synthetic Data can be found in Appendix~\ref{appendix:synthetic}. We
describe highlights of our results here. Overall, the best performing
proxy is $\fracblend$, both in terms of lowest RMSE and  highest average
$\tau$. The $\fracshort$ and $\fracreroute$ proxies are by far the
worst, particularly in terms of approximating Shapley value, but also in
terms of the ranking induced by the corresponding allocations. The
computationally more expensive proxy $\fracreroute$ always dominates
$\fracshort$; a trend which continues throughout our testing on
Real-World data as well. The proxy $\fracdist$ performs poorly at
ranking, however does surprisingly well at approximation being almost
competitive with the more sophisticated proxies. In ranking locations,
$\fracreroute$ regularly identifies the location ranked most costly
according to the Shapley value, outperforming all proxies on this task
for the synthetic data. More generally, in $\geq 60\%$ of synthetic
games the $\fracchris$, $\fracmoat$, $\fracreroute$,  and $\fracblend$
proxies each correctly identifies the most costly location.

In the majority of the synthetic games, our analysis of rankings using
Kendall's $\tau$ strongly implies that  $\fracchris$, $\fracmoat$ and
$\fracblend$ rankings are correlated with $\fracshapley$. Put simply, we
are confident that sophisticated proxies are inducing a ranking that is
similar to the one induced by the Shapley value. They also reliably
identify the most expensive location.
Among the pure proxies, the $\fracchris$ proxy outperforms all the
others at ranking by a slim margin. For example, it is able to identify
the most expensive location according to the Shapley value $66.4\%$ of
the time. Additionally, regardless of the number of locations, the mean
value for $\tau$ between $\fracshapley$ and $\fracchris$ is $\geq 0.55$,
and in \emph{every} instance with $18$ or more locations (and for the
majority of instances between 4 and 17 locations) there is a
statistically significant result for $\tau$.  Comparatively,
$\fracblend$ returns similar (and often higher) results for $\tau$ while
achieving a statistically significant correlation with the ranking
induced by $\fracshapley$ for every synthetic game instance with more
than $8$ players, save $6$. The $\tau$ analysis in the case of
$\fracmoat$ is less positive, gives strong correlation in instances with
more than $20$ locations, though still better than any of the naive
proxies.

\begin{figure} \begin{minipage}[b]{\linewidth} \centering
\includegraphics[height=9.2cm, page=8]{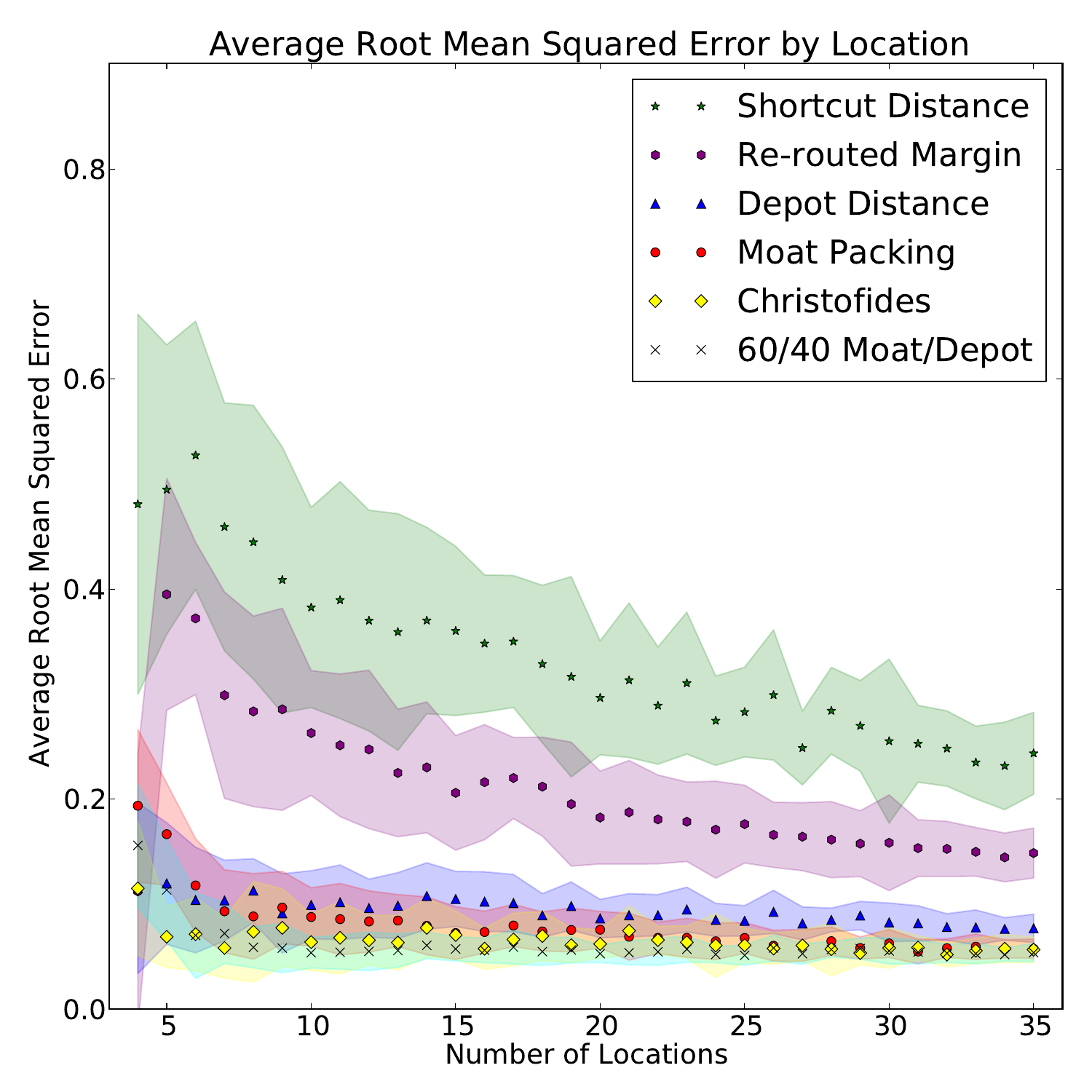}
\end{minipage} \begin{minipage}[b]{\linewidth} \centering
\includegraphics[height=9.2cm, page=16]{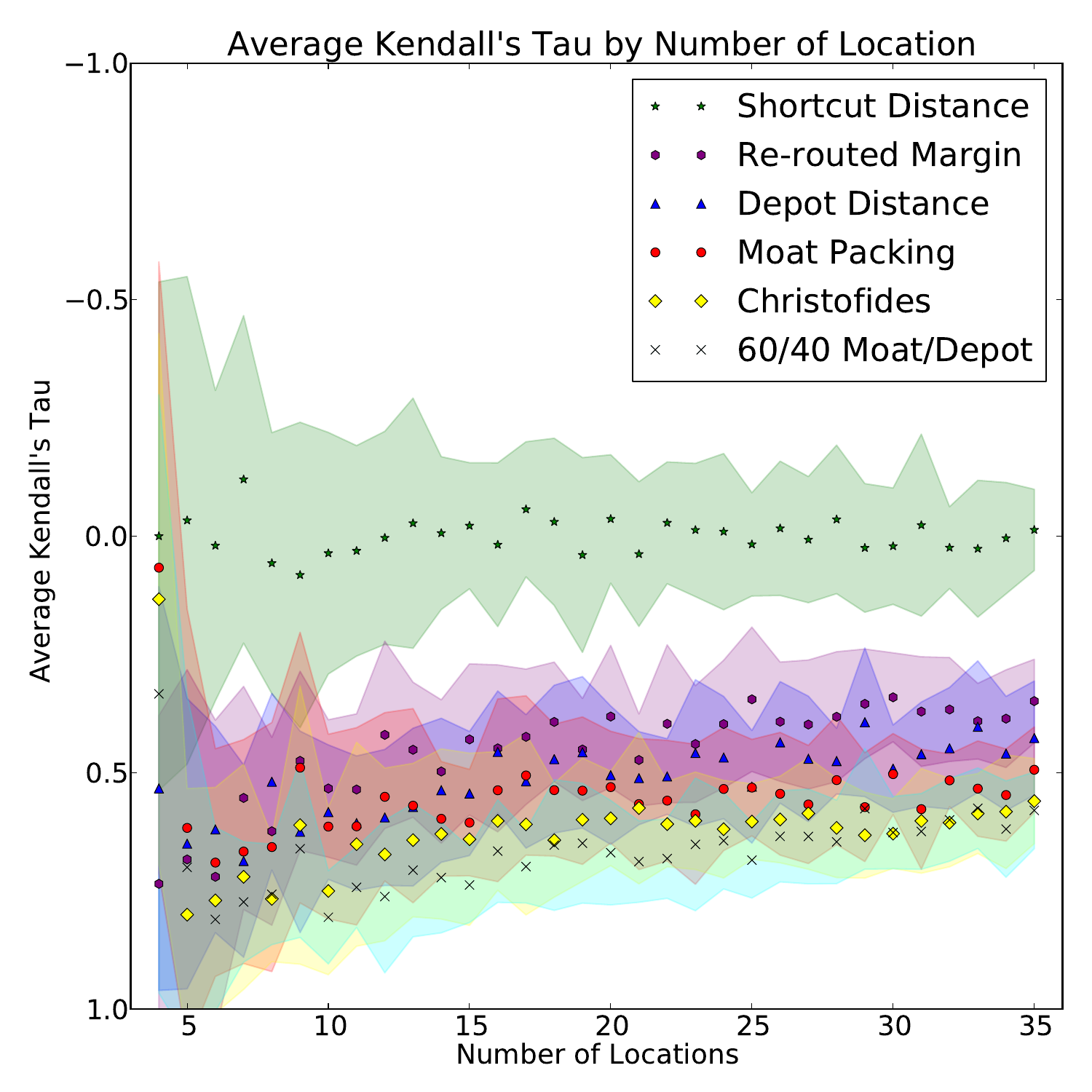}
\end{minipage} \caption{Performance of the five pure proxies and one
hybrid proxy according to: (\textbf{left}) RMSE averaged over the 20
games generated for each number of locations, and ({\bf right})
Kendall's tau rank correlation averaged over the 20 games generated for
each number of locations. The error bands correspond to plus or minus
one standard deviation.  The horizontal axis of our Kendall's tau plot
has been inverted for ease of comparison -- i.e. more correlated lists
are towards the bottom of the graph ($1.0$).} \label{fig:snakes}
\end{figure}

Our experimental analysis also considered how the types of allocation
error differ between proxies. For example, we considered questions, such
as: Do the proxies make a lot of small errors for low cost locations, or
do they  make large errors for locations that are apportioned large
costs? Knowledge about the type and severity of errors made by our
different proxies provides some guidance to the situations where we
should have confidence in proxy allocations and/or the induced rankings.

Figure~\ref{fig:synthetic_scatter_error} shows the absolute error
between each of the proxies and $\fracshapley$ graphed as a function of
the allocation according to $\fracshapley$.
For all the proxies, there appears to be a strong linear component to
the error --- many of the proxies allocate proportionally more (or less)
cost compared to the $\fracshapley$ allocation. In some cases
$\fracreroute$ allocates more than 20-times the cost allocation by
$\fracshapley$, though typically this happens in the case of locations
that received less than 10\% of the Shapley allocation. We find that
better performing proxies make more constant real-valued errors across
all locations,  regardless of actual allocation. The scatterplots for
$\fracblend$ and $\fracchris$ both show the weakest linear bias, with
$\fracblend$ showing a somewhat sub-linear bias. For example,
$\fracchris$ and $\fracmoat$ can allocate 6-times $\fracshapley$, though
 this only occurs in the case of locations whose Shapley allocation is
less that 5\% of the tour cost. Measuring the factor by which it
overestimates allocations, the $\fracdist$ proxy appears to perform
rather well, allocating at most 2.5-times the fair cost. The caveat is
that $\fracdist$ is indiscriminate, also making proportionately large
over-allocation errors to locations which are costly according to
$\fracshapley$.

\begin{figure} \begin{minipage}[b]{0.49\linewidth} \centering
\includegraphics[height=6.7cm, page=66]{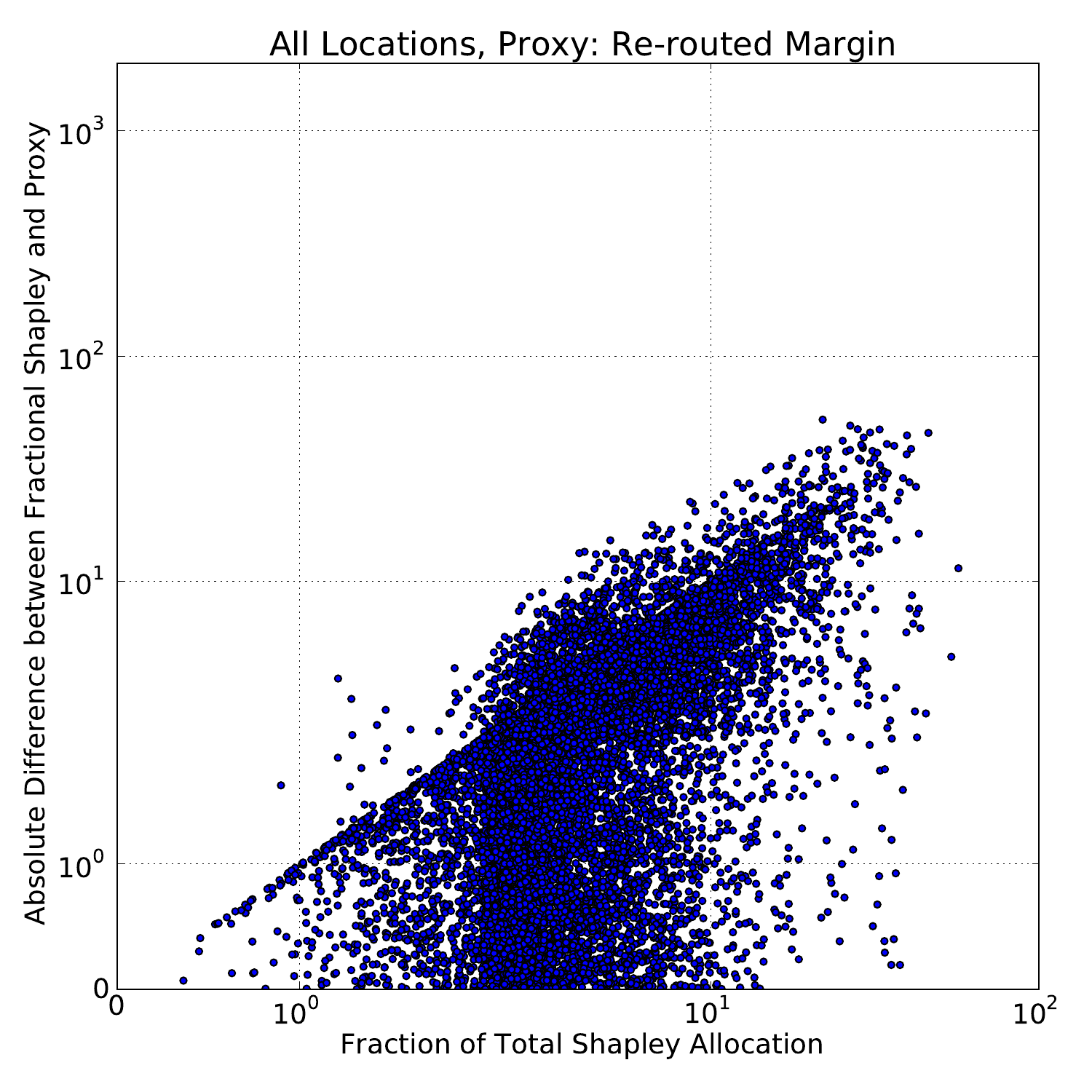}
\includegraphics[height=6.7cm,
page=165]{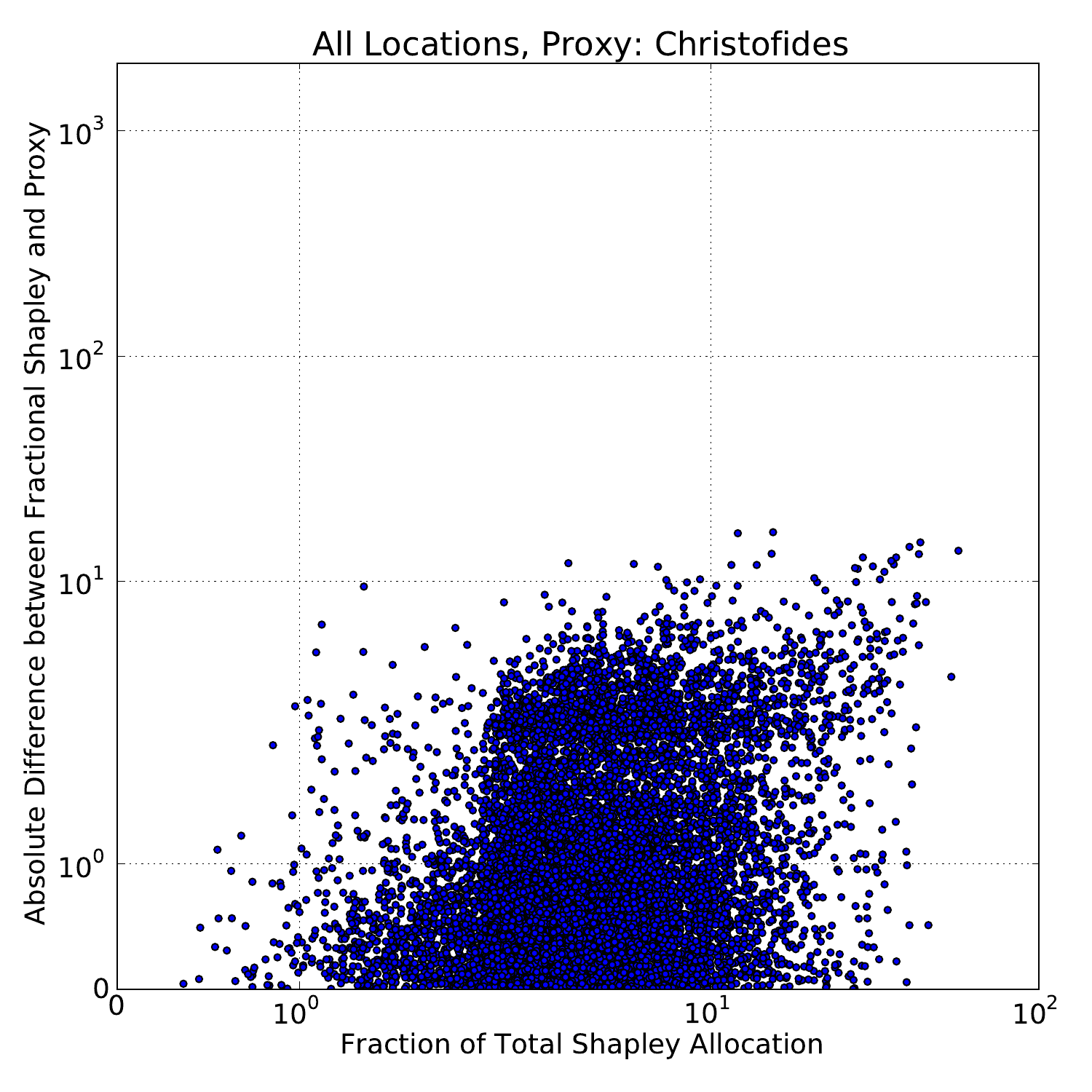}
\includegraphics[height=6.7cm, page=33]{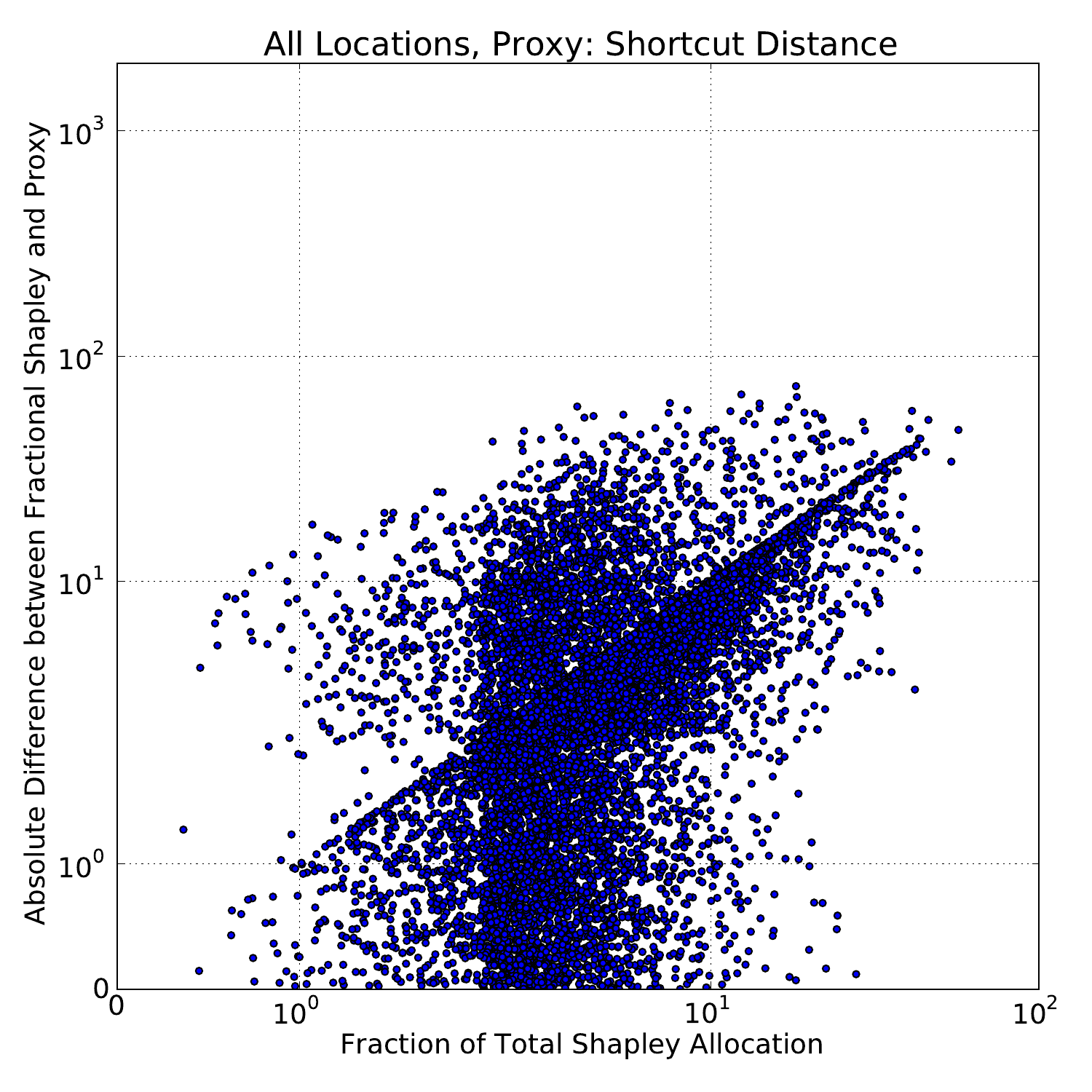}
\end{minipage} \begin{minipage}[b]{0.49\linewidth} \centering
\includegraphics[height=6.7cm,
page=132]{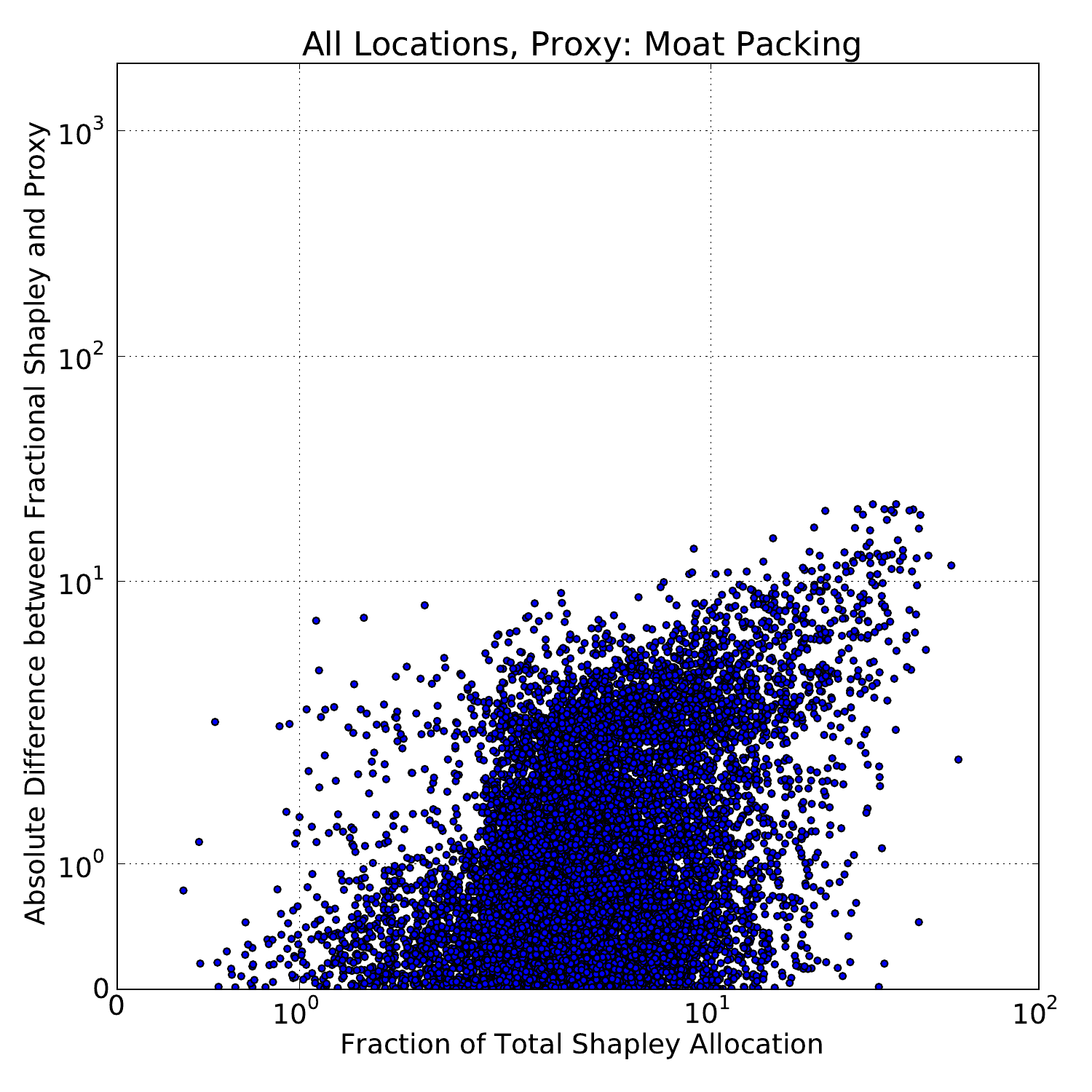}
\includegraphics[height=6.7cm, page=99]{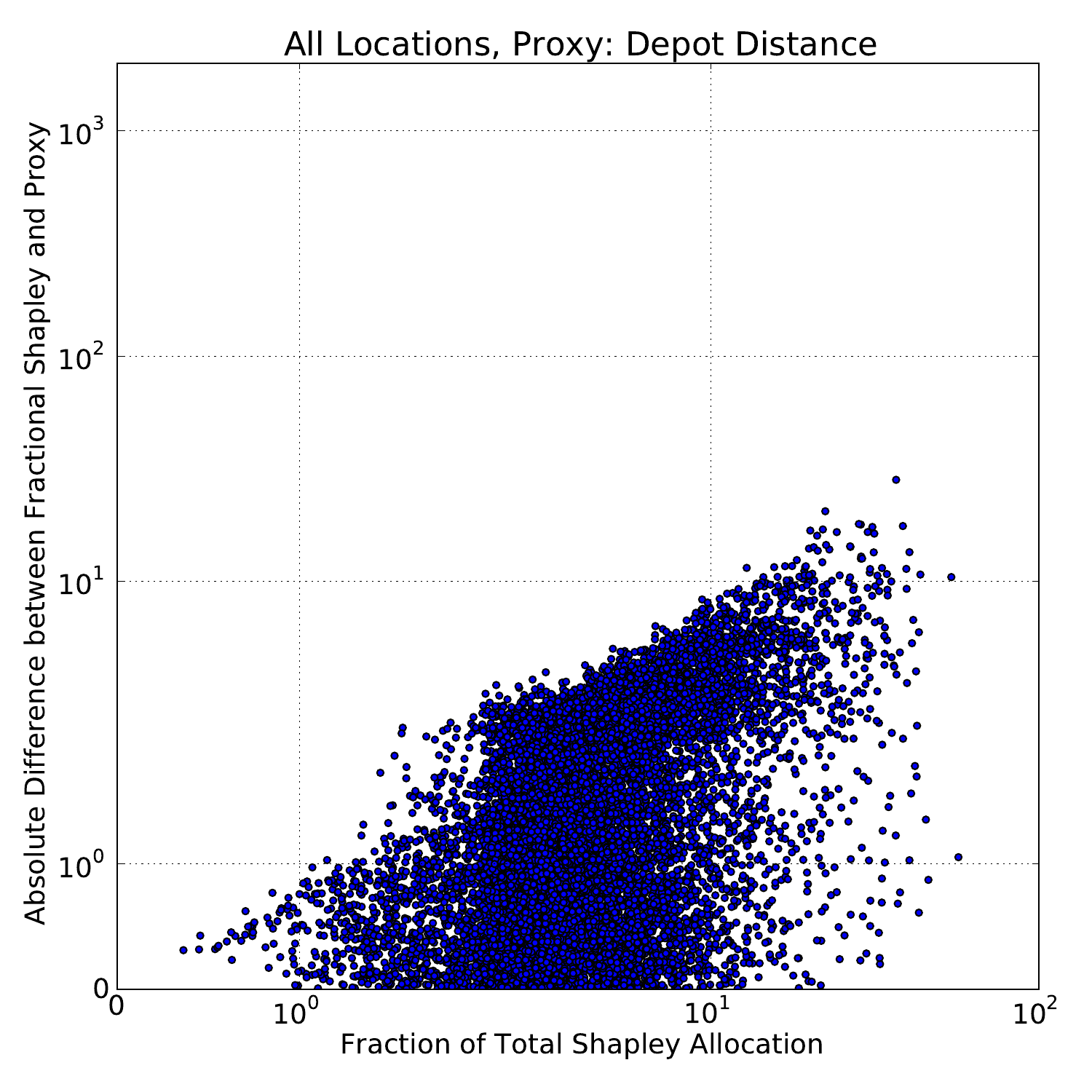}
\includegraphics[height=6.7cm,
page=198]{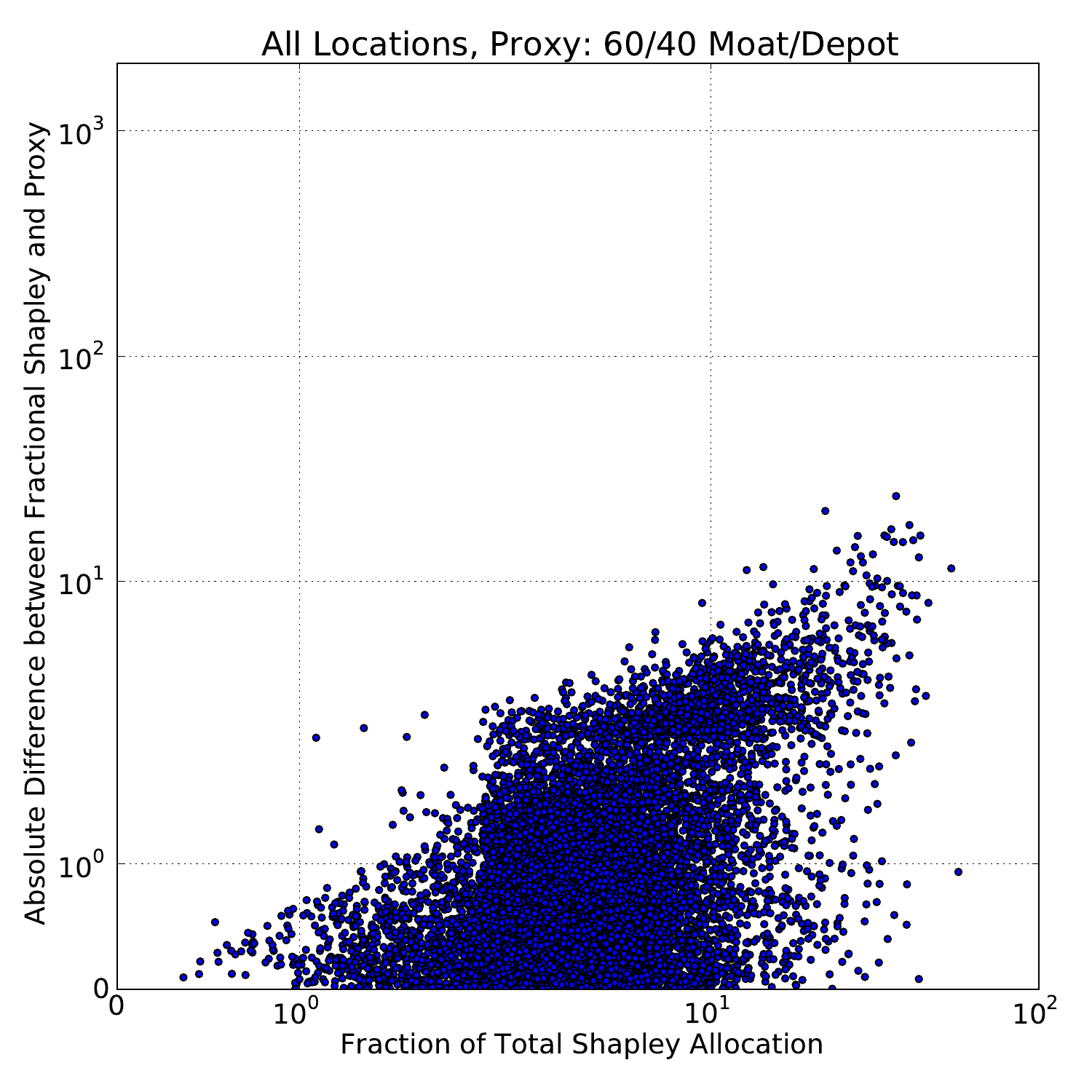} \end{minipage}
\caption{Absolute value of the difference between the $\fracshapley$ and
$\fracproxy$ plotted as a function of $\fracshapley$ for all the points
in the Synthetic data for all game sizes.  Note that these are log-log
plots to highlight the spread of the data.}
\label{fig:synthetic_scatter_error} \end{figure}

\subsection{Real-World Data}

Measuring the performance of proxies in Real-World data from Auckland,
Canberra, and Sydney,  overall we find the quality of allocation is
slighly degraded compared to measurements we made in synthetic games. We
identified no significant performance differences between cities. A
complete set of tables and results for each of the cities can be found
in Appendices~\ref{appendix:auckland} through \ref{appendix:sydney}; we
report on the combined statistics of these games in this section.
Summary statistics for these games are shown in 
Tables~\ref{tab:combined_rmse} through \ref{tab:combined_percent}.

\begin{table}[H] \centering \begin{tabular}{lcccccc}\toprule
&\multicolumn{2}{c}{10 Locations} & \multicolumn{2}{c}{20 Locations} &
\multicolumn{2}{c}{All Games} \\ & RMSE & St. Dev. & RMSE & St.Dev.&
RMSE & St. Dev.  \\ \midrule Shortcut Distance & 0.4429 & 0.1436 &
0.3239 & 0.1064 	& 0.3966 & 0.1291 \\ Re-routed Margin  & 0.4160 &
0.1328 & 0.2902 & 0.0934 	& 0.3670 & 0.1175 \\ Depot Distance    &
0.1346 & 0.0616 & 0.0870 & 0.0303 	& 0.1160 & 0.0494 \\ \midrule Moat
Packing      & 0.2478 & 0.1247 & 0.1969 & 0.0883 	& 0.2280 & 0.1105
\\ Christofides      & 0.1338 & 0.0694 & 0.0863 & 0.0311 	& 0.1153 &
0.0545 \\ 60/40 Moat/Depot  & 0.1442 & 0.0697 & 0.0765 & 0.0301 & 0.1178
& 0.0542 \\
\bottomrule \end{tabular} \caption{Root Mean Squared Error (RMSE) and
Standard Deviation (St. Dev.) for the combined Real-World datasets for
games with 10 and 20 locations. Lower is better.}
\label{tab:combined_rmse} \end{table}

\begin{table}[H] \centering \begin{tabular}{lcccccc} \toprule
&\multicolumn{2}{c}{10 Locations} & \multicolumn{2}{c}{20 Locations} &
\multicolumn{2}{c}{All Games} \\ & $\tau$ & St. Dev. & $\tau$ & St.
Dev.& $\tau$ & St. Dev.  \\ \midrule Shortcut Distance & -0.0135 &
0.2692 & 0.0798 & 0.1891 		& 0.0228 	& 0.2380 \\ Re-routed
Margin  & 0.3578  & 0.2388 & 0.3476 & 0.1993 	& 0.3538 	& 0.2234
\\ Depot Distance    & 0.1062  & 0.2382 & 0.1622 & 0.2313 		&
0.1280 	& 0.2355 \\ \midrule Moat Packing      & 0.3450  & 0.2554 &
0.3064 & 0.1710 		& 0.3300 	& 0.2225 \\ Christofides      &
0.2464  & 0.2770 & 0.3509 & 0.2258 		& 0.2871 	& 0.2571 \\
60/40 Moat/Depot  & 0.2037  & 0.2524 & 0.2531 & 0.2487 	& 0.2229 &
0.2510 \\ \bottomrule
%
\end{tabular} \caption{Average KT distance ($\tau$) and Standard
Deviation (St. Dev.) for the combined Real-World datasets for games with
10 and 20 locations. Higher is better; $+1$ means the two lists are
perfectly correlated and $-1$ means the two lists are perfectly
anti-correlated.} \label{tab:combined_kt} \end{table}

\begin{table}[H] \centering \begin{tabular}{lcccc} \toprule &
\multicolumn{2}{c}{Median} & \multicolumn{2}{c}{Num. Significant} \\ &
10 Locations           & 20 Locations          & 10 Locations         &
20 Locations         \\ \midrule Shortcut Distance & 0.4564 & 0.3349
					& 3 & 4 \\ Re-routed Margin  & 0.1730 & 0.0354
					& 15 & 25 \\ Depot Distance    & 0.4631 & 0.3276
					& 3 & 9 \\ \midrule Moat Packing      & 0.1444 &
0.0531 						& 16 & 21 \\ Christofides      & 0.2109 &
0.0275 						& 10 & 25 \\ 60/40 Moat/Depot  & 0.4042 &
0.1239 					& 9 & 15 \\ \bottomrule
\end{tabular} \caption{Median $p$ (lower is better) and count of the
number of statistically significant instances ($p < 0.05$) out of the 69
instances of 10 location games and 44 instances of 20 location games of
$\tau$ for the combined Real-World datasets.} \label{tab:combined_sig}
\end{table}

\begin{table}[H] \centering \begin{tabular}{lccc} \toprule & 10
Locations     & 20 Locations     & All Games \\ \midrule Shortcut
Distance & 5.8\%  & 15.9\% & 9.7\%  \\ Re-routed Margin  & 42.0\% &
65.9\% & 51.3\% \\ Depot Distance    & 34.8\% & 38.6\% & 36.3\% \\
\midrule Moat Packing      & 42.0\% & 61.4\% & 49.6\% \\ Christofides   
  & 39.1\% & 56.8\% & 46.0\% \\ 60/40 Moat/Depot  & 42.0\% & 54.5\% &
46.9\% \\ \bottomrule
\end{tabular} \caption{Percentage of correct top elements of the Shapley
ordering identified by the respective proxy for the 69 games of size 10
and 44 games of size 20 for the combined Real-World data.}
\label{tab:combined_percent} \end{table}

Examining the change in performance of sophisticated proxies when moving
 from the Synthetic to Real-World scenarios, the average RMSE increases
from $\approx 0.075$ to $\approx 0.153$ while the average $\tau$
decreases from $\approx 0.63$ to $\approx 0.28$. Measuring RMSE, the 
degradation in performance of $\fracmoat$ is clearly the most sever.
Measuring ranking error via $\tau$,  $\fracmoat$ degrades more
gracefully compared to either $\fracchris$ or $\fracblend$.  
Measuring all proxy performances using RMSE,  $\fracshort$ is always
dominated by $\fracreroute$, which in turn is strictly dominated by the
sophisticated  proxies. It is worth noting that in Real-World scenarios
$\fracreroute$ strictly dominates all the other proxies in its ability
to  identify the most costly location. In that regard $\fracmoat$ is a
close second. Treating ranking error, Table~\ref{tab:combined_kt} shows
that $\fracreroute$ actually performs comparably with best sophisticated
proxy, $\fracmoat$, in terms of $\tau$. Table~\ref{tab:combined_sig}
shows that the Christofides proxy $\fracchris$ achieves statistically
significant values for $\tau$ in the largest number of scenarios. The
average ranking performance of $\fracmoat$ is relatively low, which
appears to be somewhat due to the discrepancy in the number of games of
size 10 and 20. We see clearly superior ranking performance from
$\fracchris$ for the larger games.
Repeating our observations for the synthetic corpus, in the Real-World
games the sophisticated  proxies have a greater percentage of
statistically significant  results for $\tau$. For a majority of the
instances, $\fracchris$ and $\fracmoat$ achieve a statistically
significant correlation with the ranking induced by $\fracshapley$.
Table~\ref{tab:synth_combined_compare} shows a side by side comparison
of the games with 20 locations for the Real-World and Synthetic data.
Moving from synthetic to Real-World we see the performance of
$\fracchris$ and $\fracblend$ noticably degrade, though they do continue
to achieving fairly low RMSE scores. Again, it is also worth noting that
all sophisticated proxies are also good and identifying the most costly
location.

Examining Real-World games with 20 locations,
Figures~\ref{fig:real_20_scatter_error} and
\ref{fig:synthetic_20_scatter_error} give the error scatter plots for
all  proxies as a function of allocation  according to $\fracshapley$.
The linear component to the error observed in
Figure~\ref{fig:synthetic_scatter_error} for Synthetic data remains
clear in Real-World scenarios. There is however a more uniform
distribution of errors among locations in the latter. This is evidenced
by the pillar like shapes for most of the plots; demonstrating that in
the Real-World data, many of the $\fracshapley$ allocations cluster
around a uniform allocation of around 5--8\%. Indeed, the observed tight
clustering of actual Shapley values explains the respectable 
performance of $\fracdist$ in the Real-World datasets. The much taller
shapes we see in  Figure~\ref{fig:real_20_scatter_error} compared to
Figure~\ref{fig:synthetic_20_scatter_error}  indicate that proxy errors
are more randomly distributed among Real-World locations, and that in
Real-World scenarios proxies make proportionately larger allocation
errors irrespective of the actual $\fracshapley$ allocation.

\begin{table} \centering \begin{tabular}{lcccc||cccc} \toprule &
\multicolumn{2}{c}{Synthetic} & \multicolumn{2}{c||}{Real-World} &
\multicolumn{2}{c}{Synthetic} & \multicolumn{2}{c}{Real-World}\\ & RMSE
& St. Dev. & RMSE & St.Dev. & $\tau$ & St. Dev. & $\tau$ & St. Dev. \\
\midrule Shortcut Distance                     & 0.2965 & 0.0543 &
0.3239 & 0.1064 & -0.0363 & 0.1358 & 0.0798 & 0.1891 \\ Re-routed Margin
                     & 0.1826 & 0.0442 & 0.2902 & 0.0934 & 0.3813 & 
0.1505 & 0.3476 & 0.1993 \\ Depot Distance                         &
0.0864 & 0.0182 & 0.0870 & 0.0303  & 0.5053 &  0.1464 & 0.1622 & 0.2313 
\\ \midrule Moat Packing                            & 0.0758 & 0.0174 &
0.1969 & 0.0883  & 0.5304 &  0.1180 & 0.3064 & 0.1710  \\ Christofides  
                            & 0.0622 & 0.0136 & 0.0863 & 0.0311 & 0.5965
 & 0.0999 & 0.3509  & 0.2258  \\ 60/40 Moat/Depot                     &
0.0529 & 0.0084  & 0.0765 & 0.0301 & 0.6690  & 0.1105 & 0.2531  & 0.2487
\\ \bottomrule
\end{tabular} \caption{Comparison of performance between Synthetic and
Real-World datasets for games with 20 locations.  There are 20 games in
the Synthetic corpus and 44 in the Real-World corpus. Root Mean Squared
Error (RMSE) and Standard Deviation (St. Dev.) are reported on the left
where lower is better. On the right average KT distance ($\tau$) and
Standard Deviation (St. Dev.) is reported where higher is better; $+1$
means the two lists are perfectly correlated and $-1$ means the two
lists are perfectly anti-correlated.} \label{tab:synth_combined_compare}
\end{table}

\begin{figure} \begin{minipage}[b]{0.49\linewidth} \centering
\includegraphics[height=6.7cm,
page=5]{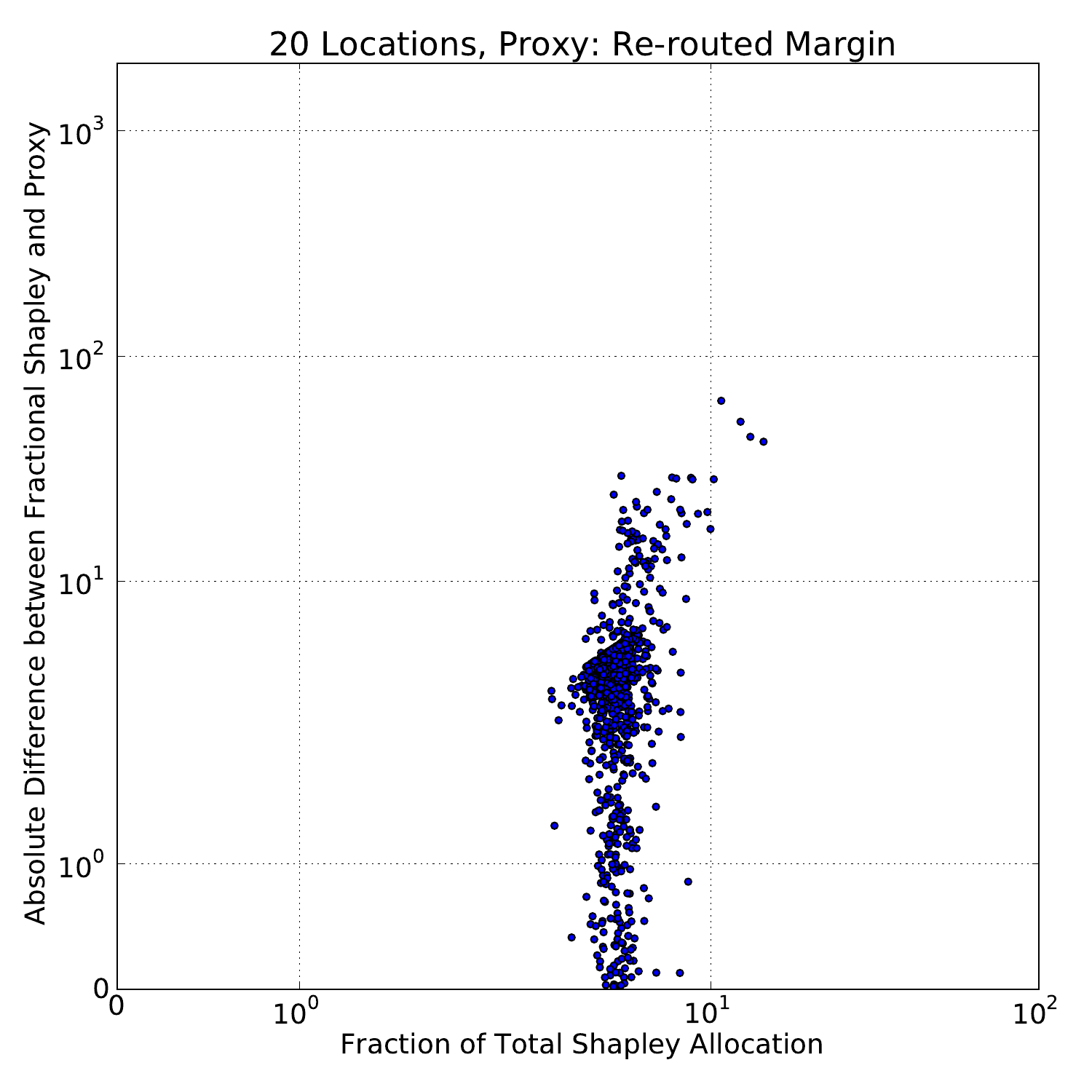}
\includegraphics[height=6.7cm,
page=8]{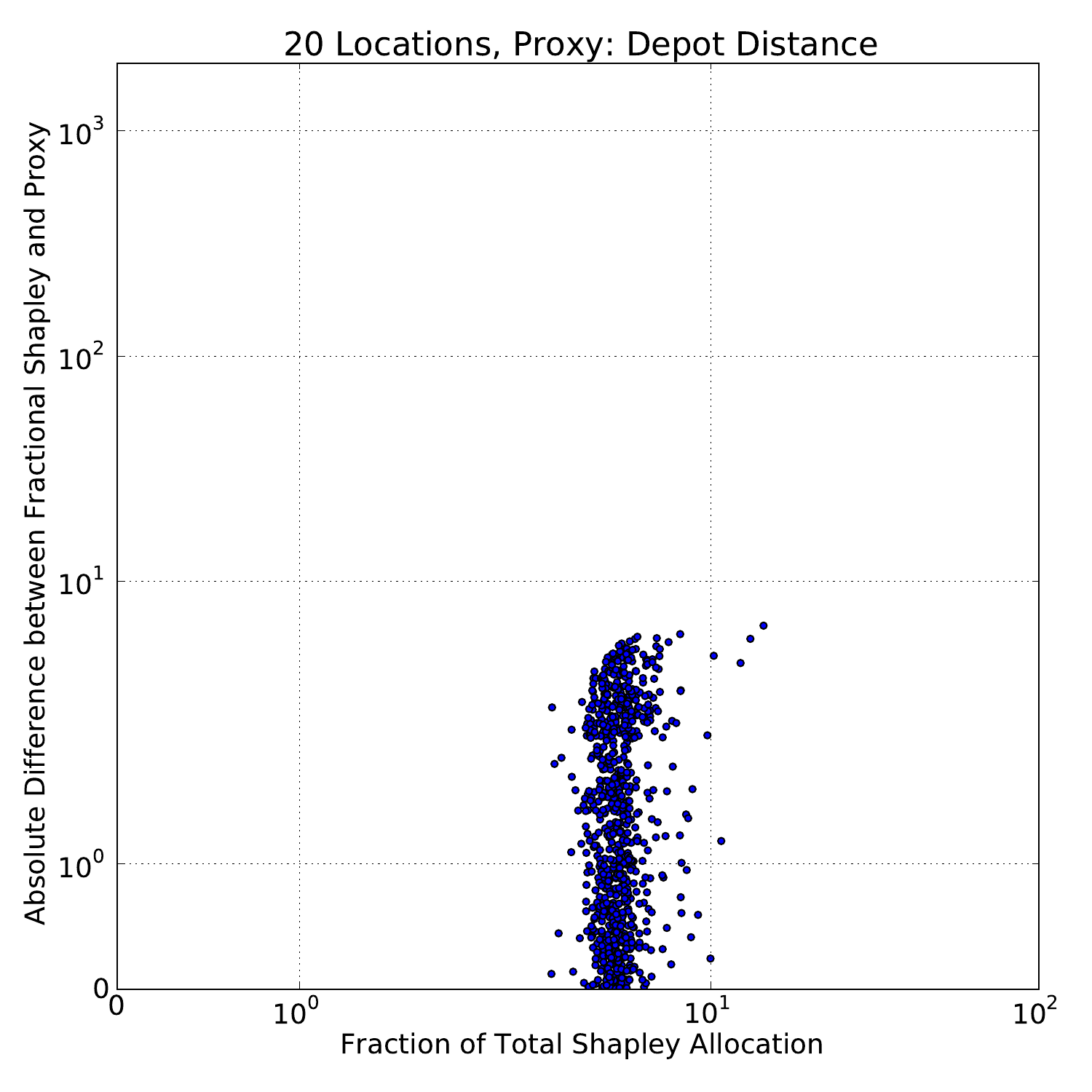}
\includegraphics[height=6.7cm,
page=14]{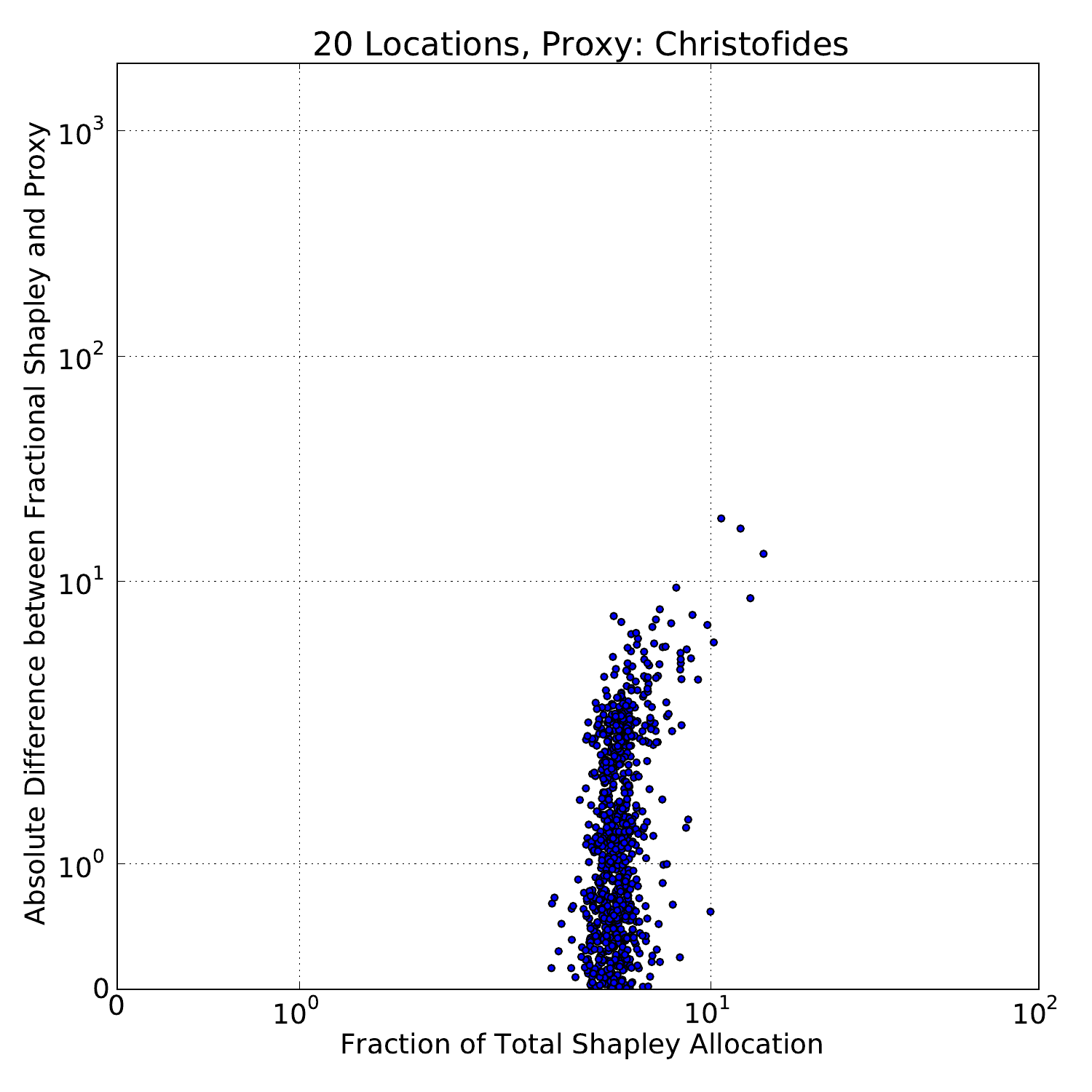} \end{minipage}
\begin{minipage}[b]{0.49\linewidth} \centering
\includegraphics[height=6.7cm,
page=2]{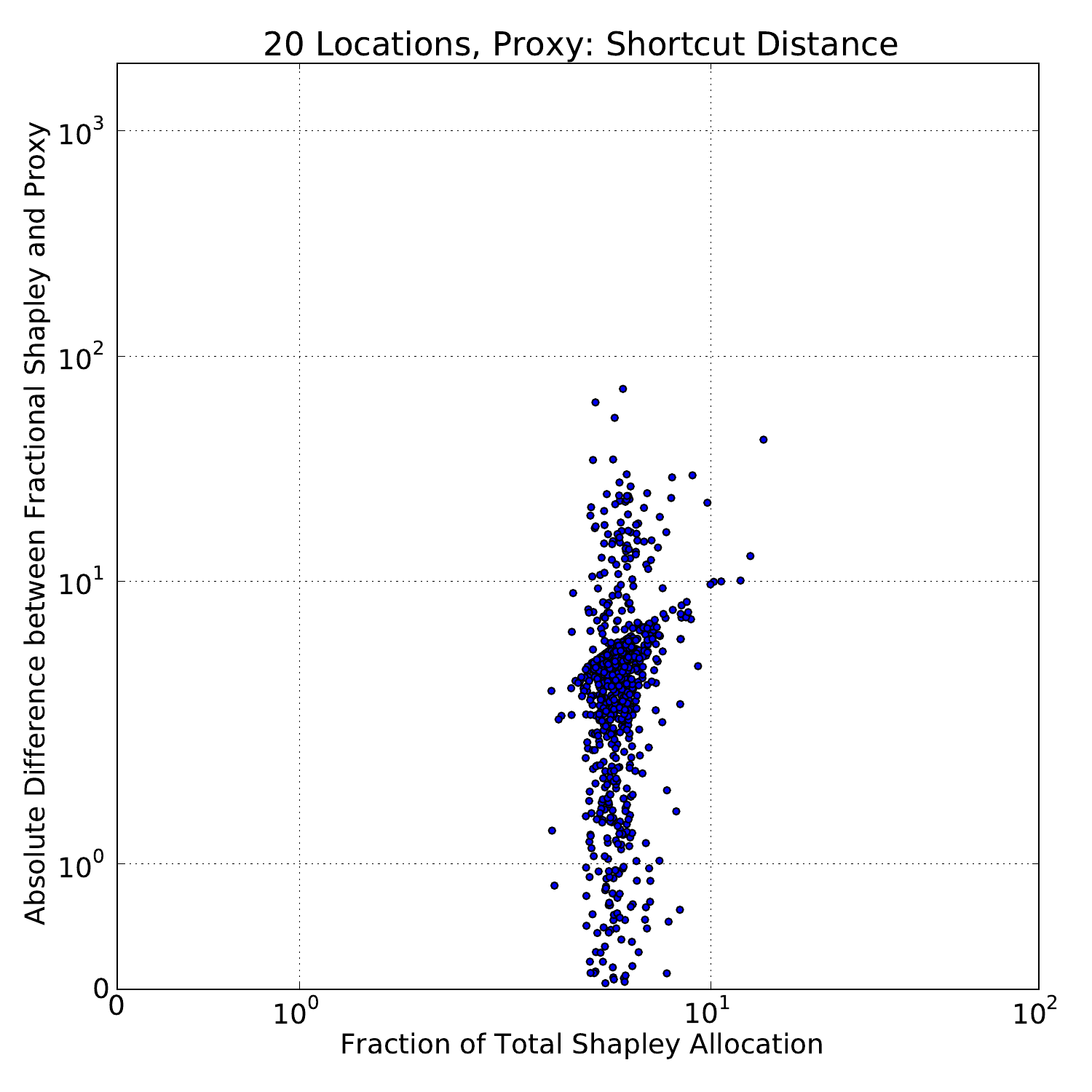}
\includegraphics[height=6.7cm,
page=11]{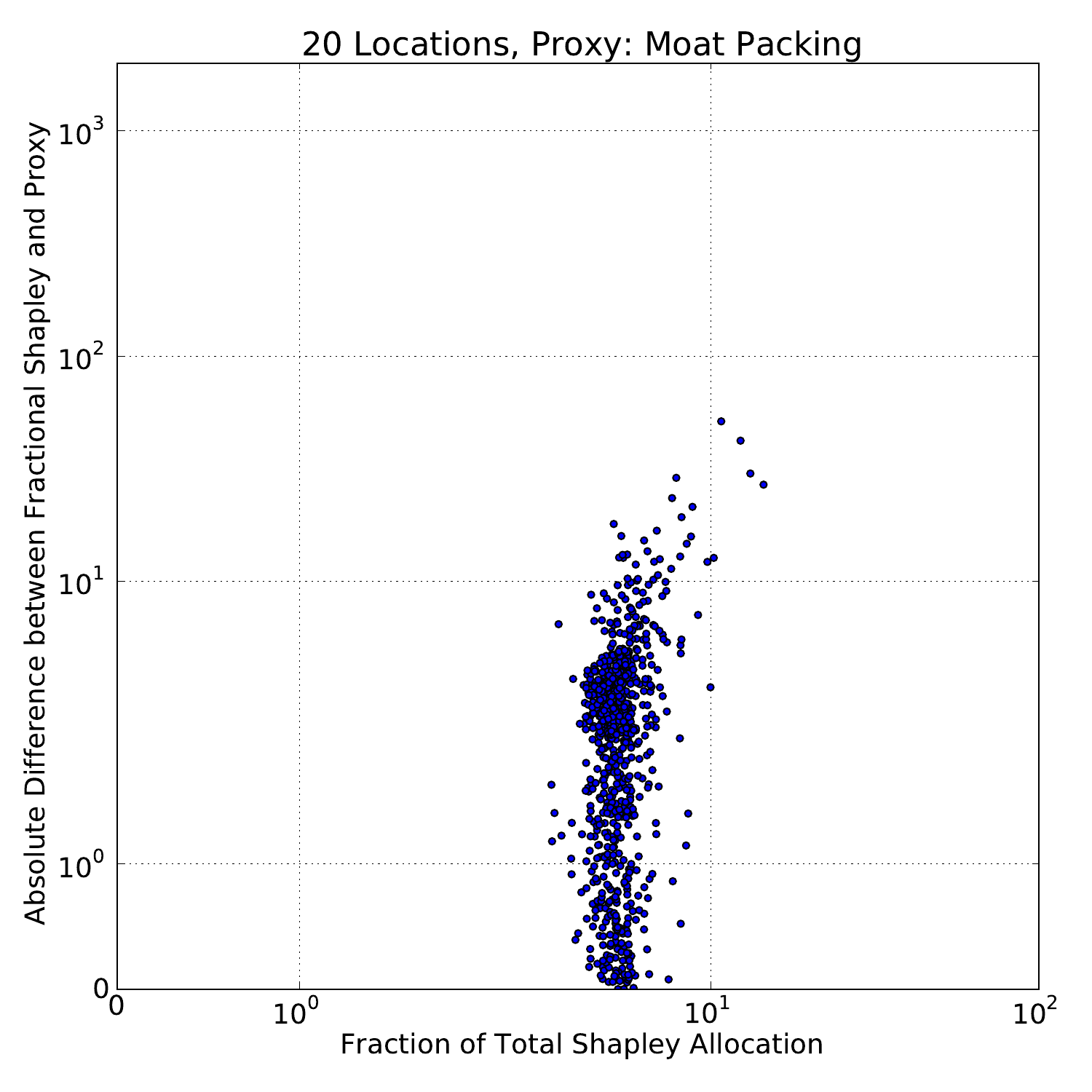}
\includegraphics[height=6.7cm,
page=17]{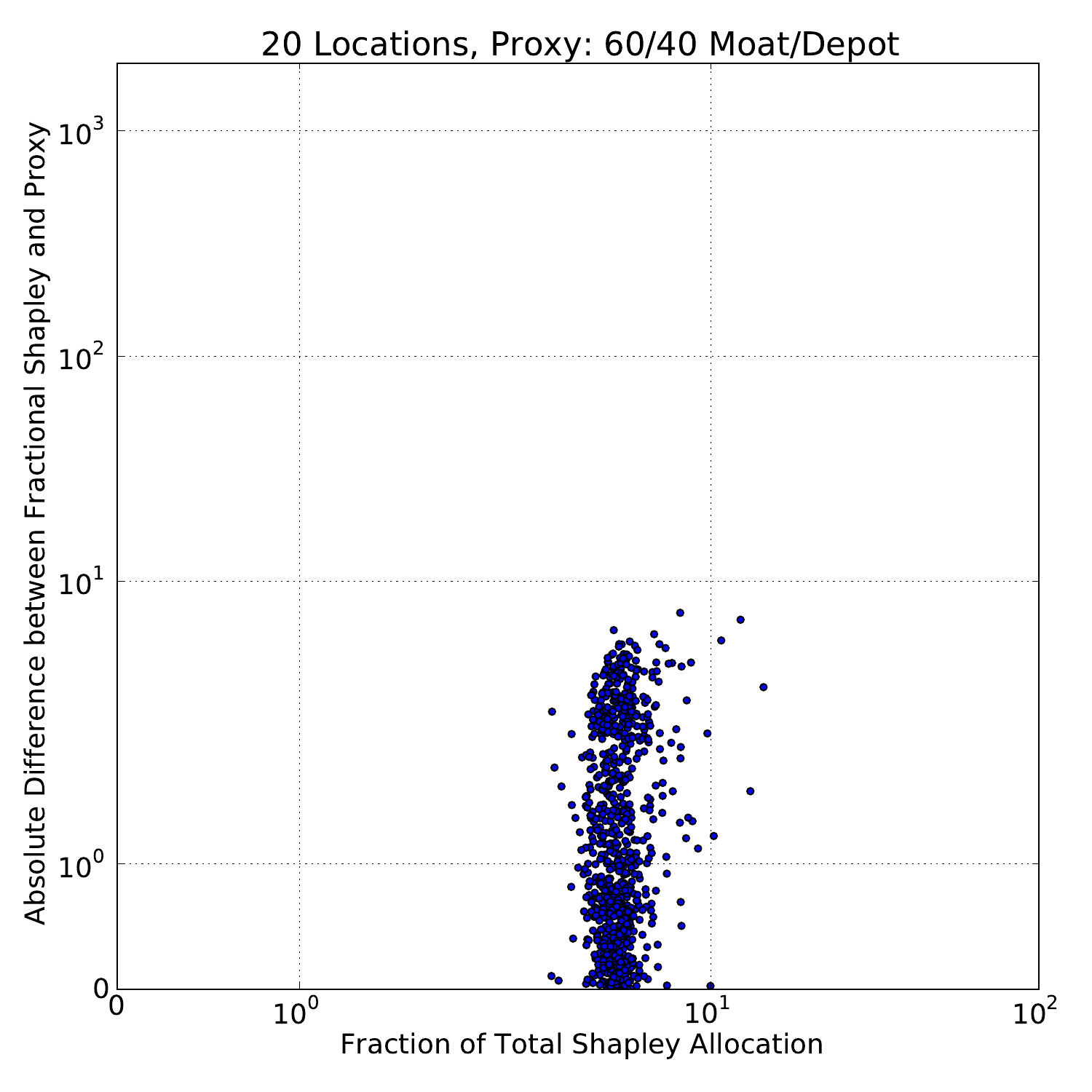} \end{minipage}
\caption{Absolute value of the difference between the $\fracshapley$ and
$\fracproxy$ plotted as a function of $\fracshapley$ for all 44 games in
the Real-World datasets with 20 locations.  Note that these are log-log
plots to highlight the spread of the data.}
\label{fig:real_20_scatter_error} \end{figure}

\begin{figure} \begin{minipage}[b]{0.49\linewidth} \centering
\includegraphics[height=6.7cm, page=50]{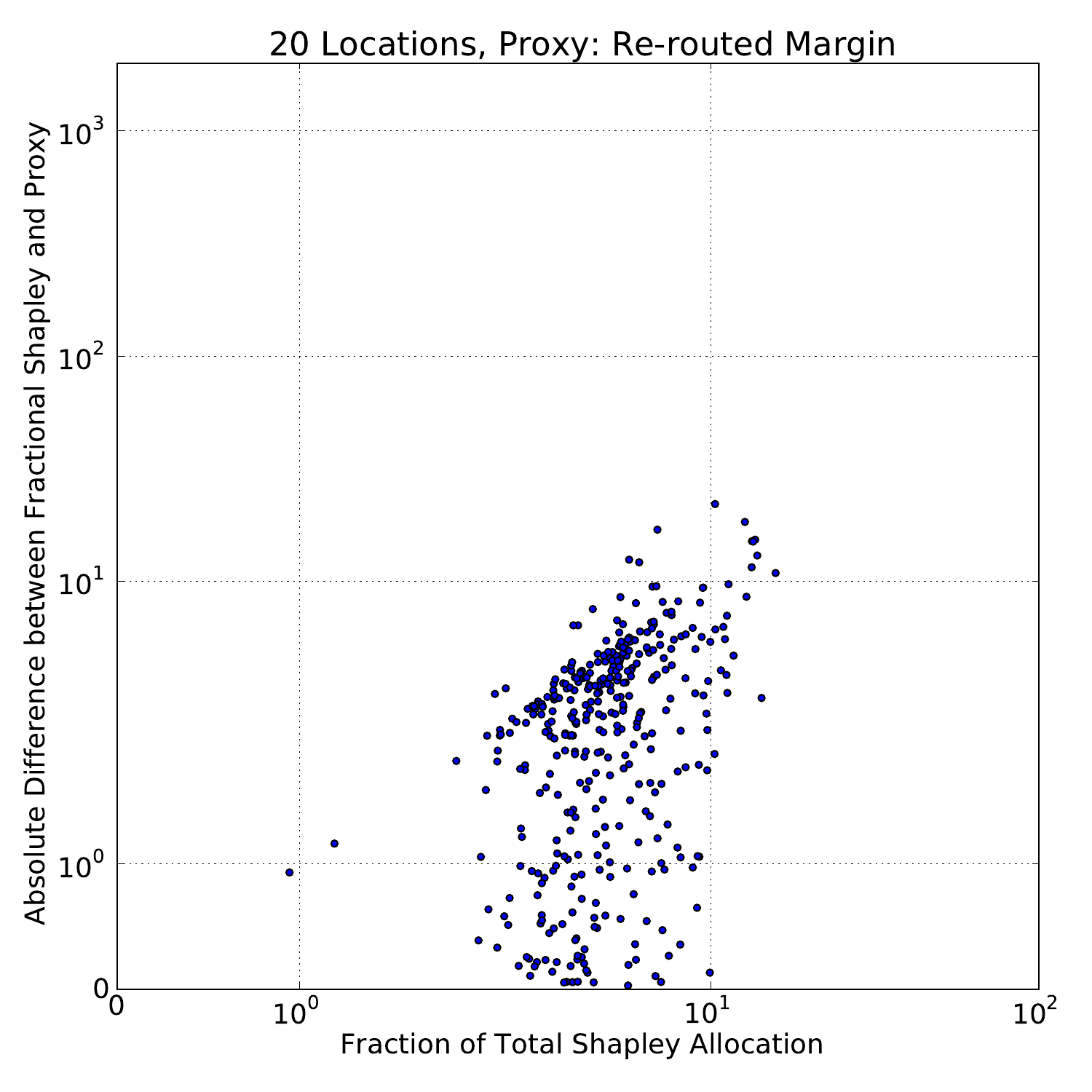}
\includegraphics[height=6.7cm, page=83]{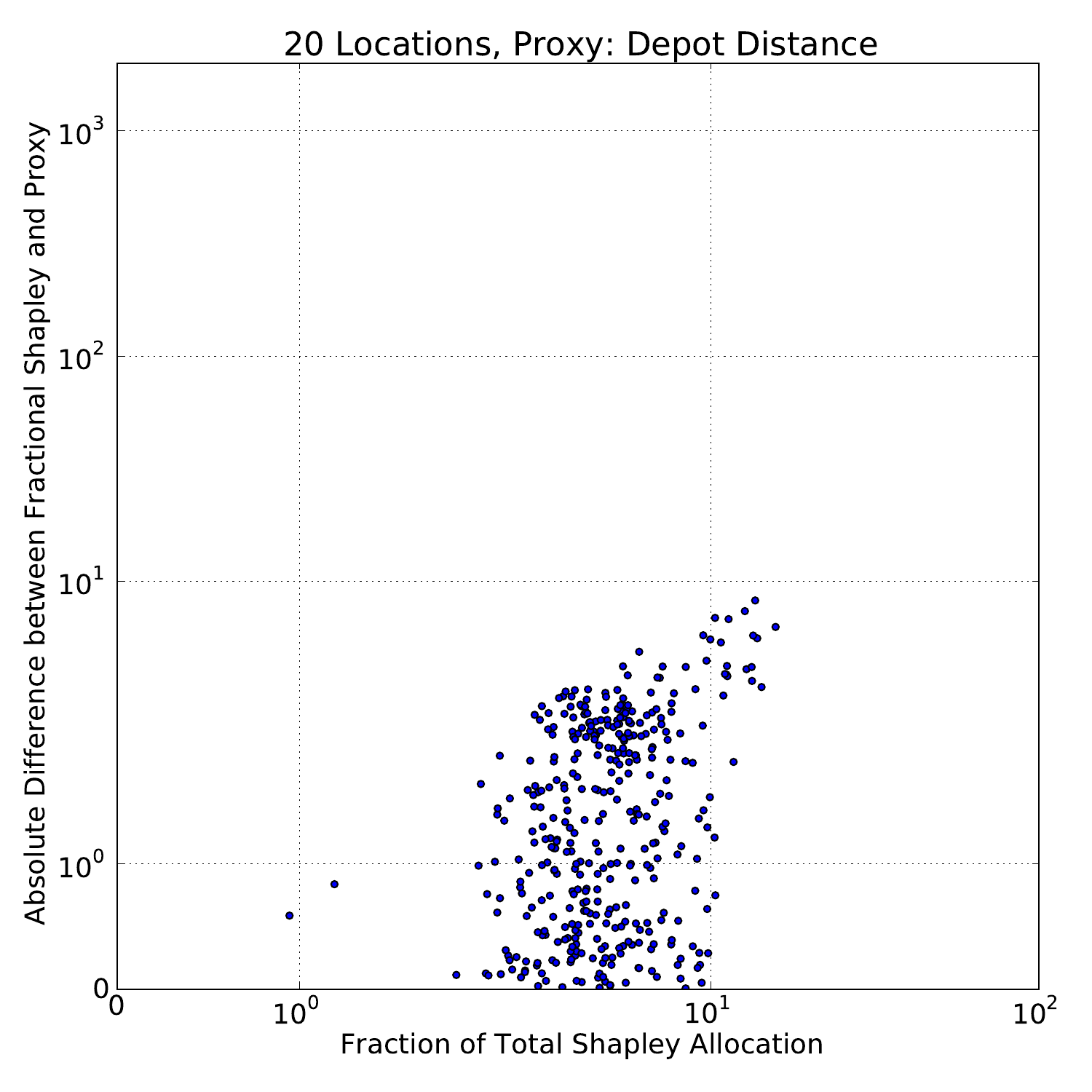}
\includegraphics[height=6.7cm,
page=149]{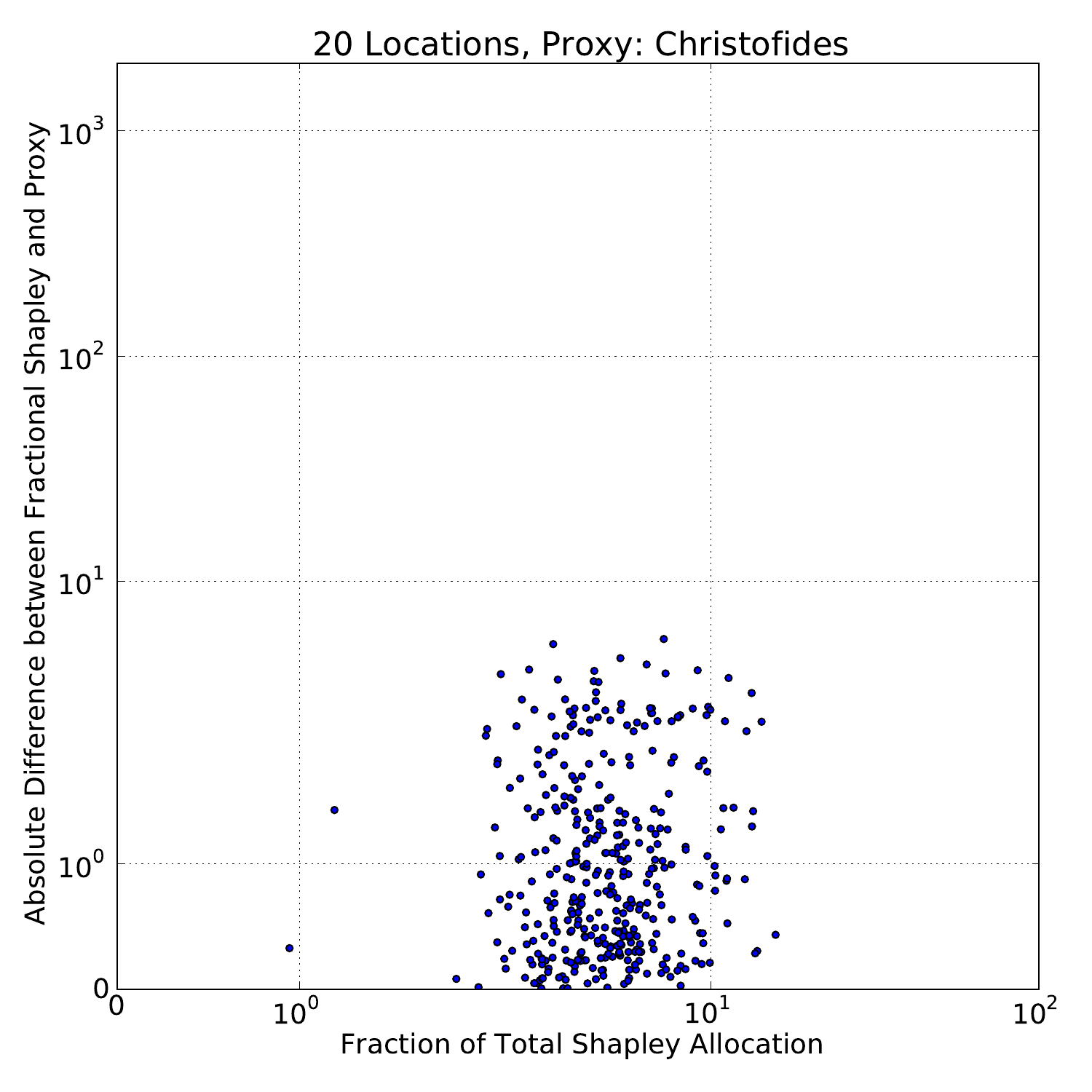} \end{minipage}
\begin{minipage}[b]{0.49\linewidth} \centering
\includegraphics[height=6.7cm, page=17]{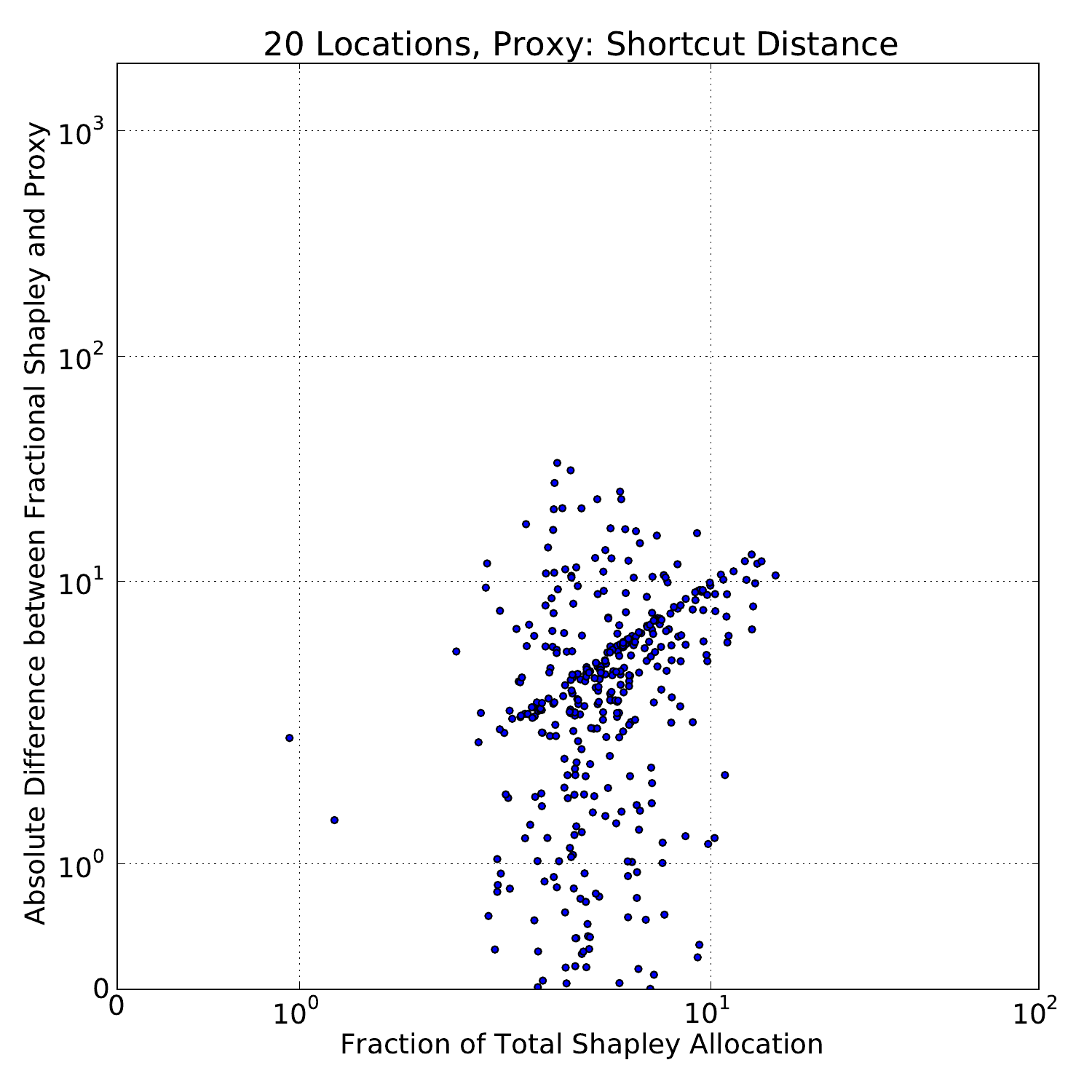}
\includegraphics[height=6.7cm,
page=116]{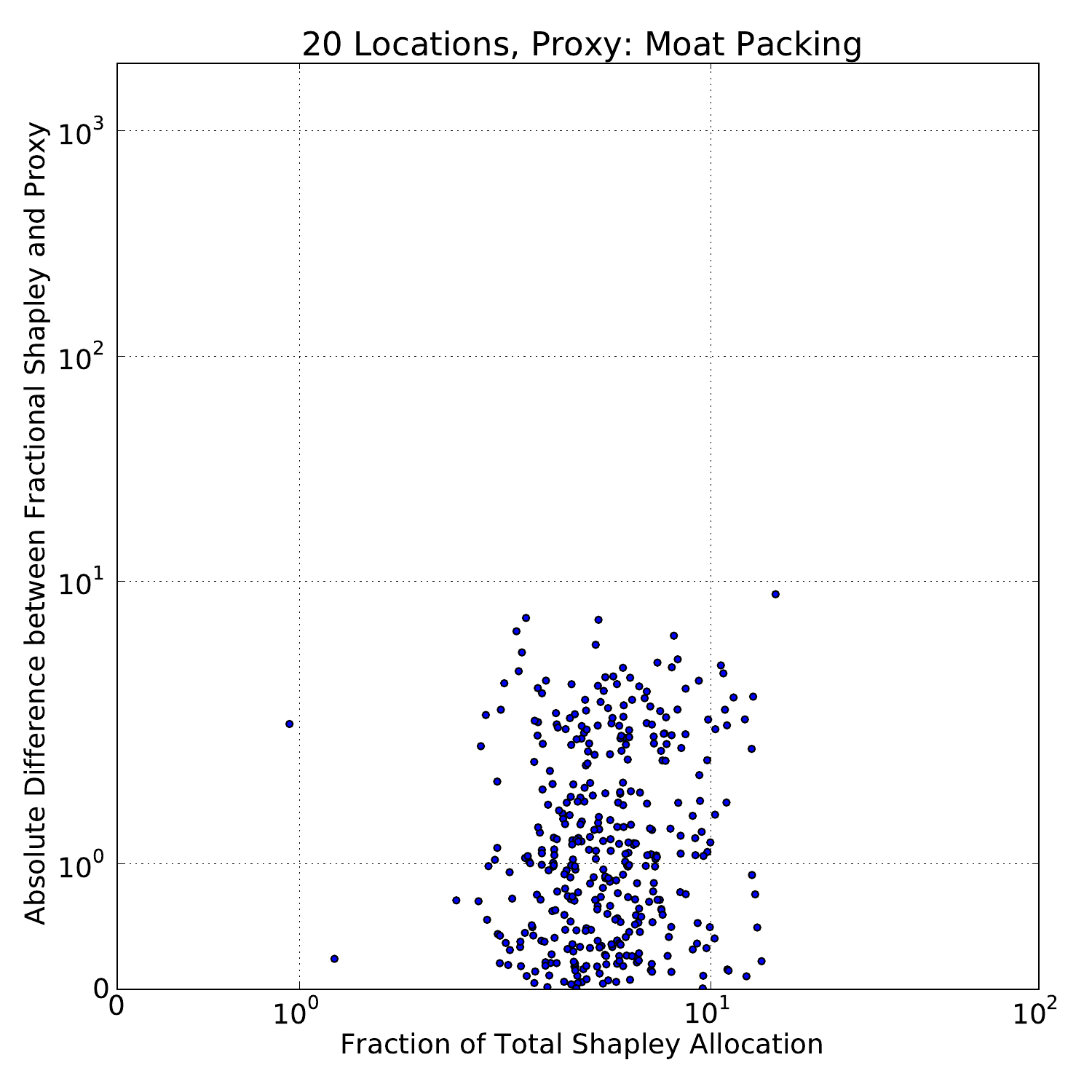}
\includegraphics[height=6.7cm,
page=182]{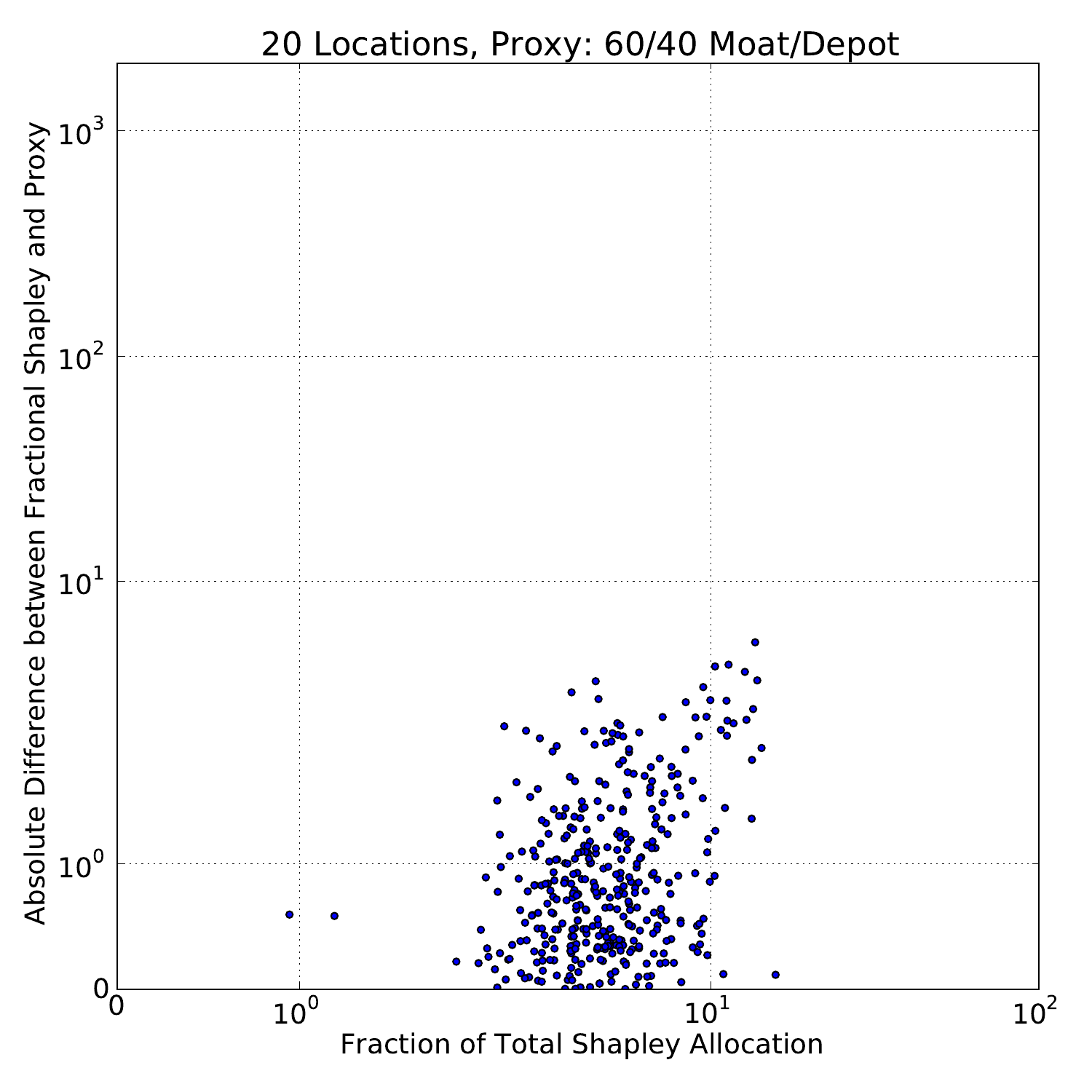} \end{minipage}
\caption{Absolute value of the difference between the $\fracshapley$ and
$\fracproxy$ plotted as a function of $\fracshapley$ for all games in
the Synthetic dataset with 20 locations.  Note that these are log-log
plots to highlight the spread of the data.}
\label{fig:synthetic_20_scatter_error} \end{figure}

\clearpage \section{Related Work}

The theory of cooperative games has a rich history in which various
solution concepts for allocating costs and other quantities have been
proposed~\cite{PeSu07a,Young94}.
In addition to the Shapley value we see allocation concepts given by the
core, the nucleolus and the bargaining set. Of those, the Shapley value
is considered the ``most important'' allocation scheme in cooperative
game theory~\cite{Wint02a}.
Application of the Shapley value spans well beyond transportation
setting.
For examples, the Shapley value has been applied in allocating the cost
of network infrastructure~\cite{Koster,MMP08a}, promoting collaboration
between agents~\cite{ZlRo94} by prescribing an allocation that
incentivises agents to collaborate in the completion of tasks, and as an
incentive compatible way to share departmental costs in corporations
\cite{Young85}. 
Considering applications in networks more broadly, use of the Shapley
value follows a general framework, where agents correspond to the nodes
(or edges) of a graph~\cite{curiel2008,Koster,MMP08a,TiDr86a,AzKe14a}.
Here the definition of the characteristic function depends on the
application domain, with proposed evaluations based on: (i)  the size of
maximum matching, (ii) network flow, (iii) the weight of a minimum
spanning tree, and (iv) the weight of a Hamiltonian
cycle~\cite{curiel2008,DeFa08a}.
Allocation concepts are not solely devised and employed for allocating
costs and other financial quantities. For example, the Shapley value has
been used directly to measure quantities indicating the {\em importance}
of agents in social networks~\cite{MoPa08}, and to measure the {\em
centrality} of nodes in networks \cite{MAS+13a}.
Another solution concept that has been used to gauge the importance of
agents is the \emph{Banzhaf value}~\cite{Banz64a}. The Banzhaf value is
defined for simple voting games -- i.e. cooperative games in which the
value of the coalition is either zero or one but the Banzhaf value of an
agent can suitably be extended to general cooperative games. However,
even within the context simple voting games, the Banzhaf value is more
suitable for measuring the influence of an agent and less suitable for
\emph{allocate} power between agents~\cite{FeMa98a}. Similarly, our
focus is to allocate costs, we focus on the Shapley value.

While solution concepts from the theory of transferable utility (TU)
cooperative games \cite{PeSu07a,CEW11a} have been used for allocations
of costs, the Shapley allocations have rarely received serious attention
in the transportation science literature. The associated computational
cost is prohibitively high for the general case, and consequently strong
notions of fairness  are often taken to be a secondary consideration. 
Though \emph{ApproShapley} is an FPRAS (fully polynomial-time randomized
approximation scheme) for computing the Shapley value if the game is
convex \cite{NowellSWW12}, this does not apply for the domain considered
in this work. Other prominent TU game solution concepts are
\emph{nucleolus} and \emph{core}. TSGs are introduced in
Potters~\cite{potters1992}, where in addition to describing that game,
the authors describe a variety of game known as the {\em routing
game}.\footnote{Note the cited Potters et al. journal publication
extends a technical report introducing the game as early as 1987.} 
you do not include this footnote, the Tamir citation is anachronistic,
we have the space, please leave it! For the latter an auxiliary
constraint forces locations to be visited, in any coalition, in the 
order they are traversed by a specific tour. Assuming that the tour
corresponds to the optimal for the underlying  TSP, then the game has a
non-empty core. \citeauthor{derks:1997}~\citeyear{derks:1997} presented
a quadratic-time  procedure for computing a core allocation of the
routing game. They also characterize suboptimal tours  that specify
routing games with non-empty cores. It should be noted that there are no
known tractable procedures to compute a tour which guarantees the core
is non-empty for the routing game. Conditions for the non-emptiness of
the core in TSGs  were further developed in
Tamir~\shortcite{Tamir198931}. We have already noted that
\citeauthor{faigle1998}~\citeyear{faigle1998} developed a procedure to
calculate a multiplicative $\epsilon$-core allocation for Euclidean
TSGs. \citeauthor{yengin2012}~\citeyear{yengin2012} develop a notion of
a {\em fixed route} game with {\em appointments} which admits a
tractable procedure for computing Shapley values. That model is not
suitable for typical scenarios that involve the delivery of goods to
locations from a depot. TU concepts in TSGs and routing games are
developed for a practical gas delivery application in Engevall et
al.~\shortcite{engevall:1998}.

Considering vehicle routing problems and transportation settings more
generally, G\"{o}the-Lundgren et al.~\shortcite{lundgren1996} develop a
column generation  procedure to calculate the \emph{nucleolus} of a
homogeneous vehicle routing problem -- i.e. all  vehicles are
equivalent. In doing so they develop a procedure to determine if the
core of  that vehicle routing game is empty. Engevall et
al.~\shortcite{Engevall:2004} extend  that work for a very practical
setting of distributing gas using a {\em heterogeneous} fleet of
vehicles.
More recently \"{O}zener et al.~\shortcite{Ozener:2013} examine a number
of solution  concepts---including allocations derived according to the
nested moat-packing of~\citeauthor{faigle1998}~\citeyear{faigle1998}, 
and a highly bespoke approximation of the Shapley allocation---in
deriving cost allocations for  real-world {\em inventory routing}
problems. They show that TU game allocations, especially 
core/duality-based allocations, have significant advantages over the
existing cost  allocations which their industrial client was using.

\section{Conclusions and Future Work}

We studied the problem of fairly apportioning costs in transportation
scenarios, specifically TSGs. The Shapley value is a highly appealing
division concept for this task.  Since it cannot be evaluated in
reasonable time, we considered a number of proxies for the Shapley
value. We examined proxy performance both in terms of approximating the
Shapley value and the ranking of locations induced by the Shapley value.
The stand-out proxies with respect to both measures are $\fracchris$ and
$\fracblend$, a mixture of $\fracdist$ and $\fracmoat$.  These proxies
can be computed in reasonable time, and exhibit good properties in both
synthetic Euclidean games and real-world transportation scenarios.

Extensions of our work should develop proxies for the more general
setting of vehicle routing games, to quantify the importance of agent
synergies that are unique to the multi-vehicle model.
The transport companies we interact with further seek to understand the
impact of time windows (both the duration and position of allowable
service times), and the effect of delivery frequency on allocated costs.
Thus, a highly motivated and rich variety of problems is available to be
considered for future work.  
Additionally, future research should consider weighted Shapley values
for situations where some coalitions (and therefore margins) are more
likely to occur than others. Formal approximation ratios, to complement
the strong empirical evidence we obtained using sophisticated proxies
should also be the subject of future research. There also remains the
need for formal studies which employ proxy allocations to inform
solutions to hard optimisation problems in transportation domains. 
Finally, scaling to larger transportation scenarios shall require new
methods which avoid treating all individual agents in a large monolithic
evaluation of the allocation of costs. An approximation strategy that
may be fruitful here was proposed in~\cite{SCCP14a}, where agents are
partitioned into groups and each agent in the group is assigned the same
Shapley value.  Measures for clustering transport agents may appeal to
proximity to pose useful aggregations of agents.

%
%
%
\section*{Acknowledgments} NICTA is funded by the Australian Government
through the Department of Communications and the Australian Research
Council through the ICT Centre of Excellence Program.
Casey Cahan was supported by an Summer Research Scholarship at The Australian National University.
Toby Walsh also receives support from the Asian Office of Aerospace Research and Development 
(AOARD 124056) and the German Federal Ministry for Education and Research through the
Alexander von Humboldt Foundation.

We would like to thank Stefano Moretti and Patrice Perny from LIP6;
Hossein Azari Soufiani from Harvard University; David Rey and Vinayak
Dixit from the rCiti Project at the University of New South Wales School
of Civil and Environmental Engineering, and the reviewers and attendees
of the 5th Workshop on Cooperative Games in MultiAgent Systems
(CoopMAS-2014) for their helpful feedback and comments on early version
of this work.

\bibliographystyle{theapa} \bibliography{tsg}

\clearpage \appendix \section{Definitions and
Notations}\label{appendix:stats}

Our work makes use of statistical measures to compare the proxies, we
provide a brief overview here and refer the reader to the textbook by
Cordor and Foreman \cite{corder-foreman:b:nonparametric} for a more
complete treatment.  Note that $abs()$ is the absolute value of the
quantity $()$. Writing $\hat{x}$ to for the average of a set $\{x_1,
\ldots, x_n\}$, the standard deviation (St.Dev) of that set is: $$St.Dev
= \sqrt{ \frac{\sum_{i=1}^{n}abs(x_i - \hat{x})^2}{n}}.$$

The Shapley value of the $i$th location, divided by the sum of the
Shapley values for all locations is written $\fracshapley_i$, and to
denote a proxy (as in the main document) we write $\fracproxy_i$. The
\emph{absolute percent difference} between a value and its proxy is
$$\frac{abs(\fracshapley_i - \fracproxy_i)}{\fracshapley_i} \times
100.$$

Our study of proxy accuracy also makes use of \emph{root mean squared
errors} RMSE, a common metric to express the error made over a number of
predictions. Taking a TSG with locations $L$, the RMSE between a Shapley
allocation $\phi^{SV}$ and a proxy $\phi^{Proxy}$  is: $$ RMSE =\sqrt{
\frac{\sum_{i \in L} (\fracshapley_i - \fracproxy_i)^2}{|L|}}. $$

Shapley values can be used to rank/order locations, from least to most
costly. Our work studies the accuracy of proxies in that task using
Kendall's tau distance (KT distance), and also the KT rank correlation
coefficient, $\tau$. The KT distance measures the amount of disagreement
between two rankings. We study the ranking of locations induced by the
Shapley value and its proxies. In the case that the rankings correspond
to total orders\footnote{Which is always the case in our experiments.}
the KT distance is called the {\em bubble-sort} distance, and is equal
to the number of bubble-sort swap operations necessary to make two lists
agree. It is customary to report KT distance as a normalized value
(correlation coefficient) between 1 and -1, where 1 means that two lists
are perfectly correlated (equal) and -1 means that two lists are
perfectly anti-correlated (they are equal if one list is reversed).

In detail, let  $X$ and $Y$ be two partial orders over a set of items.
If $a \gtrless b \in X \cap Y$ then we say $X$ and $Y$ are
\emph{concordant} on $(a,b)$. If $a = b \in X \cup Y$ then we say there
is a \emph{tie}, and otherwise $(a,b)$ is \emph{discordant}. Where $M$
is the number of concordant pairs,  $N$ discordant pairs, $T$ ties
exclusively in $X$, $U$ ties exclusively in $Y$, the normalised KT
distance $\tau$ between $X$ and $Y$ is: $$\tau = \frac{M - N}{\sqrt{ (M
+ N + T) \times (M + N + U)}}$$

Our analysis makes use of the significance, or $p$-value of a KT
statistic. The $p$-value is computed using a two-tailed $t$-test where
the null hypothesis is that there is no correlation between orderings
($\tau = 0$). This means that if we take our significance threshold
$\alpha = 0.05$, as is common in the scientific literature, we can
reject the null hypothesis when $p \leq 0.05$.  The interpretation of
this statistic is that when $p \geq 0.05$ we fail to reject the null
hypothesis. A $p$-value $ \leq 0.05$ is a statistically significant
result, meaning it is unlikely that two random, uncorrelated lists would
show such a high degree of correlation.

%
\section{Synthetic Data}\label{appendix:synthetic}

For each $i\in[4,\ldots,35]$, we generate $20$ instances of the
Euclidean TSG with $i$ locations occurring uniformly at random in a
square  of dimension $1,000$. The horizontal and vertical coordinates of
the locations are represented using 32-bit floating point numbers.

Tables \ref{tab:synth_rmse1} to \ref{tab:synth_ptop} represent a
selected amount of raw data from our experiments.  The first two tables
show the RMSE and $\tau$ for various numbers of locations. The
subsequent two tables show the median and maximum (least significant)
$p$ values for the $\tau$ statistic.  Finally, Table
\ref{tab:synth_ptop} gives the percentage of correctly identified most
costly locations.  The tables in subsequent sections are the same for
the Real-World data.

\begin{table}[H] \centering \begin{tabular}{lcccccccc} \toprule
&\multicolumn{2}{c}{5 Locations} & \multicolumn{2}{c}{10 Locations} &
\multicolumn{2}{c}{15 Locations} & \multicolumn{2}{c}{20 Locations}\\ &
RMSE & St. Dev. & RMSE & St.Dev.& RMSE & St. Dev. & RMSE & St. Dev. \\
\midrule Shortcut Distance                     & 0.4948  & 0.1379 &
0.3826 & 0.0954 & 0.3603 & 0.0806 & 0.2965 & 0.0543  \\ Re-routed Margin
                     & 0.3951  & 0.1104 & 0.2630 & 0.0594 & 0.2061 &
0.0546 & 0.1826 & 0.0442  \\ Depot Distance                         &
0.1198  & 0.0579 & 0.0994 & 0.0325 & 0.1050 & 0.0263 & 0.0864 & 0.0182 
\\ \midrule Moat Packing                            & 0.1667  & 0.0487 &
0.0879 & 0.0278 & 0.0726 & 0.0252 & 0.0758 & 0.0174   \\ Christofides   
                      	& 0.0690 & 0.0292 & 0.0640 &  0.0268 & 0.0708
& 0.0229 & 0.0622 & 0.0136  \\ 60/40 Moat/Depot                     &
0.1136  & 0.0483 & 0.0538 & 0.0146 & 0.0575 & 0.0115 & 0.0529 & 0.0084 
\\ \bottomrule \end{tabular} \caption{Root Mean Squared Error (RMSE) and
Standard Deviation (St. Dev.) for the Synthetic data for games with
between 5 and 20 locations. Lower is better.} \label{tab:synth_rmse1}
\end{table}

\begin{table}[H] \centering \begin{tabular}{lcccccccc} \toprule
&\multicolumn{2}{c}{25 Locations} & \multicolumn{2}{c}{30 Locations} &
\multicolumn{2}{c}{35 Locations} & \multicolumn{2}{c}{All Games}\\ &
RMSE & St. Dev. & RMSE & St.Dev.& RMSE & St. Dev. & RMSE & St. Dev. \\
\midrule Shortcut Distance                     & 0.2830 & 0.0427 &
0.2553 & 0.0781 & 0.2437 & 0.0390 & 0.3309       & 0.0754     \\
Re-routed Margin                      & 0.1763 & 0.0371 & 0.1585 &
0.0455 & 0.1487 & 0.0238 & 0.2186       & 0.0536     \\ Depot Distance  
                      & 0.0843 & 0.0145 & 0.0827 & 0.0185 & 0.0771 &
0.0135 & 0.0935       & 0.0259     \\ \midrule Moat Packing             
              & 0.0679 & 0.0146 & 0.0627 & 0.0134 & 0.0576 & 0.0092 &
0.0845       & 0.0223     \\ Christofides                          	&
0.0610 & 0.0163 & 0.0584 & 0.0193 & 0.0568 & 0.0124 & 0.0632       &
0.0201     \\ 60/40 Moat/Depot                     & 0.0514 & 0.0095 &
0.0555 & 0.0125 & 0.0539 & 0.0086 & 0.0627       & 0.0162    \\
\bottomrule \end{tabular} \caption{Root Mean Squared Error (RMSE) and
Standard Deviation (St. Dev.) for the Synthetic data for games with
between 25 and 35 locations, as well as an average over all games. Lower
is better.} \label{tab:synth_rmse2} \end{table}

\begin{table}[H] \centering \begin{tabular}{lcccccccc} \toprule
&\multicolumn{2}{c}{5 Locations} & \multicolumn{2}{c}{10 Locations} &
\multicolumn{2}{c}{15 Locations} & \multicolumn{2}{c}{20 Locations}\\ &
$\tau$ & St. Dev. & $\tau$ & St. Dev.& $\tau$ & St. Dev. & $\tau$ & St.
Dev.  \\ \midrule Shortcut Distance                      & -0.0333 &
0.5153 & 0.0361 & 0.2554 & -0.0220 & 0.1332 & -0.0363 & 0.1358 \\
Re-routed Margin                      & 0.6833  & 0.4010 & 0.5333 &
0.1453 & 0.4297 &  0.1598 & 0.3813 &  0.1505   \\ Depot Distance        
                & 0.6500  & 0.3069 & 0.5833 &  0.1422 & 0.5440 &  0.1311
& 0.5053 &  0.1464  \\ \midrule Moat Packing                         
	& 0.6167  & 0.4628 & 0.6139 & 0.1952 & 0.6055 &  0.1126 & 0.5304 & 
0.1180 \\ Christofides                          	& 0.8000  & 0.2667 &
0.7500 & 0.1770 & 0.6407 &  0.1819 & 0.5965  & 0.0999 \\ 60/40
Moat/Depot 		       & 0.7000   &0.3636 & 0.8056 & 0.0986 & 0.7374 
& 0.0794 & 0.6690  & 0.1105   \\ \bottomrule \end{tabular}
\caption{Average KT distance ($\tau$) and Standard Deviation (St. Dev.)
for the Synthetic data for games with between 5 and 20 locations. Higher
is better; $+1$ means the two lists are perfectly correlated and $-1$
means the two lists are perfectly anti-correlated.} \end{table}

\begin{table}[H] \centering \begin{tabular}{lcccccccc} \toprule
&\multicolumn{2}{c}{25 Locations} & \multicolumn{2}{c}{30 Locations} &
\multicolumn{2}{c}{35 Locations} & \multicolumn{2}{c}{All Games}\\ &
$\tau$ & St. Dev. & $\tau$ & St. Dev.& $\tau$ & St. Dev. & $\tau$ & St.
Dev.  \\ \midrule Shortcut Distance                      & 0.0174 &
0.1093 & 0.0212 & 0.1230 & -0.0132 & 0.0859 & -0.0043    & 0.1940     \\
Re-routed Margin                      & 0.3449 & 0.1526 & 0.3406 &
0.0940 & 0.3487  & 0.0884 & 0.4374     & 0.1702     \\ Depot Distance   
                     & 0.5297 & 0.1190 & 0.4914 & 0.0924 & 0.4267 & 
0.1206 & 0.5329     & 0.1512     \\ \midrule Moat Packing               
          	& 0.5315 & 0.1020 & 0.5030 & 0.0857 & 0.4938  & 0.0905 &
0.5564     & 0.1667     \\ Christofides                          	&
0.6033 &  0.0800 & 0.6288 & 0.0736 & 0.5601  & 0.0900 & 0.6542     &
0.1384     \\ 60/40 Moat/Depot 		       & 0.6848  & 0.0802 &
0.6266 & 0.0765 & 0.5797  & 0.0812 & 0.6862     & 0.1271    \\
\bottomrule \end{tabular} \caption{Average KT distance ($\tau$) and
Standard Deviation (St. Dev.) for the Synthetic data for games with
between 25 and 35 locations, as well as an average over all games.
Higher is better; $+1$ means the two lists are perfectly correlated and
$-1$ means the two lists are perfectly anti-correlated.} \end{table}

\begin{table}[H] \centering \begin{tabular}{lccccccc} \toprule
&\multicolumn{7}{c}{Number of Locations}\\ & 5      & 10     & 15     &
20     & 25     & 30     & 35     			   \\ \midrule Shortcut
Distance              & 0.4969 & 0.5316 & 0.5470 & 0.6763 & 0.5516 &
0.5495 & 0.6253  \\ Re-routed Margin               & 0.1079 & 0.0371 &
0.0328 & 0.0359 & 0.0298 & 0.0079 & 0.0023  \\ Depot Distance           
      & 0.1742 & 0.0218 & 0.0087 & 0.0019 & 0.0002 & 0.0001 & 0.0012  \\
\midrule Moat Packing                     & 0.1742 & 0.0123 & 0.0022 &
0.0021 & 0.0004 & 0.0002 & 0.0000  \\ Christofides                      
& 0.0415 & 0.0035 & 0.0009 & 0.0002 & 0.0001 & 0.0000 & 0.0000  \\ 60/40
Moat/Depot 		& 0.1079 & 0.0018 & 0.0001 & 0.0000 & 0.0000 & 0.0000
& 0.0000 \\ \bottomrule \end{tabular} \caption{Median $p$ values out of
20 games per number of locations of $\tau$ for the Synthetic data. 
Lower is better, $p < 0.05$ required for statistical significance.}
\end{table}

\begin{table}[H] \centering \begin{tabular}{lccccccc} \toprule
&\multicolumn{7}{c}{Number of Locations}\\ & 5      & 10     & 15     &
20     & 25     & 30     & 35     \\ \midrule

Shortcut Distance                     &1.0000 & 1.0000 & 0.8695 & 0.9164
& 0.9604 & 0.9402 & 0.9882 \\ Re-routed Margin                    
&1.0000 & 0.5316 & 0.6222 & 0.5520 & 0.5516 & 0.3294 & 0.1680 \\ Depot
Distance                        & 1.0000 & 0.2109 & 0.1124 & 0.1955 &
0.0594 & 0.0468 & 0.1196 \\ \midrule Moat Packing                       
    &1.0000 & 1.0000 & 0.1394 & 0.0637 & 0.0197 & 0.0244 & 0.0067 \\
Christofides                          	&1.0000 & 0.4042 & 0.1394 &
0.0191 & 0.0007 & 0.0002 & 0.0029 \\ 60/40 Moat/Depot 		      &
0.4969 & 0.0218 & 0.0037 & 0.0107 & 0.0007 & 0.0006 & 0.0008 \\
\bottomrule \end{tabular} \caption{Maximum $p$ values out of 20 games
per number of locations of $\tau$ for the Synthetic data.  Lower is
better, $p < 0.05$ required for statistical significance.} \end{table}

\begin{table}[H] \centering \begin{tabular}{lccccccc||c} \toprule
&\multicolumn{7}{c||}{Number of Locations} & \\ & 5      & 10     & 15  
  & 20     & 25     & 30     & 35     &All Games \\ \midrule Shortcut
Distance                     & 35.0\% & 20.0\% & 0.0\%  & 5.0\%  & 0.0\%
 & 5.0\%  & 5.0\%  & 10.0\%          \\ Re-routed Margin                
     & 85.0\% & 90.0\% & 65.0\% & 65.0\% & 70.0\% & 50.0\% & 65.0\% &
70.0\%          \\ Depot Distance                        & 75.0\% &
25.0\% & 30.0\% & 45.0\% & 20.0\% & 35.0\% & 40.0\% & 38.6\%          \\
\midrule Moat Packing                          & 65.0\% & 80.0\% &
75.0\% & 50.0\% & 55.0\% & 50.0\% & 60.0\% & 62.1\%          \\
Christofides                          & 85.0\% & 75.0\% & 60.0\% &
45.0\% & 80.0\% & 65.0\% & 55.0\% & 66.4\%          \\ 60/40 Moat/Depot
& 70.0\% & 75.0\% & 75.0\% & 50.0\% & 65.0\% & 60.0\% & 55.0\% & 64.3\%
\\ \bottomrule \end{tabular} \caption{Percentage of correct top elements
of the Shapley ordering identified by the respective proxy for the 20
instances per game size for the Synthetic data.} \label{tab:synth_ptop}
\end{table}

\begin{figure} \begin{minipage}[b]{0.49\linewidth} \centering
\includegraphics[height=6.7cm, page=66]{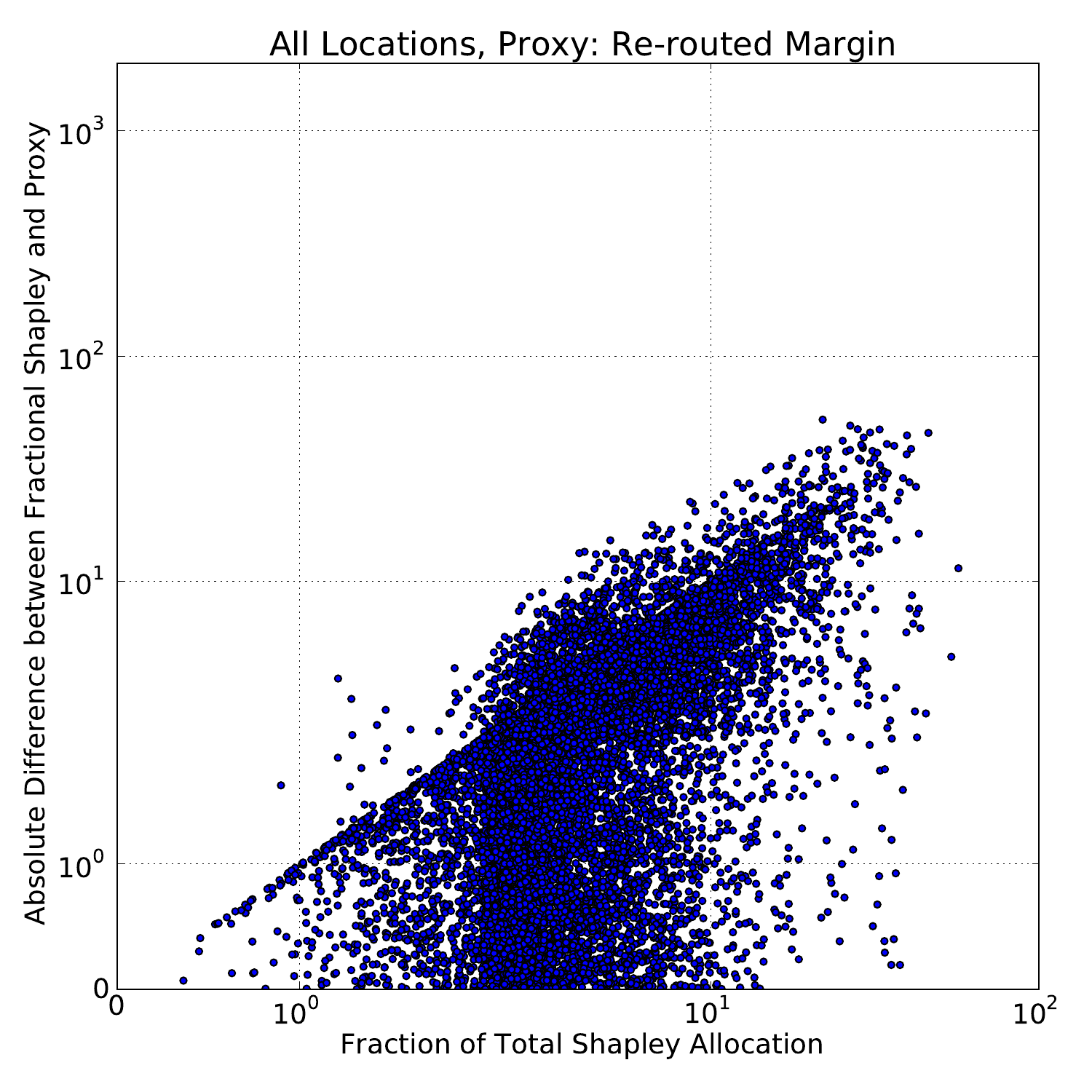}
\includegraphics[height=6.7cm,
page=165]{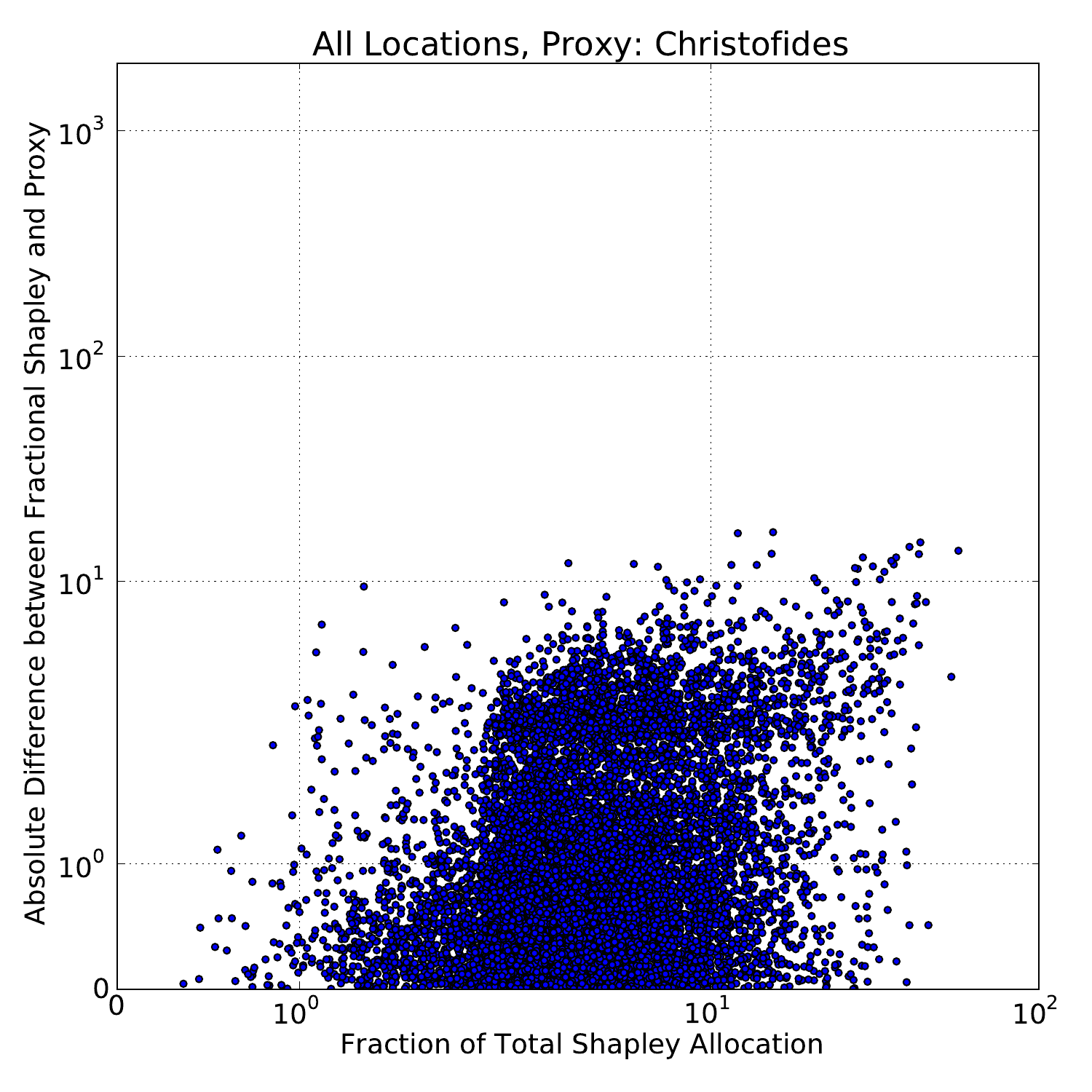}
\includegraphics[height=6.7cm, page=33]{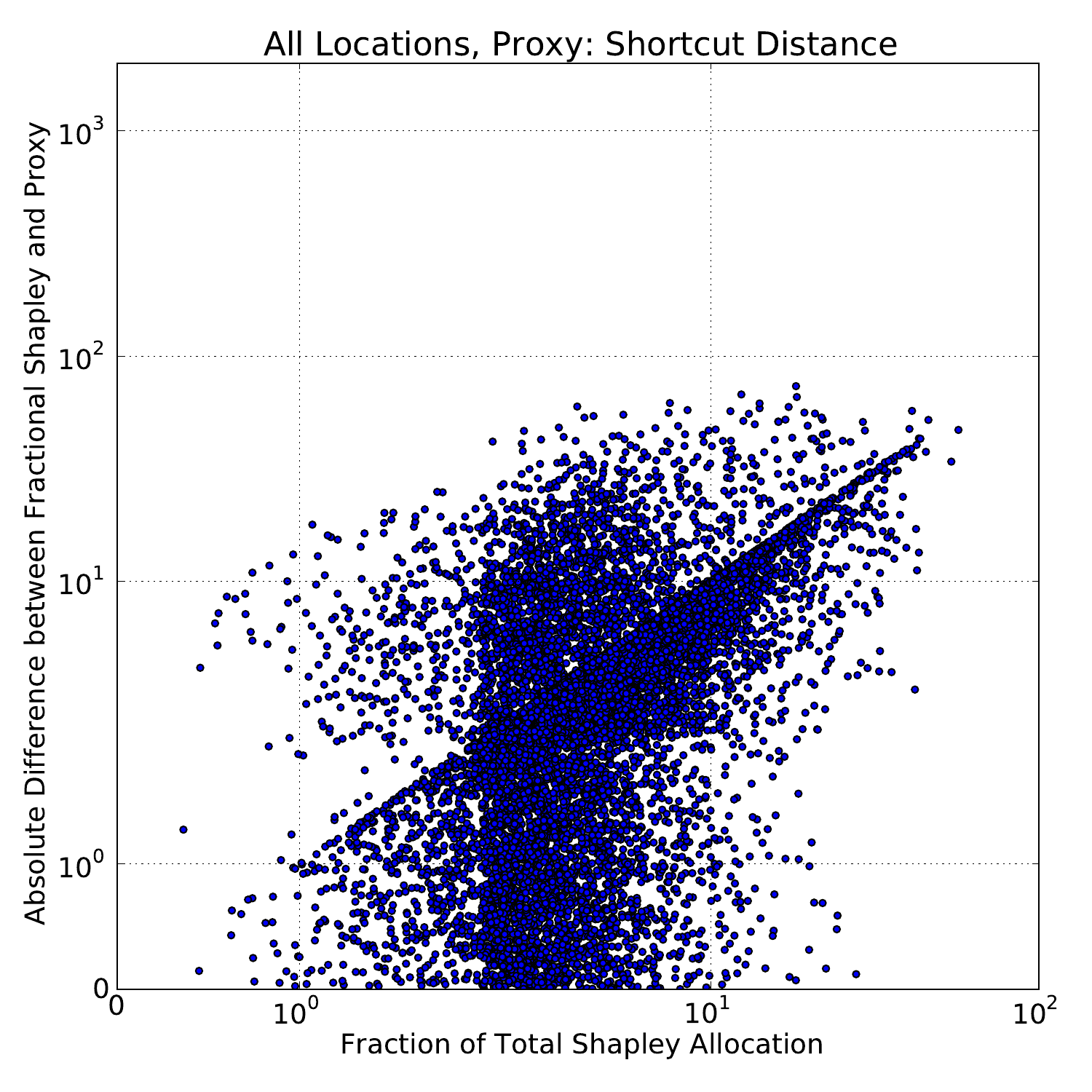}
\end{minipage} \begin{minipage}[b]{0.49\linewidth} \centering
\includegraphics[height=6.7cm,
page=132]{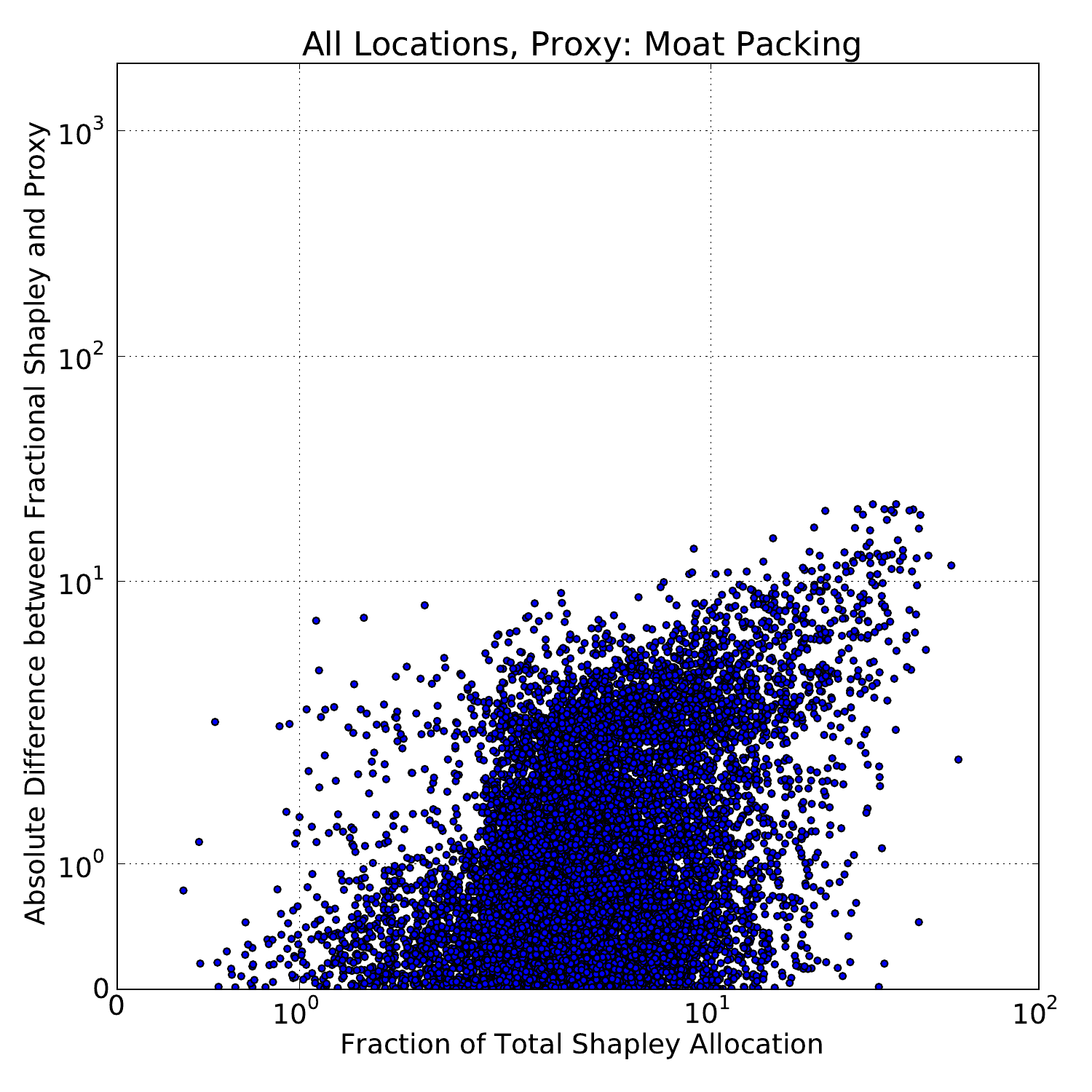}
\includegraphics[height=6.7cm, page=99]{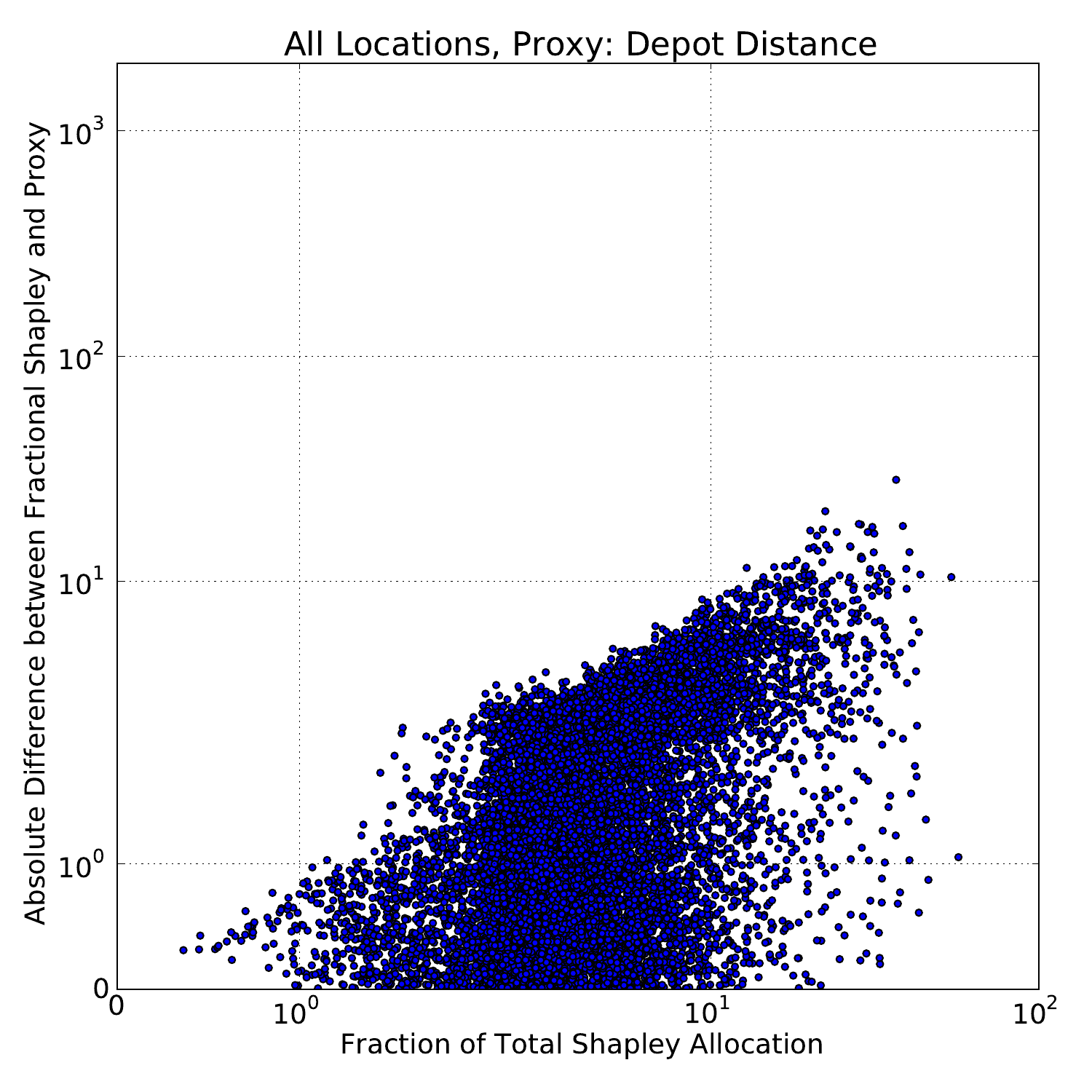}
\includegraphics[height=6.7cm,
page=198]{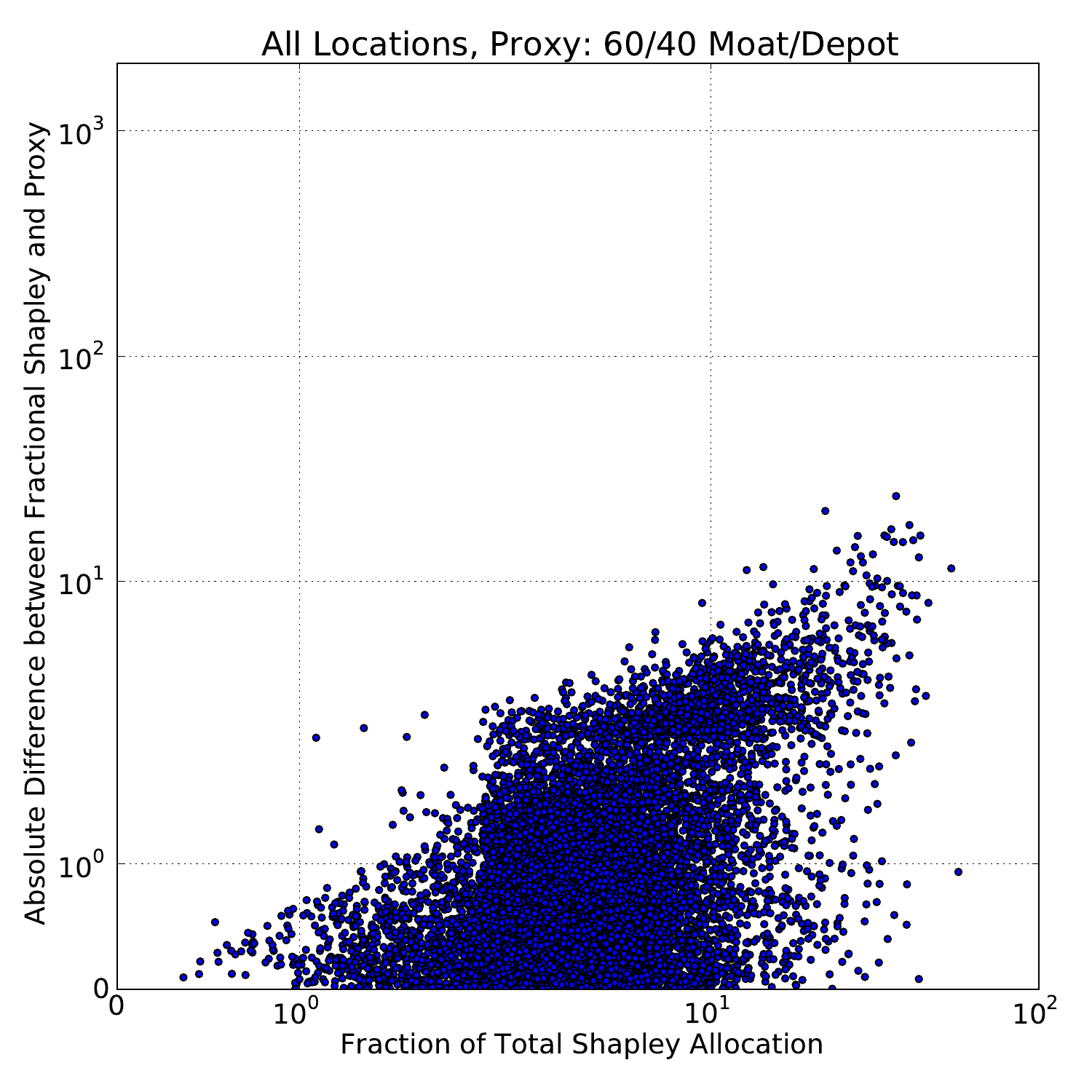} \end{minipage}
\caption{Absolute value of the difference between the $\fracshapley$ and
$\fracproxy$ plotted as a function of $\fracshapley$ for all the points
in the Synthetic data for all game sizes.  Note that these are log-log
plots to highlight the spread of the data. This is a repeat of
Figure~\ref{fig:synthetic_scatter_error}} \end{figure}

%
\clearpage \section{Auckland Data}\label{appendix:auckland}

For Auckland we obtained 13 instances of 10 location games and 8
instances of 20 location games.

\begin{table}[H] \centering \begin{tabular}{lcccccc} \toprule
&\multicolumn{2}{c}{10 Locations} & \multicolumn{2}{c}{20 Locations} &
\multicolumn{2}{c}{All Games} \\ & RMSE & St. Dev. & RMSE & St.Dev.&
RMSE & St. Dev.  \\ \midrule Shortcut Distance & 0.4208 & 0.1140 &
0.3498 & 0.1599 & 0.3853 & 0.1370 \\ Re-routed Margin  & 0.4111 & 0.1052
& 0.3222 & 0.1427 & 0.3667 & 0.1240 \\ Depot Distance    & 0.1680 &
0.0923 & 0.0937 & 0.0419 & 0.1309 & 0.0671 \\ \midrule Moat Packing     
& 0.2079 & 0.0947 & 0.2165 & 0.1364 & 0.2122 & 0.1156 \\ Christofides   
  & 0.1409 & 0.0731 & 0.0887 & 0.0514 & 0.1148 & 0.0623 \\ 60/40
Moat/Depot  & 0.1541 & 0.0831 & 0.0879 & 0.0414 & 0.1210 & 0.0623 \\
\bottomrule \end{tabular} \caption{Root Mean Squared Error (RMSE) and
Standard Deviation (St. Dev.) for the Auckland data for games with 10
and 20 locations. Lower is better.} \end{table}

\begin{table}[H] \centering \begin{tabular}{lcccccc} \toprule
&\multicolumn{2}{c}{10 Locations} & \multicolumn{2}{c}{20 Locations} &
\multicolumn{2}{c}{All Games} \\ & $\tau$ & St. Dev. & $\tau$ & St.
Dev.& $\tau$ & St. Dev.  \\ \midrule Shortcut Distance & 0.0470 & 0.3513
& 0.1083 & 0.1528 & 0.0777 & 0.2521 \\ Re-routed Margin  & 0.1815 &
0.2986 & 0.2538 & 0.2080 & 0.2177 & 0.2533 \\ Depot Distance    & 0.0085
& 0.3026 & 0.1520 & 0.3265 & 0.0803 & 0.3146 \\ \midrule Moat Packing   
  & 0.2122 & 0.2652 & 0.2210 & 0.1528 & 0.2166 & 0.2090 \\ Christofides 
    & 0.1068 & 0.3442 & 0.2456 & 0.3474 & 0.1762 & 0.3458 \\ 60/40
Moat/Depot  & 0.0513 & 0.2853 & 0.1886 & 0.3369 & 0.1200 & 0.3111\\
\bottomrule \end{tabular} \caption{Average KT distance ($\tau$) and
Standard Deviation (St. Dev.) for the Auckland data for games with 10
and 20 locations. Higher is better; $+1$ means the two lists are
perfectly correlated and $-1$ means the two lists are perfectly
anti-correlated.} \end{table}

\begin{table}[H] \centering \begin{tabular}{lcccc} \toprule &
\multicolumn{2}{c}{Median} & \multicolumn{2}{c}{Maximum} \\ & 10
Locations           & 20 Locations          & 10 Locations         & 20
Locations         \\ \midrule Shortcut Distance & 0.2109 & 0.2704 &
0.8348 & 0.8065 \\ Re-routed Margin  & 0.4042 & 0.2053 & 1.0000 & 0.6492
\\ Depot Distance    & 0.4042 & 0.2081 & 1.0000 & 0.9164 \\ \midrule
Moat Packing      & 0.4042 & 0.1700 & 0.8348 & 0.5754 \\ Christofides   
  & 0.4042 & 0.3460 & 1.0000 & 0.9164 \\ 60/40 Moat/Depot  & 0.6767 &
0.1904 & 1.0000 & 0.7529 \\ \bottomrule \end{tabular} \caption{Median
and Maximum $p$ values out of 13 instances of 10 location games and 8
instances of 20 location games of $\tau$ for the Auckland data. Lower is
better, $p < 0.05$ required for statistical significance.} \end{table}

\begin{table}[H] \centering \begin{tabular}{lccc} \toprule & 10
Locations     & 20 Locations     & All Games \\ \midrule Shortcut
Distance & 7.7\%  & 0.0\%  & 4.8\%  \\ Re-routed Margin  & 15.4\% &
37.5\% & 23.8\% \\ Depot Distance    & 7.7\%  & 12.5\% & 9.5\%  \\
\midrule Moat Packing      & 7.7\%  & 25.0\% & 14.3\% \\ Christofides   
  & 7.7\%  & 75.0\% & 33.3\% \\ 60/40 Moat/Depot  & 7.7\%  & 25.0\% &
14.3\% \\ \bottomrule \end{tabular} \caption{Percentage of correct top
elements of the Shapley ordering identified by the respective proxy for
the 13 games of size 10 and 8 games of size 20 for the Auckland data.}
\end{table}

\begin{figure}[h] \centering \begin{minipage}[b]{0.49\linewidth}
\centering \includegraphics[width=\columnwidth,
page=1]{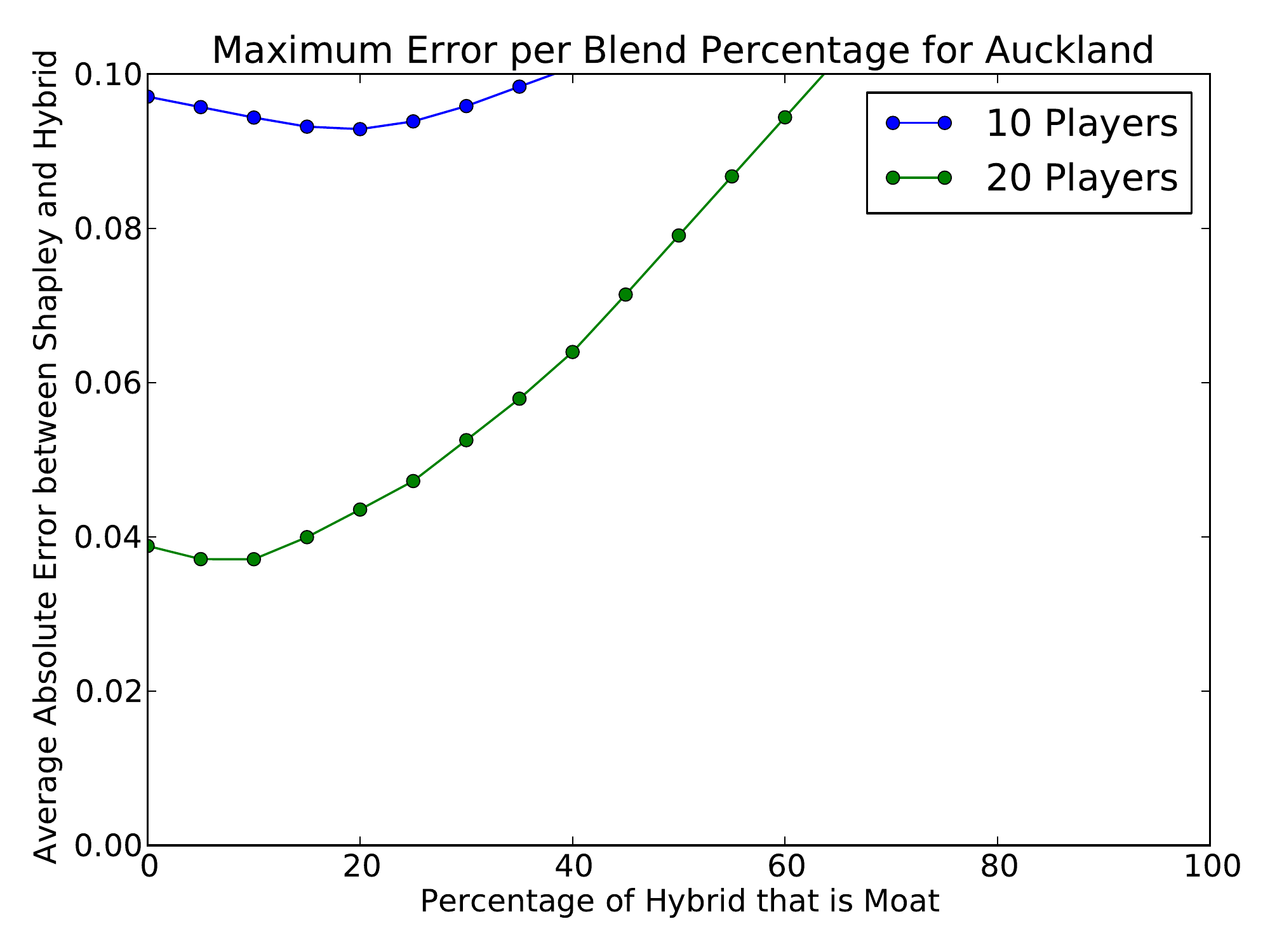} \end{minipage}
\begin{minipage}[b]{0.49\linewidth} \centering
\includegraphics[width=\columnwidth, page=3]{./Emp_Auckland}
\end{minipage} \caption{Effect of the blending parameter $\lambda$ on
the error of Shapley allocation prediction for the Auckland dataset. The
left-hand graph shows the average worst case error that any single
location experiences, while the right-hand graph shows the average error
over all locations.} \end{figure}

\begin{figure} \begin{minipage}[b]{0.49\linewidth} \centering
\includegraphics[height=6.7cm,
page=6]{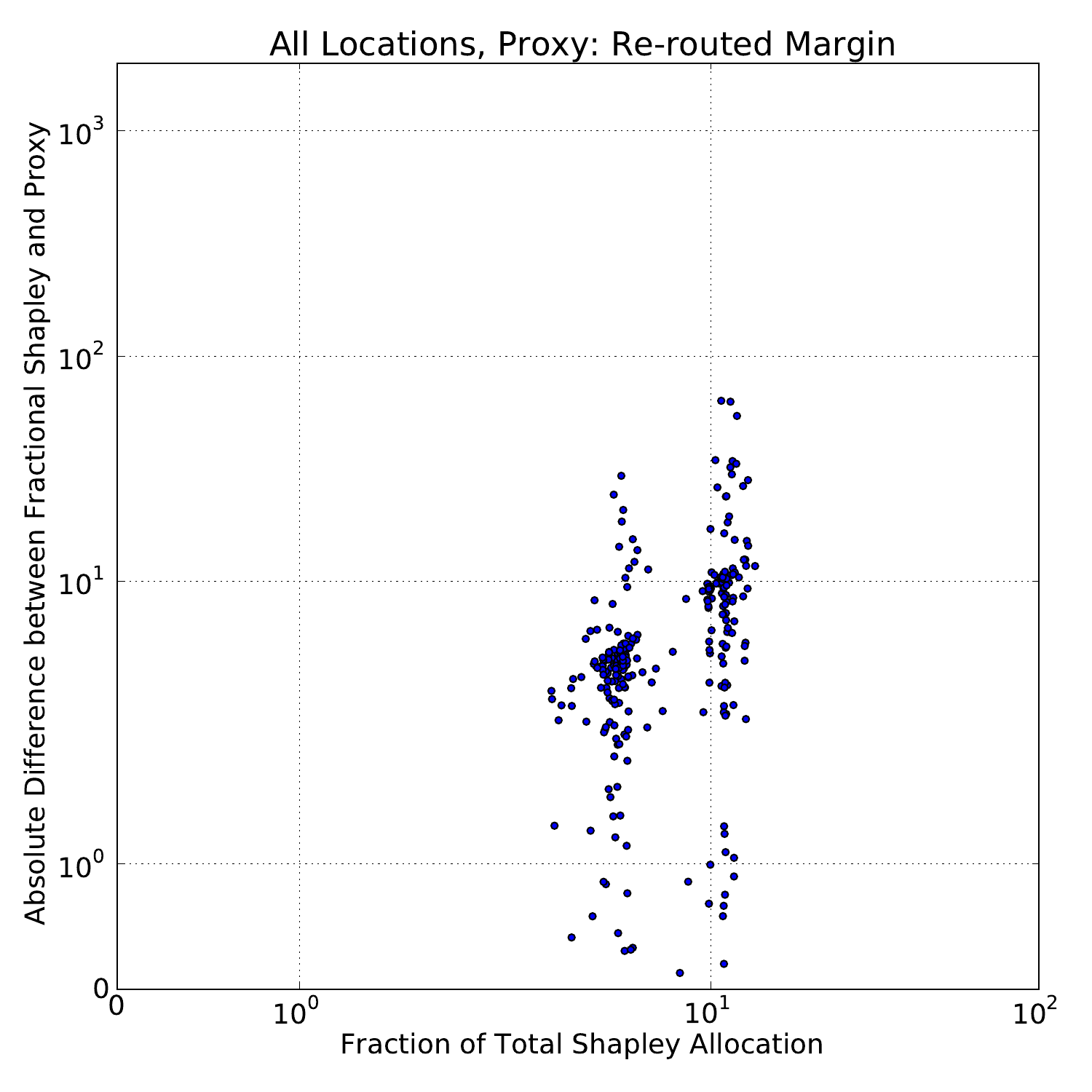}
\includegraphics[height=6.7cm,
page=9]{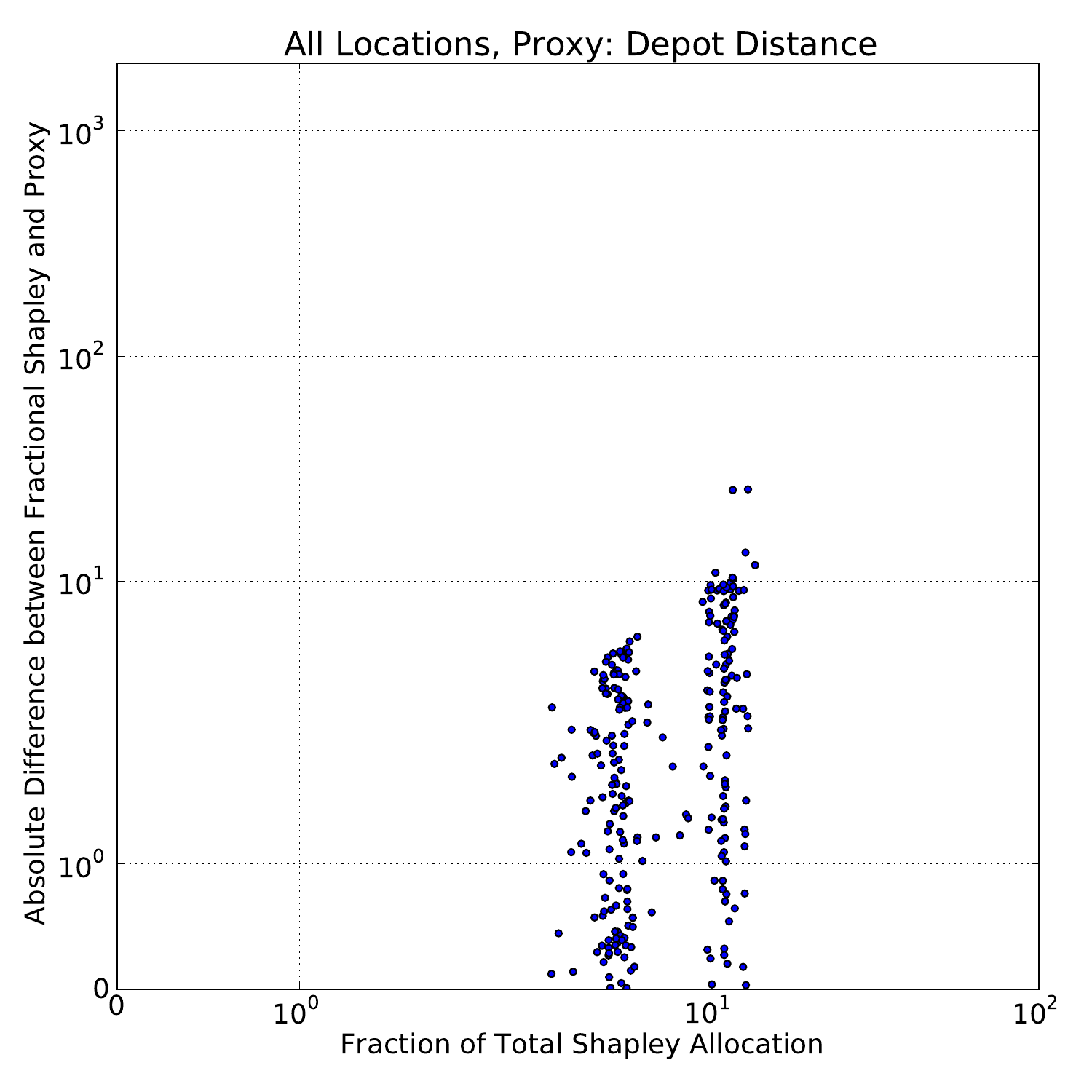}
\includegraphics[height=6.7cm,
page=15]{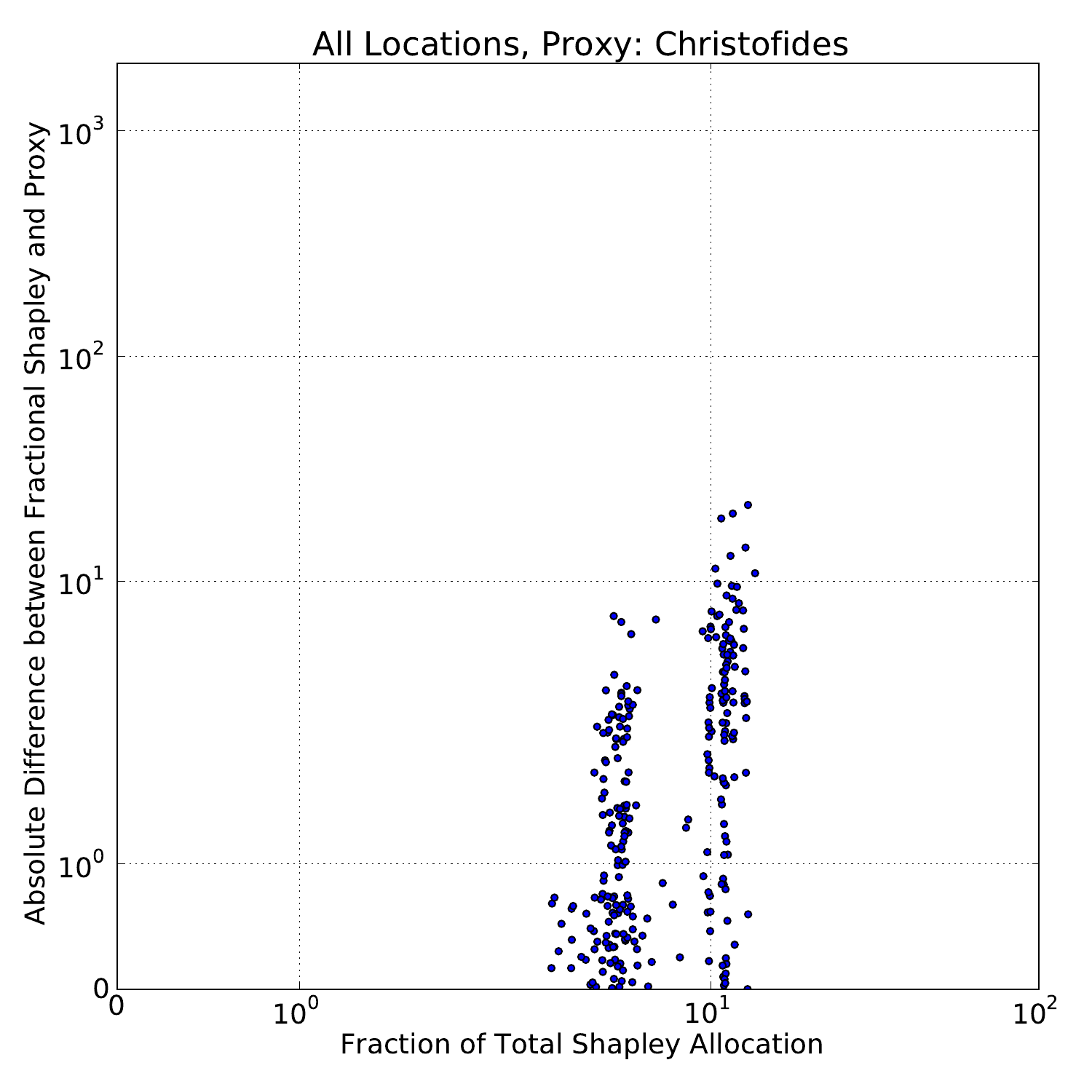} \end{minipage}
\begin{minipage}[b]{0.49\linewidth} \centering
\includegraphics[height=6.7cm,
page=3]{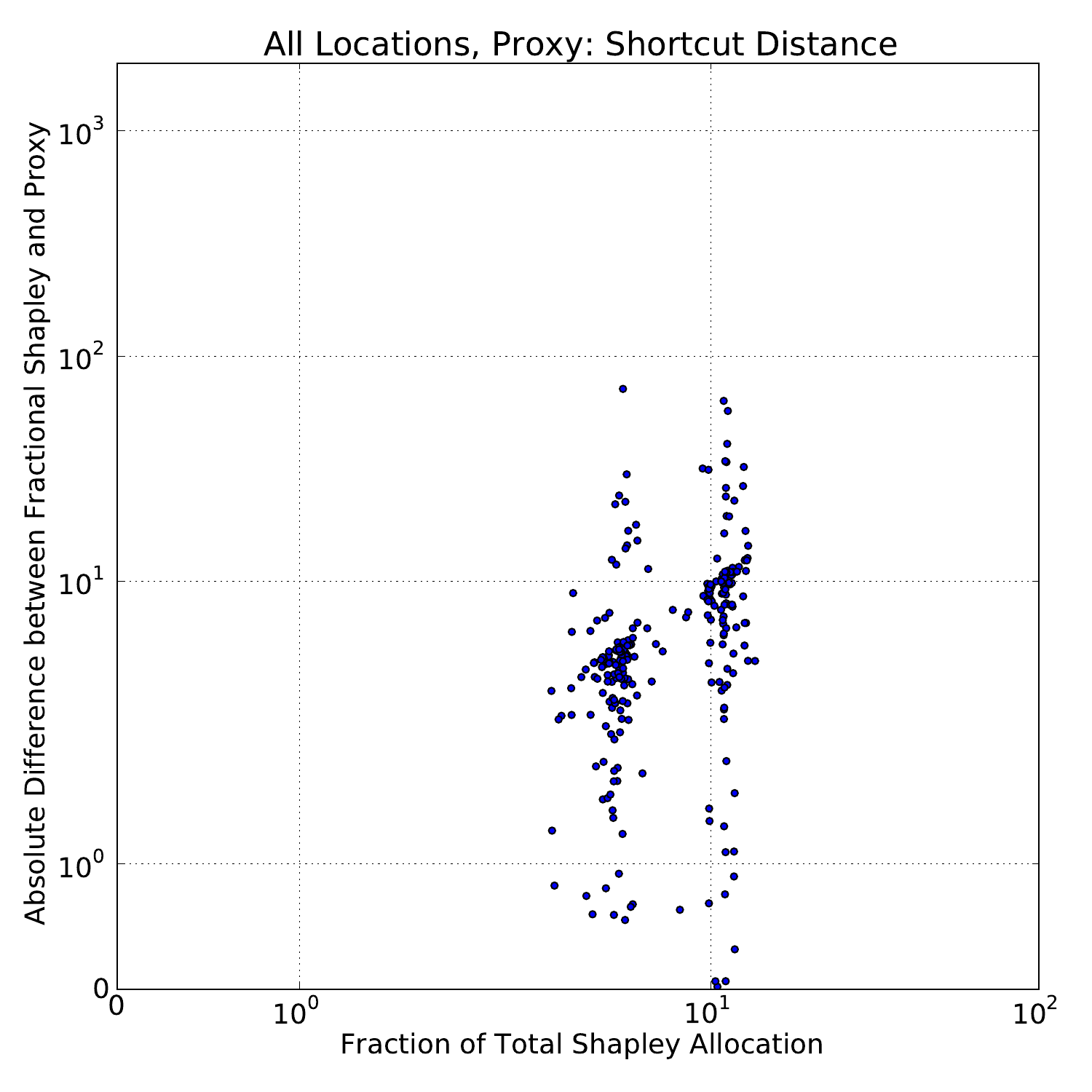}
\includegraphics[height=6.7cm,
page=12]{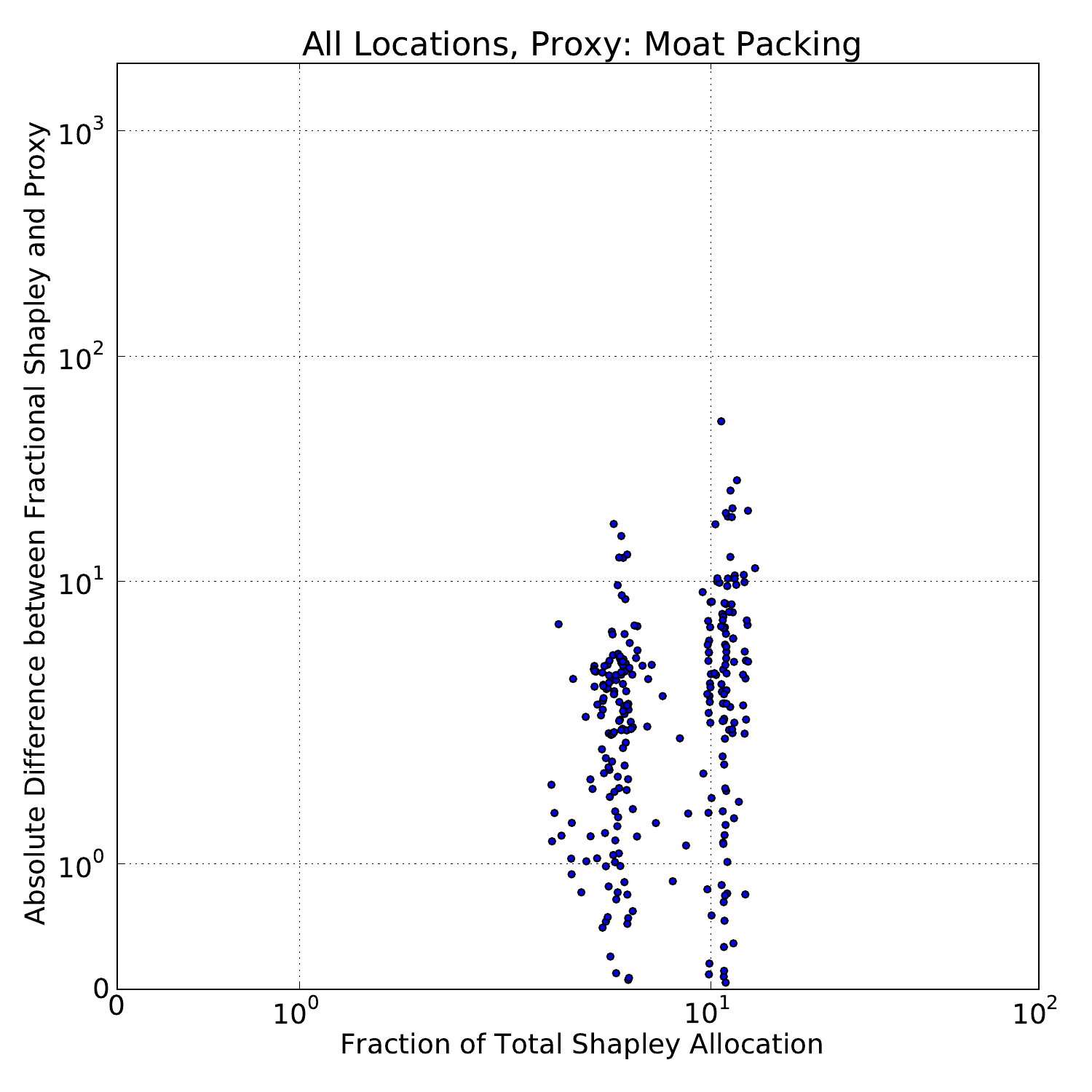}
\includegraphics[height=6.7cm,
page=18]{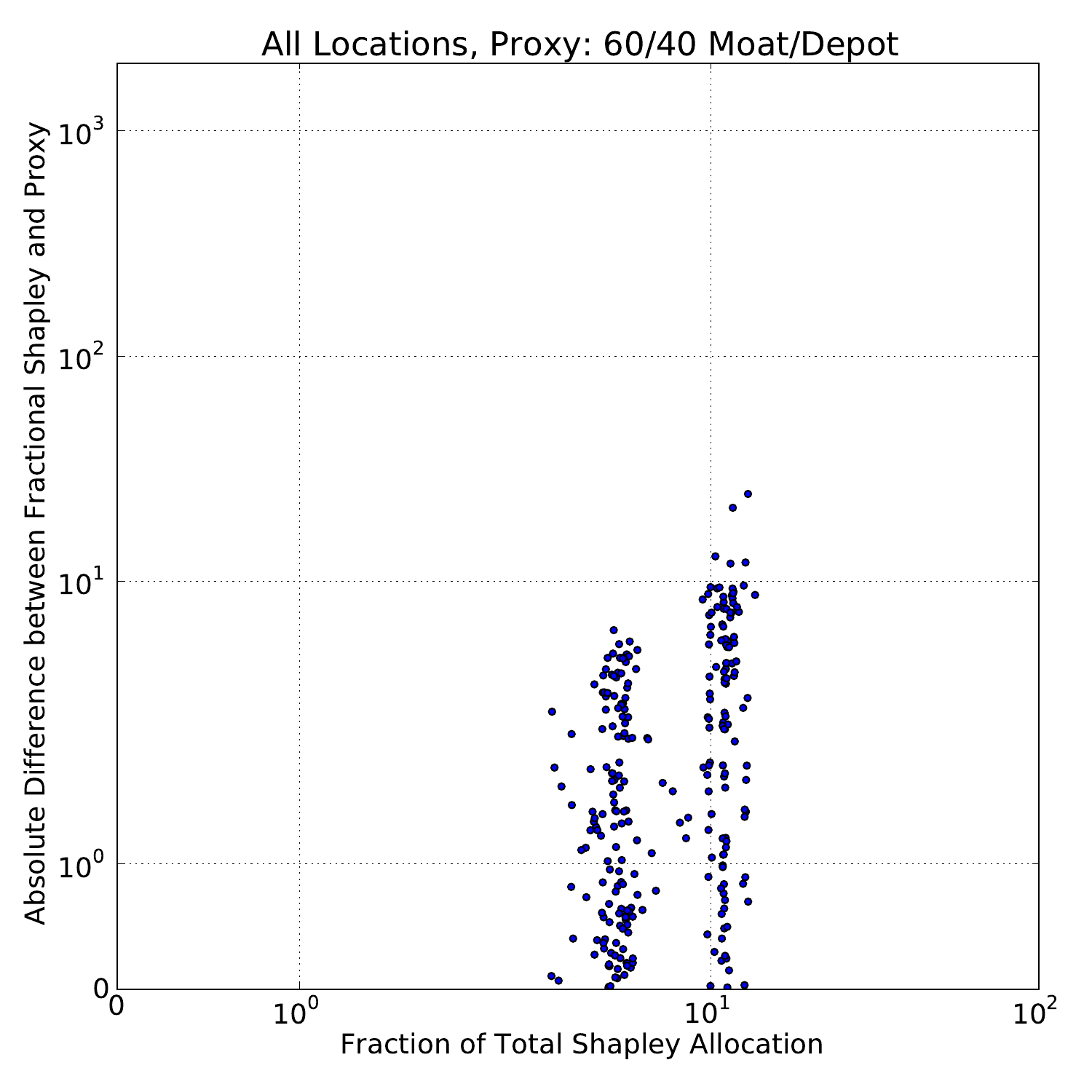} \end{minipage}
\caption{Absolute value of the difference between the $\fracshapley$ and
$\fracproxy$ plotted as a function of $\fracshapley$ for all the points
in the Auckland data for all game sizes.  Note that these are log-log
plots to highlight the spread of the data.} \end{figure}

\newpage \section{Canberra Data}\label{appendix:canberra}

For the Canberra data we obtained 7 instances each of problems with 10
and 20 locations, respectively.

\begin{table}[H] \centering \begin{tabular}{lcccccc} \toprule
&\multicolumn{2}{c}{10 Locations} & \multicolumn{2}{c}{20 Locations} &
\multicolumn{2}{c}{All Games} \\ & RMSE & St. Dev. & RMSE & St.Dev.&
RMSE & St. Dev.  \\ \midrule Shortcut Distance & 0.3651 & 0.0763 &
0.2827 & 0.0088 & 0.3239 & 0.0426 \\ Re-routed Margin  & 0.3353 & 0.0930
& 0.2528 & 0.0149 & 0.2941 & 0.0540 \\ Depot Distance    & 0.1405 &
0.0362 & 0.0870 & 0.0262 & 0.1138 & 0.0312 \\ \midrule Moat Packing     
& 0.1717 & 0.0696 & 0.1597 & 0.0279 & 0.1657 & 0.0488 \\ Christofides   
  & 0.1206 & 0.0464 & 0.0830 & 0.0219 & 0.1018 & 0.0342 \\ 60/40
Moat/Depot  & 0.1291 & 0.0395 & 0.0777 & 0.0221 & 0.1034 & 0.0308 \\
\bottomrule \end{tabular} \caption{Root Mean Squared Error (RMSE) and
Standard Deviation (St. Dev.) for the Canberra data for games with 10
and 20 locations. Lower is better.} \end{table}

\begin{table}[H] \centering \begin{tabular}{lcccccc} \toprule
&\multicolumn{2}{c}{10 Locations} & \multicolumn{2}{c}{20 Locations} &
\multicolumn{2}{c}{All Games} \\ & $\tau$ & St. Dev. & $\tau$ & St.
Dev.& $\tau$ & St. Dev.  \\ \midrule

Shortcut Distance & -0.0714 & 0.1918 & 0.0126 & 0.1758 & -0.0294 &
0.1838 \\ Re-routed Margin  & 0.3095  & 0.2552 & 0.2239 & 0.2124 &
0.2667  & 0.2338 \\ Depot Distance    & 0.0556  & 0.0985 & 0.0693 &
0.1161 & 0.0625  & 0.1073 \\ \midrule Moat Packing      & 0.2505  &
0.2860 & 0.1512 & 0.1458 & 0.2009  & 0.2159 \\ Christofides      &
0.1190  & 0.2109 & 0.2331 & 0.1916 & 0.1761  & 0.2013 \\ 60/40
Moat/Depot  & 0.1190  & 0.1888 & 0.1178 & 0.1632 & 0.1184  & 0.1760 \\
\bottomrule \end{tabular} \caption{Average KT distance ($\tau$) and
Standard Deviation (St. Dev.) for the Canberra data for games with 10
and 20 locations. Higher is better; $+1$ means the two lists are
perfectly correlated and $-1$ means the two lists are perfectly
anti-correlated.} \end{table}

\begin{table}[H] \centering \begin{tabular}{lcccc} \toprule &
\multicolumn{2}{c}{Median} & \multicolumn{2}{c}{Maximum} \\ & 10
Locations           & 20 Locations          & 10 Locations         & 20
Locations         \\ \midrule Shortcut Distance & 0.6767 & 0.3103 &
1.0000 & 0.8886 \\ Re-routed Margin  & 0.2971 & 0.0744 & 0.4042 & 0.9721
\\ Depot Distance    & 0.8348 & 0.3449 & 1.0000 & 0.9164 \\ \midrule
Moat Packing      & 0.2109 & 0.3818 & 0.8348 & 0.9164 \\ Christofides   
  & 0.6767 & 0.1515 & 0.8348 & 0.7529 \\ 60/40 Moat/Depot  & 0.6767 &
0.2208 & 0.8348 & 0.9721 \\ \bottomrule \end{tabular} \caption{Median
and Maximum $p$ values out of 7 games per number of locations of $\tau$
for the Canberra data. Lower is better, $p < 0.05$ required for
statistical significance.} \end{table}

\begin{table}[H] \centering \begin{tabular}{lccc} \toprule & 10
Locations     & 20 Locations     & All Games \\ \midrule Shortcut
Distance & 0.0\%  & 0.0\%  & 0.0\%  \\ Re-routed Margin  & 0.0\%  &
28.6\% & 14.3\% \\ Depot Distance    & 0.0\%  & 28.6\% & 14.3\% \\
\midrule Moat Packing      & 0.0\%  & 42.9\% & 21.4\% \\ Christofides   
  & 0.0\%  & 42.9\% & 21.4\% \\ 60/40 Moat/Depot  & 14.3\% & 42.9\% &
28.6\% \\ \bottomrule \end{tabular} \caption{Percentage of correct top
elements of the Shapley ordering identified by the respective proxy for
the 7 trials per game size for the Canberra data.} \end{table}

\begin{figure}[h] \centering \begin{minipage}[b]{0.49\linewidth}
\centering \includegraphics[width=\columnwidth,
page=1]{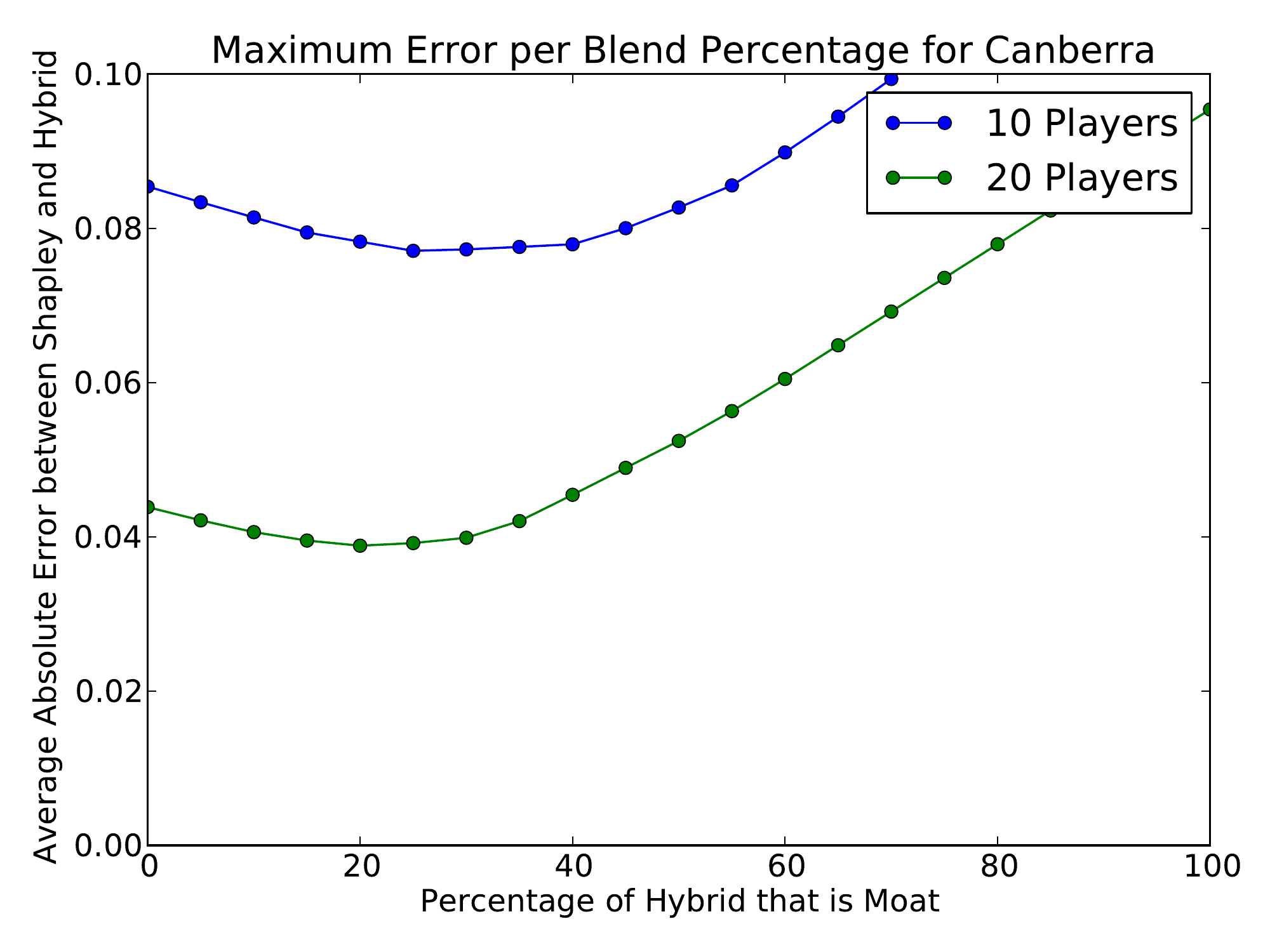} \end{minipage}
\begin{minipage}[b]{0.49\linewidth} \centering
\includegraphics[width=\columnwidth, page=3]{./Emp_Canberra}
\end{minipage} \caption{Effect of the blending parameter $\lambda$ on
the error of Shapley allocation prediction for the Canberra dataset. The
left-hand graph shows the average worst case error that any single
location experiences, while the right-hand graph shows the average error
over all locations.} \end{figure}

\begin{figure} \begin{minipage}[b]{0.49\linewidth} \centering
\includegraphics[height=6.7cm,
page=6]{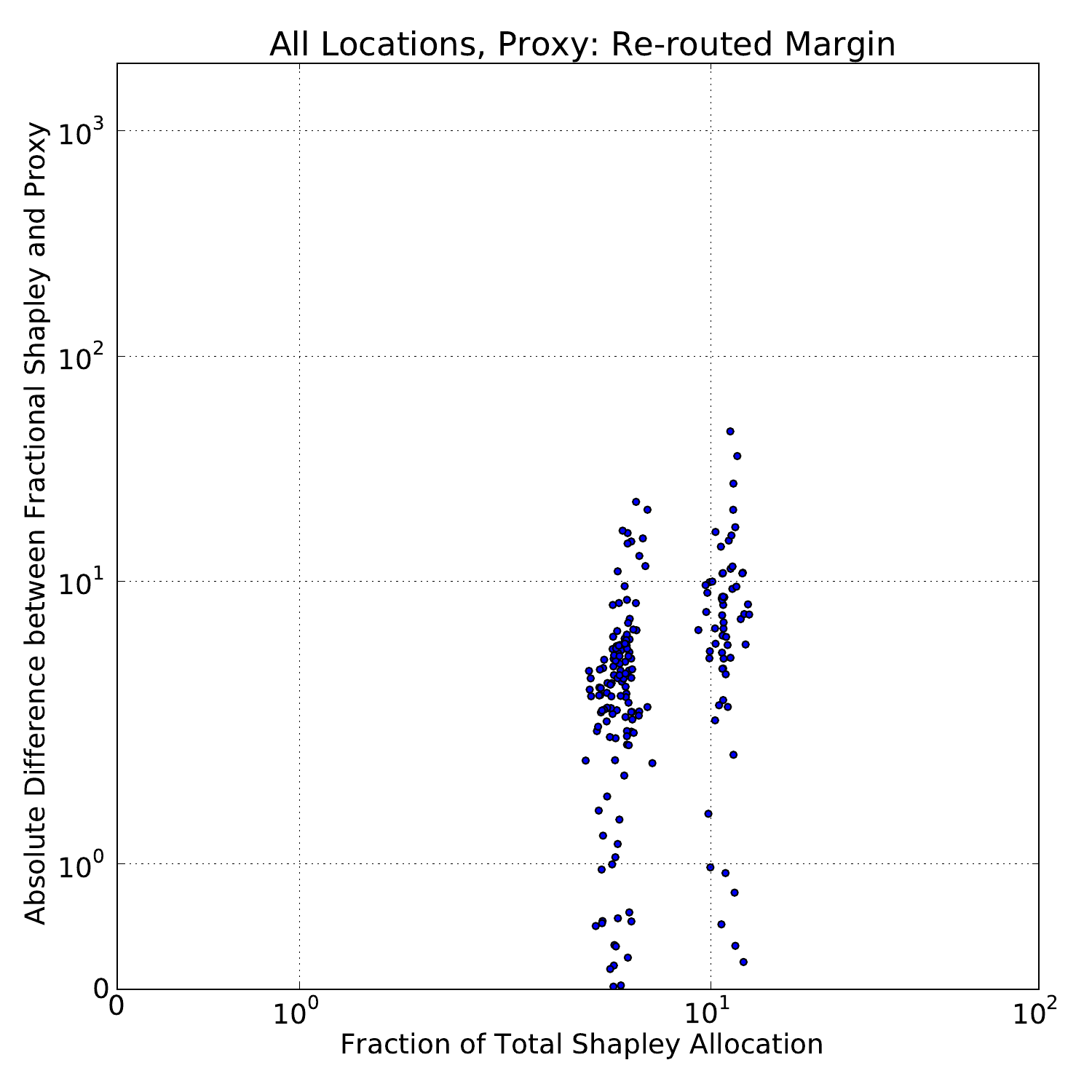}
\includegraphics[height=6.7cm,
page=9]{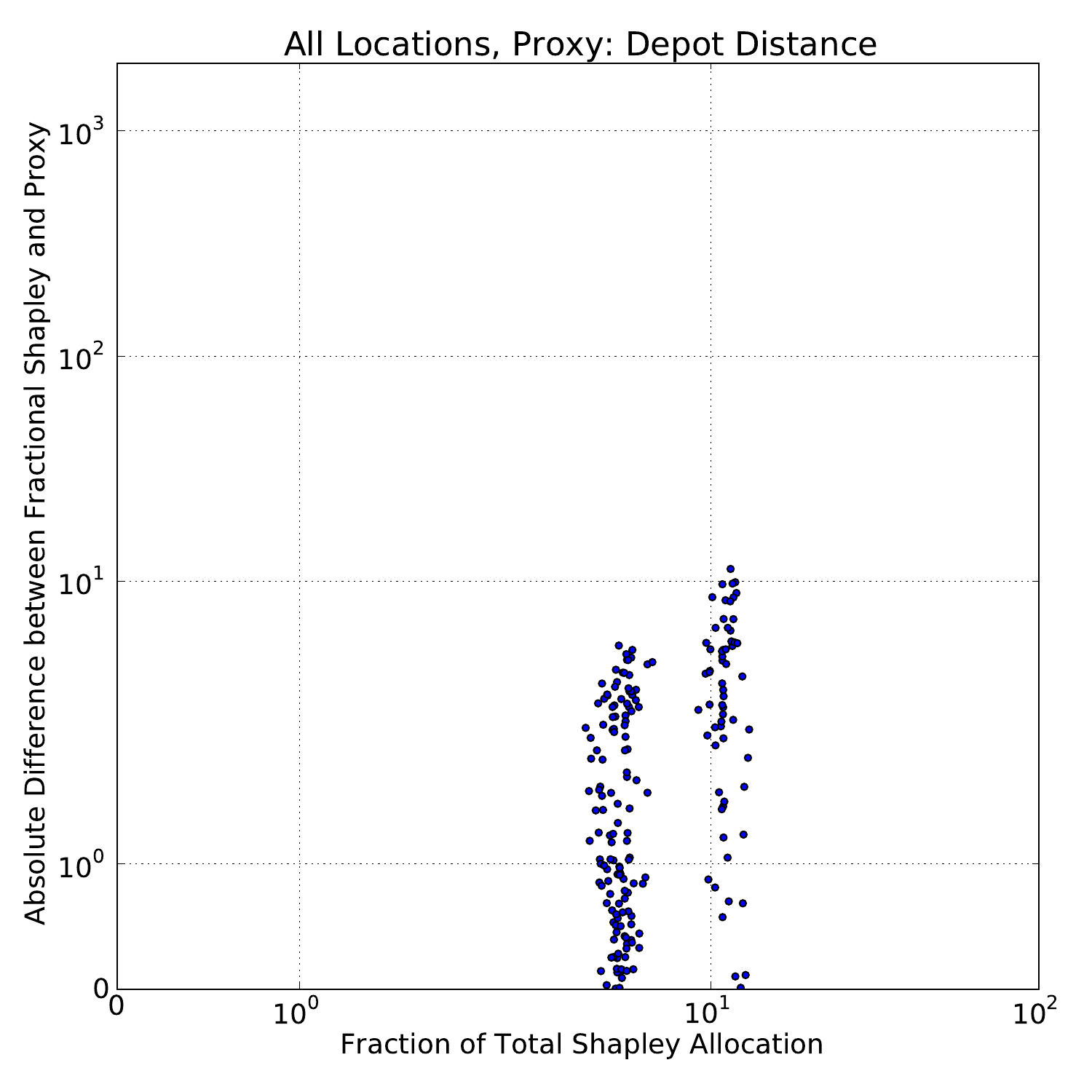}
\includegraphics[height=6.7cm,
page=15]{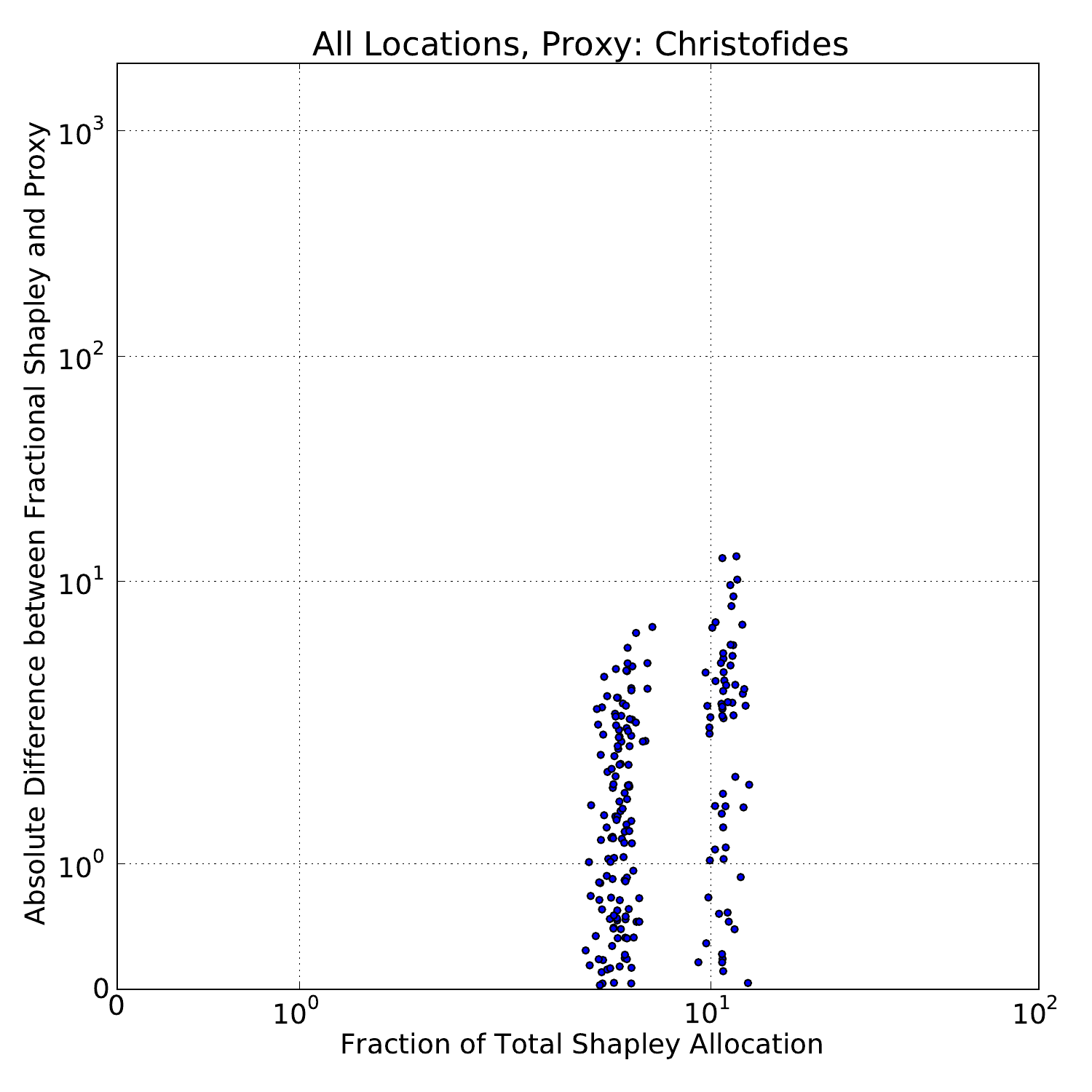} \end{minipage}
\begin{minipage}[b]{0.49\linewidth} \centering
\includegraphics[height=6.7cm,
page=3]{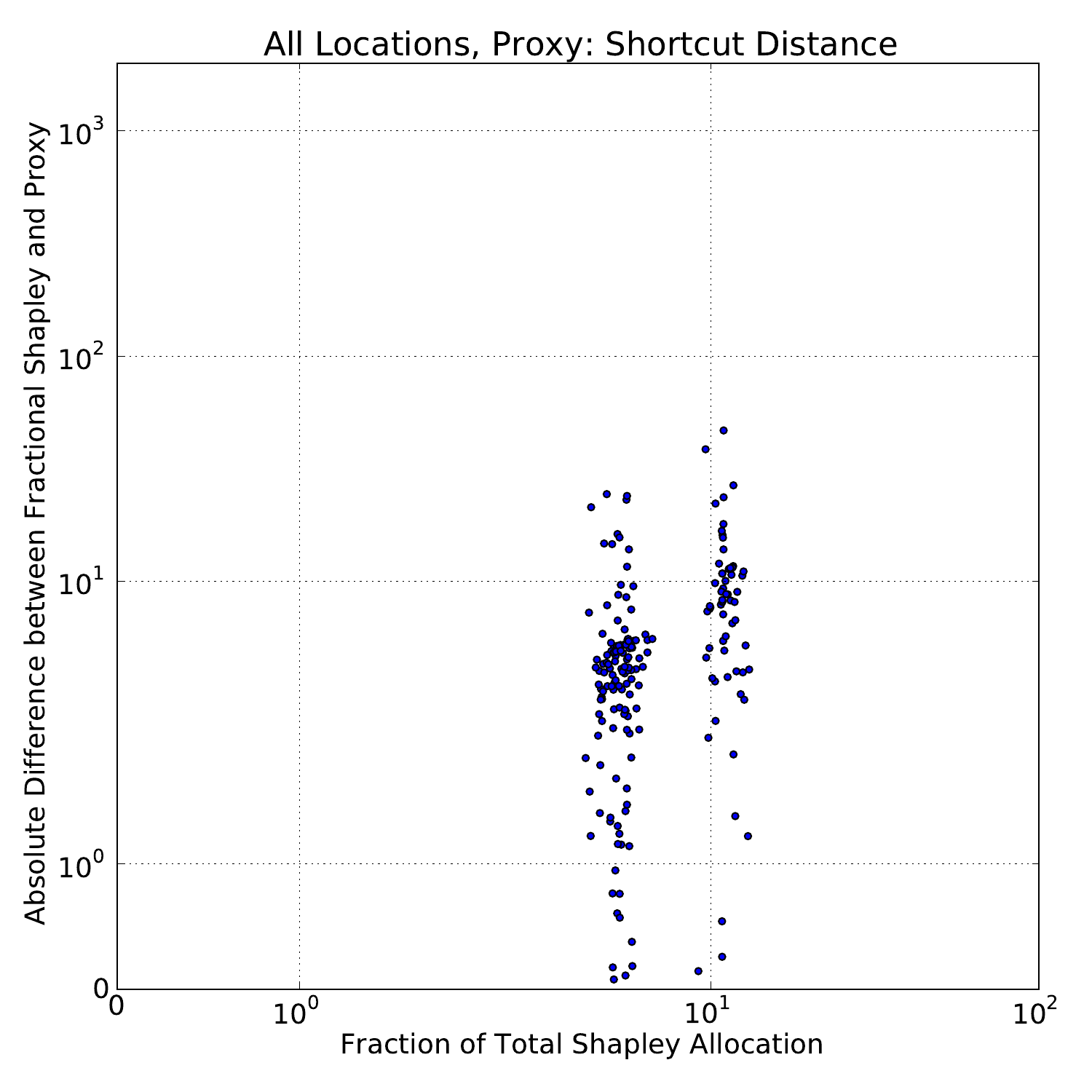}
\includegraphics[height=6.7cm,
page=12]{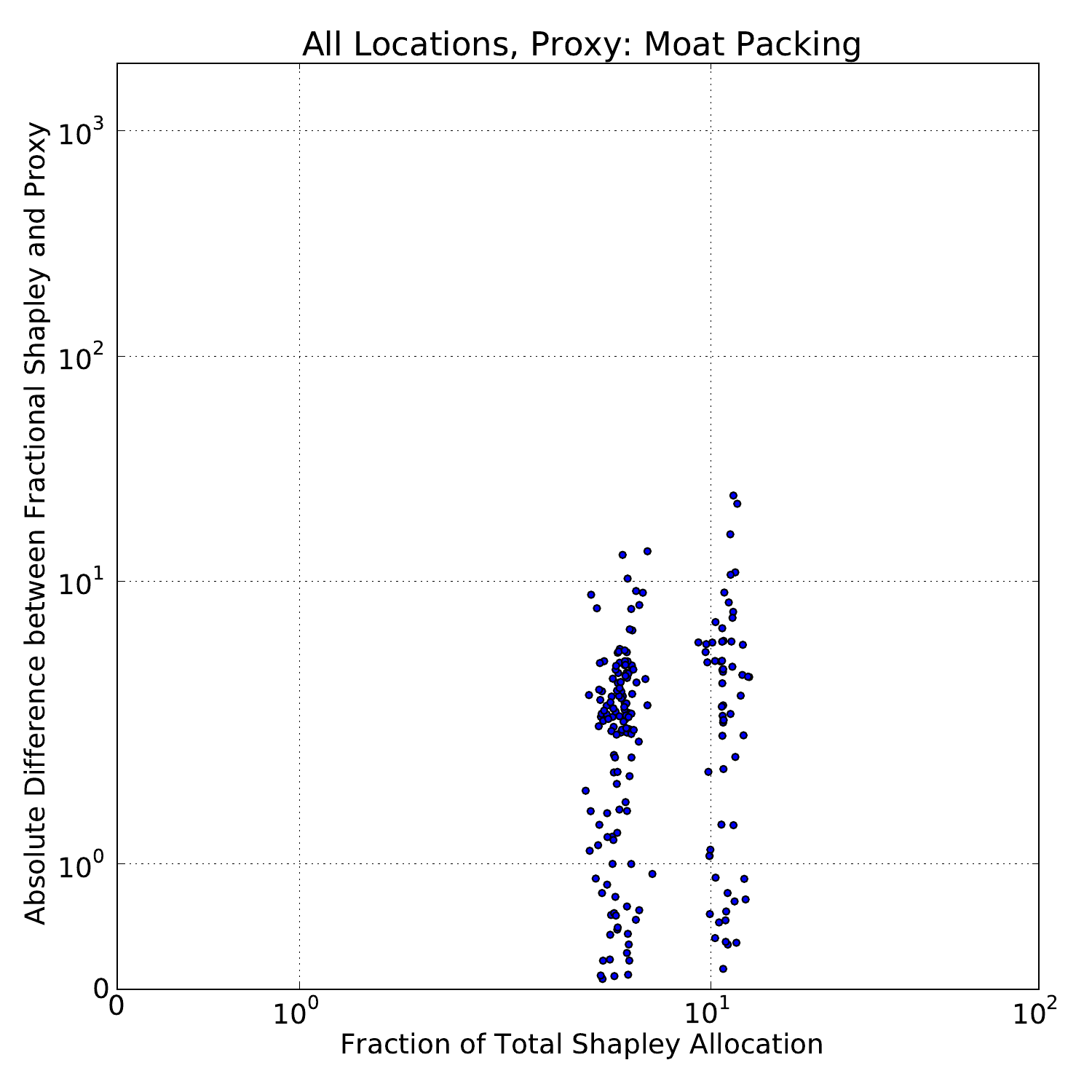}
\includegraphics[height=6.7cm,
page=18]{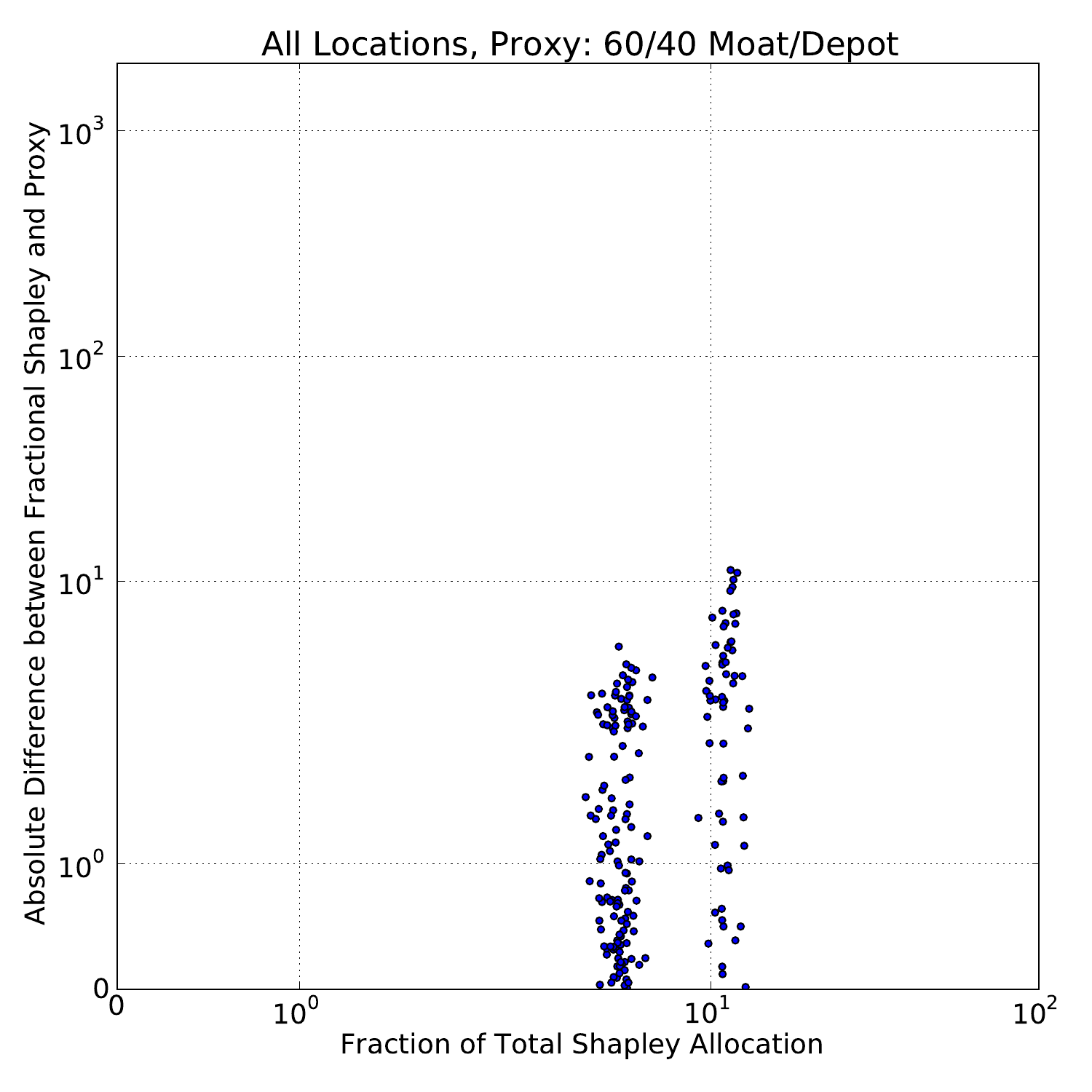} \end{minipage}
\caption{Absolute value of the difference between the $\fracshapley$ and
$\fracproxy$ plotted as a function of $\fracshapley$ for all the points
in the Canberra data for all game sizes.  Note that these are log-log
plots to highlight the spread of the data.} \end{figure}

\clearpage \section{Sydney Data}\label{appendix:sydney}

For Sydney we obtained 49 instances of 10 location games and 29
instances of 20 location games.

\begin{table}[H] \centering \begin{tabular}{lcccccc} \toprule
&\multicolumn{2}{c}{10 Locations} & \multicolumn{2}{c}{20 Locations} &
\multicolumn{2}{c}{All Games} \\ & RMSE & St. Dev. & RMSE & St.Dev.&
RMSE & St. Dev.  \\ \midrule Shortcut Distance & 0.4598 & 0.1532 &
0.3266 & 0.0976 & 0.3932 & 0.1254 \\ Re-routed Margin  & 0.4288 & 0.1398
& 0.2903 & 0.0832 & 0.3596 & 0.1115 \\ Depot Distance    & 0.1248 &
0.0500 & 0.0852 & 0.0269 & 0.1050 & 0.0385 \\ \midrule Moat Packing     
& 0.2692 & 0.1308 & 0.2005 & 0.0779 & 0.2349 & 0.1044 \\ Christofides   
  & 0.1338 & 0.0708 & 0.0864 & 0.0248 & 0.1101 & 0.0478 \\ 60/40
Moat/Depot  & 0.1438 & 0.0687 & 0.0731 & 0.0271 & 0.1085 & 0.0479 \\
\bottomrule \end{tabular} \caption{Root Mean Squared Error (RMSE) and
Standard Deviation (St. Dev.) for the Sydney data for games with 10 and
20 locations. Lower is better.} \end{table}

\begin{table}[H] \centering \begin{tabular}{lcccccc} \toprule
&\multicolumn{2}{c}{10 Locations} & \multicolumn{2}{c}{20 Locations} &
\multicolumn{2}{c}{All Games} \\ & $\tau$ & St. Dev. & $\tau$ & St.
Dev.& $\tau$ & St. Dev.  \\ \midrule Shortcut Distance & -0.0213 &
0.2501 & 0.0882 & 0.1974 & 0.0335 & 0.2238 \\ Re-routed Margin  & 0.4114
 & 0.1895 & 0.4033 & 0.1681 & 0.4074 & 0.1788 \\ Depot Distance    &
0.1394  & 0.2242 & 0.1874 & 0.2140 & 0.1634 & 0.2191 \\ \midrule Moat
Packing      & 0.3937  & 0.2306 & 0.3675 & 0.1457 & 0.3806 & 0.1882 \\
Christofides      & 0.3016  & 0.2444 & 0.4083 & 0.1597 & 0.3550 & 0.2021
\\ 60/40 Moat/Depot  & 0.2563  & 0.2303 & 0.3035 & 0.2191 & 0.2799 &
0.2247 \\ \bottomrule \end{tabular} \caption{Average KT distance
($\tau$) and Standard Deviation (St. Dev.) for the Sydney data for games
with 10 and 20 locations. Higher is better; $+1$ means the two lists are
perfectly correlated and $-1$ means the two lists are perfectly
anti-correlated.} \end{table}

\begin{table}[H] \centering \begin{tabular}{lcccc} \toprule &
\multicolumn{2}{c}{Median} & \multicolumn{2}{c}{Maximum} \\ & 10
Locations           & 20 Locations          & 10 Locations         & 20
Locations         \\ \midrule Shortcut Distance & 0.4631 & 0.4393 &
1.0000 & 0.9720 \\ Re-routed Margin  & 0.1400 & 0.0265 & 0.7505 & 0.8546
\\ Depot Distance    & 0.4042 & 0.3449 & 1.0000 & 0.9721 \\ \midrule
Moat Packing      & 0.0953 & 0.0328 & 1.0000 & 0.8063 \\ Christofides   
  & 0.2109 & 0.0130 & 1.0000 & 0.5997 \\ 60/40 Moat/Depot  & 0.2971 &
0.0744 & 1.0000 & 0.9164 \\ \bottomrule \end{tabular} \caption{Median
and Maximum $p$ values out of 49 instances of 10 location games and 29
instances of 20 location games of $\tau$ for the Sydney data. Lower is
better, $p < 0.05$ required for statistical significance.} \end{table}

\begin{table}[H] \centering \begin{tabular}{lccc} \toprule & 10
Locations     & 20 Locations     & All Games \\ \midrule

Shortcut Distance & 6.1\%  & 24.1\% & 12.8\% \\ Re-routed Margin  &
55.1\% & 82.8\% & 65.4\% \\ Depot Distance    & 46.9\% & 48.3\% & 47.4\%
\\ \midrule Moat Packing      & 57.1\% & 75.9\% & 64.1\% \\ Christofides
     & 53.1\% & 55.2\% & 53.8\% \\ 60/40 Moat/Depot  & 55.1\% & 65.5\% &
59.0\% \\ \bottomrule \end{tabular} \caption{Percentage of correct top
elements of the Shapley ordering identified by the respective proxy for
the 49 instances of 10 location games and 29 instances of 20 location
games.} \end{table}

\begin{figure}[h] \centering \begin{minipage}[b]{0.49\linewidth}
\centering \includegraphics[width=\columnwidth,
page=1]{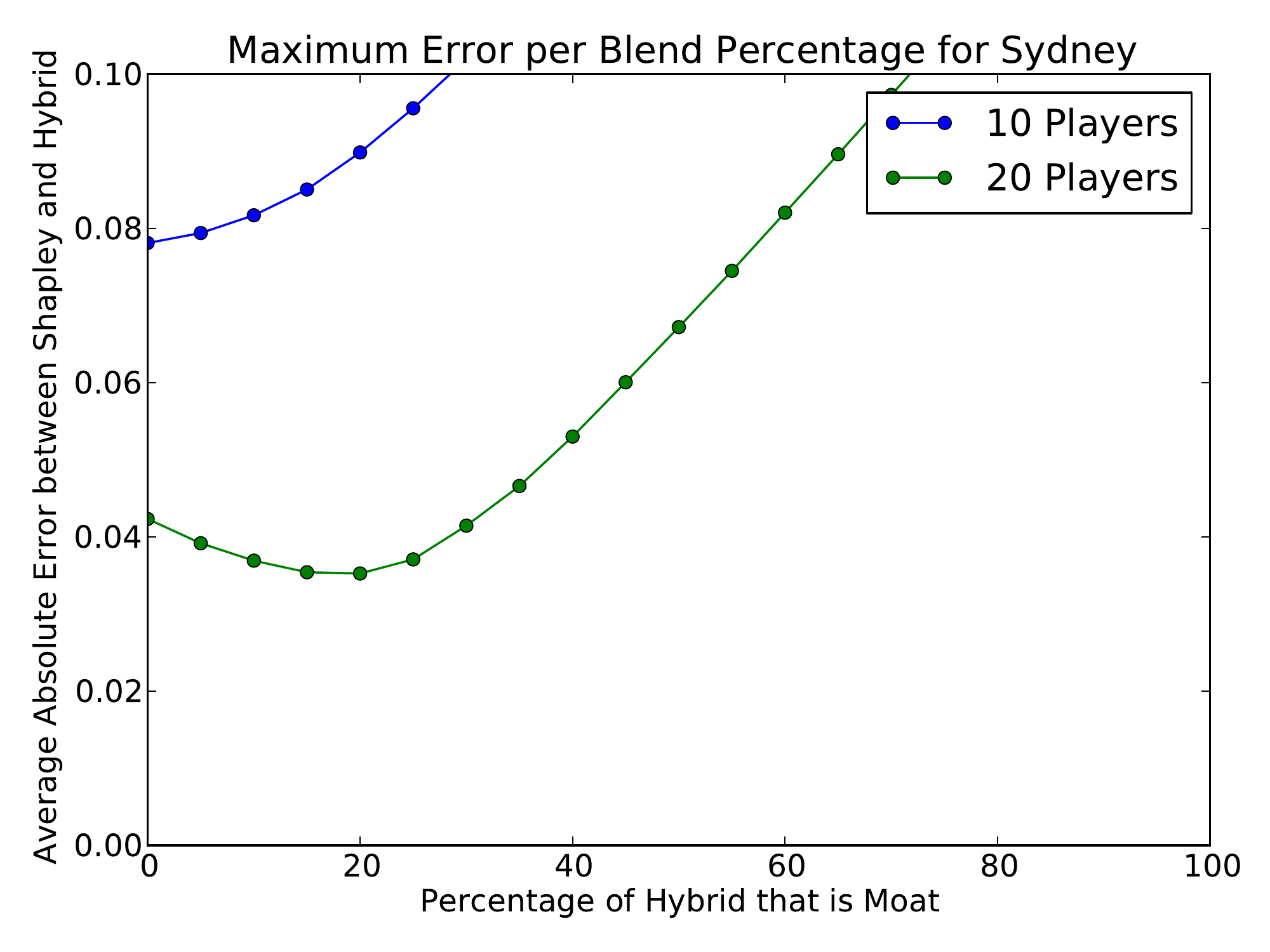} \end{minipage}
\begin{minipage}[b]{0.49\linewidth} \centering
\includegraphics[width=\columnwidth, page=3]{./Emp_Sydney}
\end{minipage} \caption{Effect of the blending parameter $\lambda$ on
the error of Shapley allocation prediction for the Sydney dataset. The
left-hand graph shows the average worst case error that any single
location experiences, while the right-hand graph shows the average error
over all locations.} \end{figure}

\begin{figure} \begin{minipage}[b]{0.49\linewidth} \centering
\includegraphics[height=6.7cm, page=6]{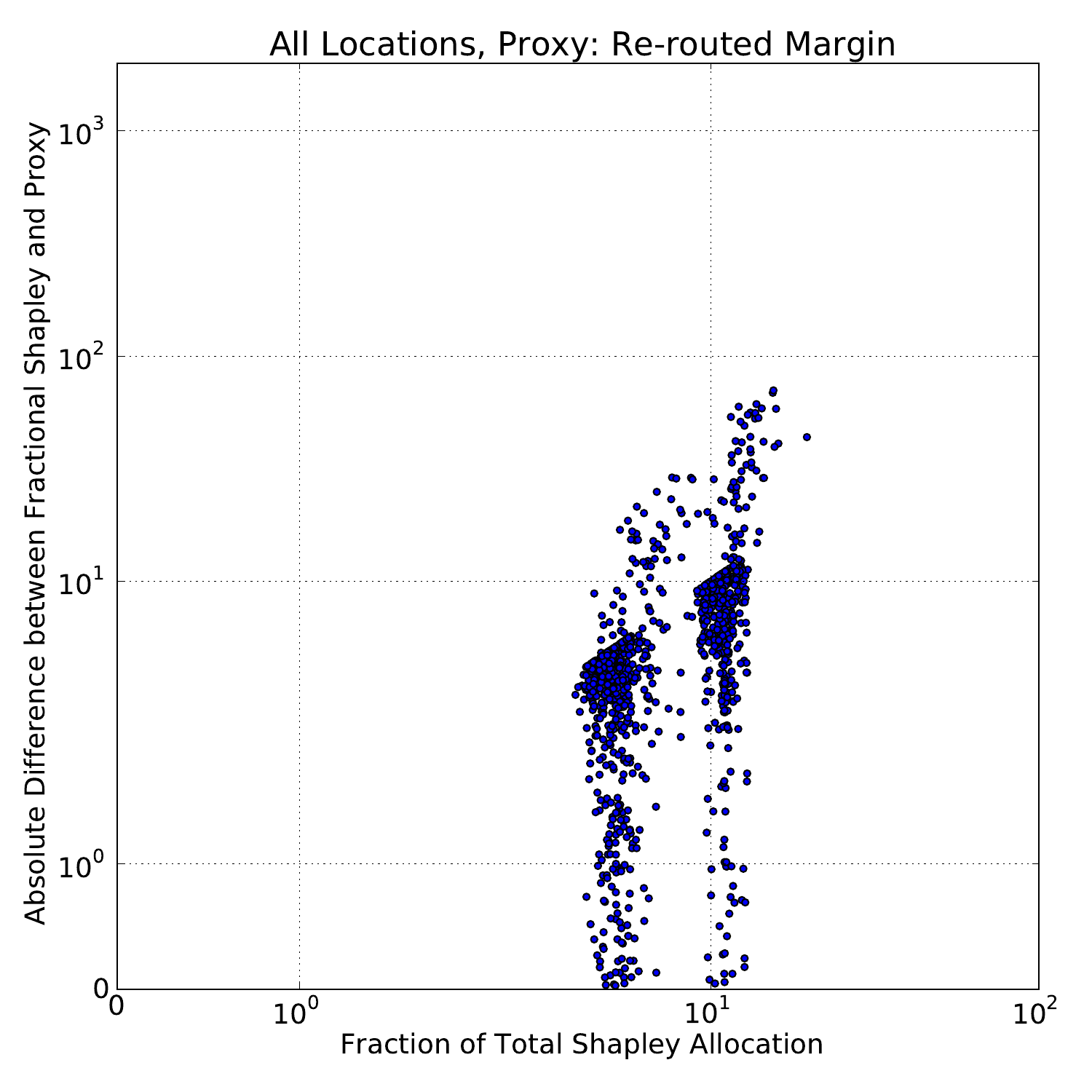}
\includegraphics[height=6.7cm, page=9]{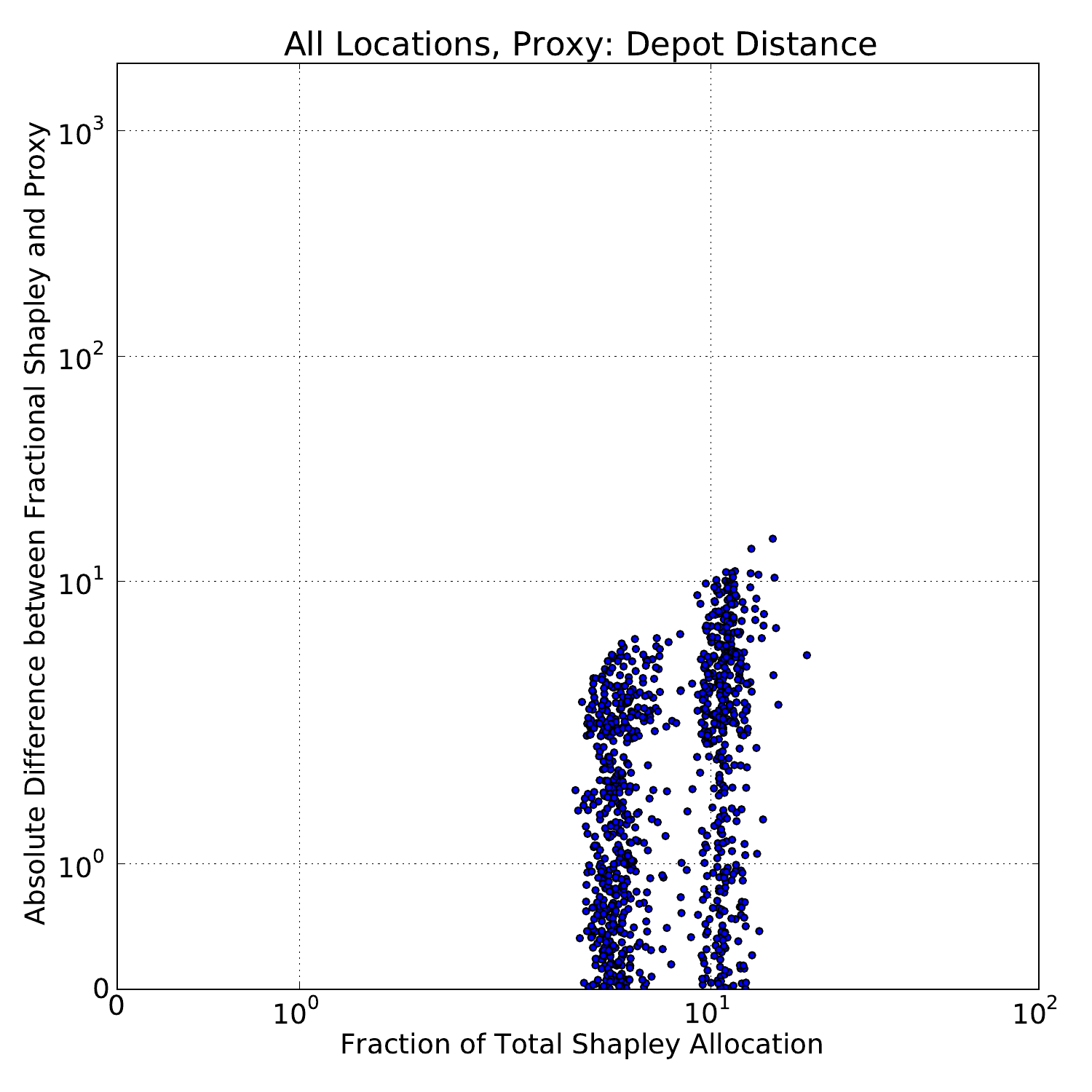}
\includegraphics[height=6.7cm,
page=15]{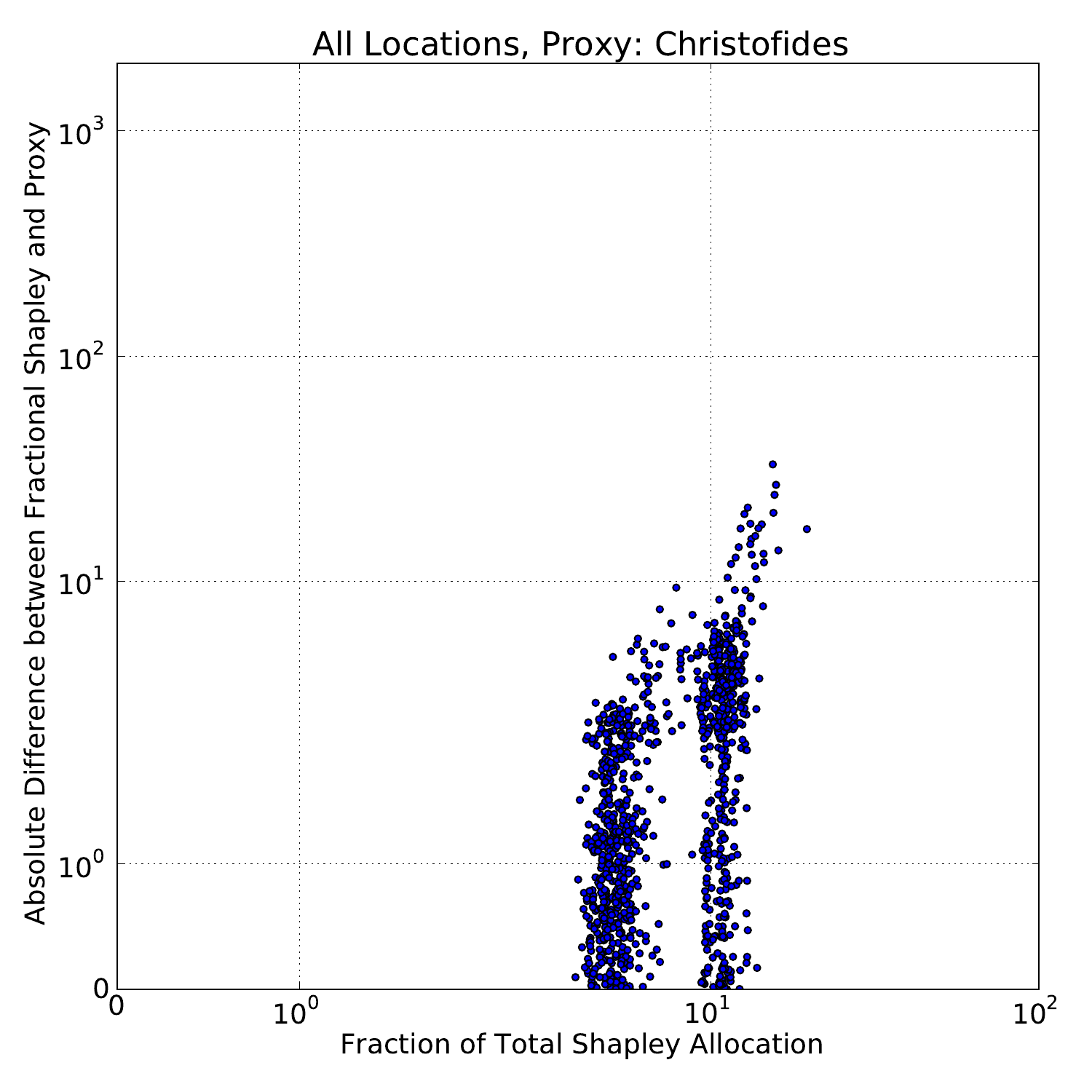} \end{minipage}
\begin{minipage}[b]{0.49\linewidth} \centering
\includegraphics[height=6.7cm, page=3]{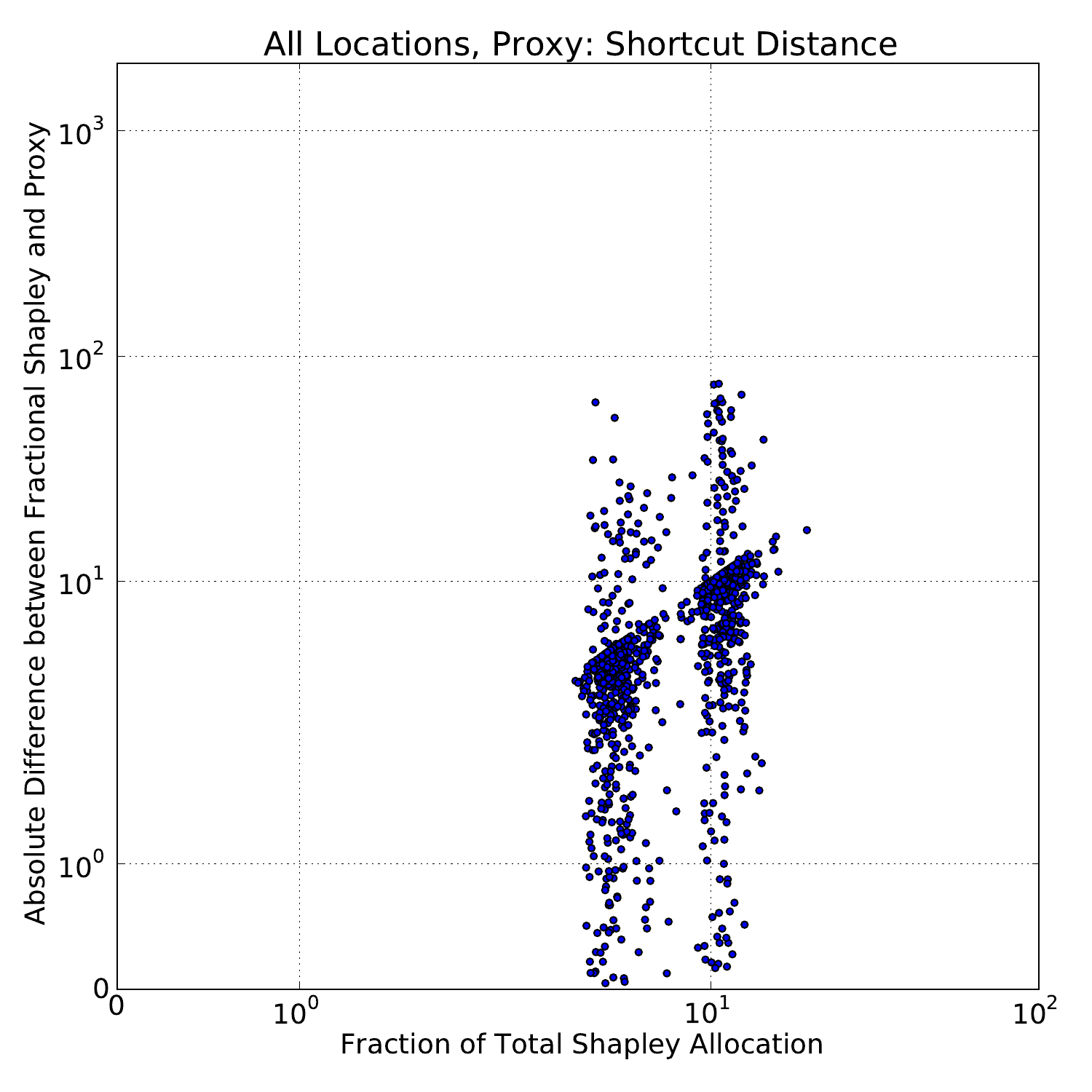}
\includegraphics[height=6.7cm,
page=12]{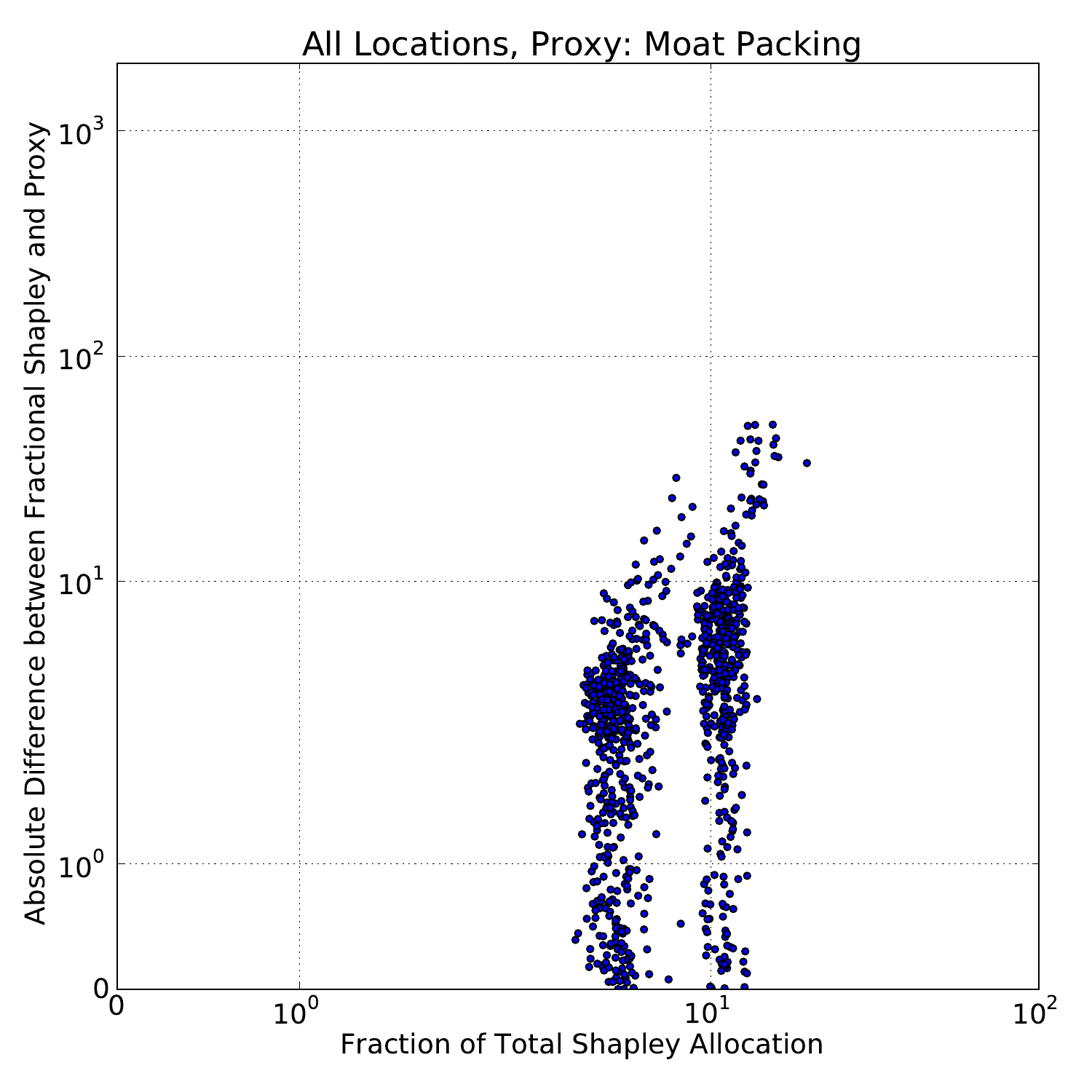}
\includegraphics[height=6.7cm,
page=18]{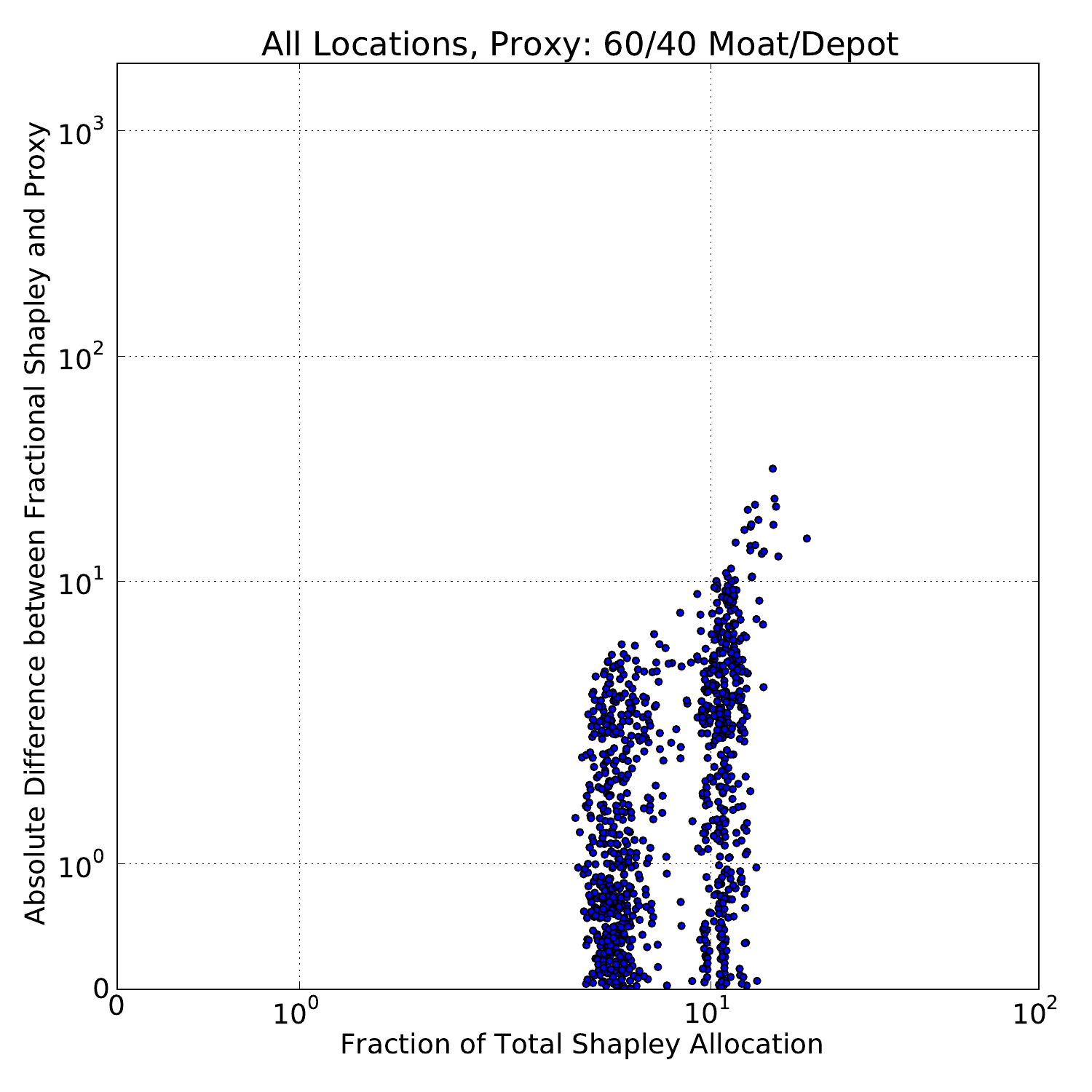} \end{minipage}
\caption{Absolute value of the difference between the $\fracshapley$ and
$\fracproxy$ plotted as a function of $\fracshapley$ for all the points
in the Sydney data for all game sizes.  Note that these are log-log
plots to highlight the spread of the data.} \end{figure} \end{document}